\documentclass[onecolumn,a4paper,11pt,accepted=2024-12-08]{quantumarticle}
\pdfoutput=1

\usepackage[utf8]{inputenc}
\usepackage{graphicx}
\usepackage{hyperref}
\usepackage{amsmath}
\usepackage{amssymb}
\usepackage{xcolor}
\usepackage{amsthm}
\usepackage{subcaption}
\usepackage{cleveref}
\usepackage{tikz}
\usepackage{quantikz}
\usepackage{amsthm}
\usepackage[numbers,sort&compress]{natbib}
\usepackage[perpage,symbol]{footmisc}

\newcommand\etal{\emph{et al.}~}

\crefname{section}{Sec.}{Secs.}
\crefname{figure}{Fig.}{Figs.}
\Crefname{section}{Section}{Sections}
\Crefname{figure}{Figure}{Figures}

\newtheorem{theorem}{Theorem}
\newtheorem{remark}{Remark}

\title{Accommodating Fabrication Defects on Floquet Codes with Minimal Hardware Requirements}
\author{Campbell McLauchlan$^{*\dagger}$}
\email{campbell.mclauchlan@gmail.com}
\affiliation{Riverlane, St Andrews House, 59 St Andrews Street, Cambridge, CB2 3BZ, United Kingdom}
\affiliation{Department of Physics and Astronomy, University College London, London, WC1E 6BT, United Kingdom}
\author{Gy\"{o}rgy P. Geh\'{e}r}
\email{gehergyuri@gmail.com, george.geher@riverlane.com}
\affiliation{Riverlane, St Andrews House, 59 St Andrews Street, Cambridge, CB2 3BZ, United Kingdom}
\author{Alexandra E. Moylett$^{*\ddagger}$}
\email{alex.moylett@gmail.com}
\affiliation{Riverlane, St Andrews House, 59 St Andrews Street, Cambridge, CB2 3BZ, United Kingdom}
\date{8th December 2024}

\begin{document}

\maketitle
\def\thefootnote{*}\footnotetext{These authors contributed equally to this work.}

\def\thefootnote{$\dagger$}\footnotetext{Current address: Centre for Engineered Quantum Systems, School of Physics, The University of Sydney, Sydney, NSW 2006, Australia}

\def\thefootnote{$\ddagger$}\footnotetext{Current address: Nu Quantum, 21 JJ Thompson Avenue, Cambridge, CB3 0FA, United Kingdom}

\renewcommand{\thefootnote}{\fnsymbol{footnote}}

\begin{abstract}
    Floquet codes are an intriguing generalisation of stabiliser and subsystem codes, which can provide good fault-tolerant characteristics while benefiting from reduced connectivity requirements in hardware. A recent question of interest has been how to run Floquet codes on devices which have defective---and therefore unusable---qubits. This is an under-studied issue of crucial importance for running such codes on realistic hardware. To address this challenge, we introduce a new method of accommodating defective qubits on a wide range of two-dimensional Floquet codes, which requires no additional connectivity in the underlying quantum hardware, no modifications to the original Floquet code's measurement schedule, can accommodate boundaries, and is optimal in terms of the number of qubits and stabilisers removed. We numerically demonstrate that, using this method, the planar honeycomb code is fault tolerant up to a fabrication defect probability of $\approx 12\%$. We find the fault-tolerant performance of this code under defect noise is competitive with that of the surface code, despite its sparser connectivity. We finally propose multiple ways this approach can be adapted to the underlying hardware, through utilising any additional connectivity available, and treating defective auxiliary qubits separately to defective data qubits. Our work therefore serves as a guide for the implementation of Floquet codes in realistic quantum hardware.
\end{abstract}

\section{Introduction}

Floquet codes are a relatively new and exciting area in quantum error correction (QEC). From a theoretical perspective, they have a mathematical structure which is distinct from other families of codes such as stabiliser and subsystem codes, which lends them some interesting and unique properties \cite{Hastings2021dynamically}. In addition, from a more practical perspective, Floquet codes benefit from considerably reduced connectivity requirements in hardware, compared to the surface or colour codes, while maintaining a reasonably high threshold under Pauli noise and a low tera-quop footprint for certain noise models~\cite{Gidney2021faulttolerant, Gidney2022benchmarkingplanar, Paetznick2023PerformancePlanarFloquetCodes}. Finally, from a computational perspective, they provide simple constructions for implementing logical operations from the Clifford group \cite{davydova2023quantum}.

While Pauli noise is the commonly studied error model in QEC, there are many other forms of noise beyond Pauli errors which appear when developing quantum devices in practice, such as qubit loss~\cite{Barrett2010MBQCErasure,Whiteside2014Loss}, erasure~\cite{erasure_codes,Sahay2023BiasedErasure}, leakage~\cite{Fowler2013Leakage,Ghosh2013Leakage,Brown2019Leakage,Miao2023OvercomingLeakage} and cosmic ray events~\cite{Vepsalainen2020IonizingRadiation,McEwen2022CosmicRays}. Fabrication defects, where an error in the fabrication of a quantum device renders some of the qubits so noisy as to be effectively useless, is a more permanent form of error which must be dealt with. These defective qubits can be identified in advance through device calibration. One instance of such an error appearing in real devices can be found in the Sycamore processor used in Google's 2019 quantum computational supremacy experiment \cite{Arute2019QuantumSupremacy}. It is worth avoiding defective components when running a QEC code, as they can introduce types of noise that are harmful to the overall performance of the code~\cite{Ghosh2013Leakage,Chen2021ExponentialSuppression}.

Several works have already shown how to accommodate fabrication defects on other QEC codes. For the surface code, code deformation is used to isolate the defects from the rest of the lattice, resulting in a code that maintains a threshold under Pauli noise at finite defect rates~\cite{Stace2010DegeneracySurfaceCodes, Nagayama2017DefectiveLattice, Auger2017FabricationSurfaceCode, Strikis2023ScalableQC,Suzuki_surface_code_defects_2022, GransSamuelsson2024PairwiseMeasurementsSurfaceCode,debroy2024lucisurfacecodedropouts}. As for the case of Floquet codes, Aasen \etal \cite{aasen2023faulttolerant} have initiated the study by proposing methods for accommodating fabrication defects on the 4.8.8~\cite{Paetznick2023PerformancePlanarFloquetCodes} and honeycomb~\cite{Hastings2021dynamically} Floquet codes. This leaves room for substantial further work in finding improvements to these methods and benchmarking Floquet code performance in defective hardware, to help determine how promising a candidate for large-scale QEC this type of code is. In particular, it is crucial to consider methods that make no additional assumptions about the device connectivity, since low connectivity requirements are one of the main benefits of Floquet codes. Moreover, the ability to generalise strategies to a large array of Floquet codes would be very useful.

In this paper we present a new way of accommodating fabrication defects on Floquet codes that respects hardware constraints and is highly generalisable. We present our method as an iteratively applicable algorithm which provides a new Floquet code given a set of defective qubits. It removes at most two qubits per defect---which is optimal when preserving important properties of the lattice---requires no additional connectivity in the hardware, and can accommodate boundaries. We also require no modifications to the original measurement schedule of the Floquet code, which ultimately means that our method is suitable for use in a wide range of Floquet codes (specifically, all such codes definable on a colour code lattice~\cite{Hastings2021dynamically,Paetznick2023PerformancePlanarFloquetCodes,Kesselring2022AnyonCondensation,Davydova2023CSSFloquet}).

We numerically benchmark the performance of our method by implementing circuits for a planar honeycomb code memory experiment (specifically considering the model of Gidney \etal \cite{Gidney2022benchmarkingplanar}) in Stim \cite{Gidney2021Stim}, a fast Clifford circuit simulator, and correcting for Pauli errors on the adapted code using PyMatching, a matching decoder \cite{Higgott2023PyMatching}. We find considerable robustness to fabrication defects using our method: we determine a defect threshold of approximately $13\%$, and numerically demonstrate that Pauli error thresholds exist for defect probabilities close to this threshold. Hence, the region of parameter space in which the code remains fault-tolerant is large and comparable to that of the surface code~\cite{Stace2010DegeneracySurfaceCodes}. The reduction in Pauli error threshold with defect probability is also quantitatively similar to previous findings for the surface code~\cite{Auger2017FabricationSurfaceCode}. This is quite surprising, as the honeycomb code achieves this with more restrictive constraints on the connectivity of the device.

We further introduce two methods which take advantage of the underlying hardware. First, we propose even better methods for accommodating defects in hardware in which extra connectivity is present, such as on square lattices \cite{Krinner2022SurfaceCodeSuperconducting, GoogleSuppressing, RigettiAnkaa}. And secondly, we propose a method for better accommodating defective auxiliary qubits in lattices in which those are present---or, more generally, defective connections between data qubits---such as the heavy hex lattice \cite{IBMHeavyHex}.

The rest of this paper is laid out as follows. In \cref{sec:background} we provide some background on the honeycomb code (as a prototypical example from the family of Floquet codes we consider) and fabrication defects. In \cref{sec:strategy} we describe our strategy for removing defective qubits from the code lattice. Results from simulations are described in \cref{sec:simulation-results}. We propose some improvements that can be made depending on the hardware of the quantum processor in \cref{sec:hardware-improvements}. Finally, we conclude with some future research directions in \cref{sec:conclusion}.

\section{Background}
\label{sec:background}

A Floquet code is defined by a series of check measurements and a sequence in which they are measured\footnote{For our purposes, we do not distinguish between finite or infinite-period dynamical codes, as other authors do~\cite{Townsend_Teague_2023,davydova2023quantum}; they will both be, in principle, amenable to our techniques.}. Many Floquet codes are definable on the lattice of a colour code~\cite{Kesselring2022AnyonCondensation}. We focus here on the honeycomb code~\cite{Hastings2021dynamically,Gidney2021faulttolerant}, however this example can be easily generalised.

\subsection{Honeycomb code}

The honeycomb code is a Floquet code defined on a honeycomb lattice\footnote{This can be generalised to any lattice that is 3-face-colourable and tri-valent, i.e., where each vertex is adjacent to three edges.}. Each vertex represents a physical qubit, and each edge represents a two-qubit Pauli measurement between neighbouring qubits. An example lattice is presented in \cref{fig:example-honeycomb}, with the plaquettes coloured red, blue and green in such a way that no neighbouring plaquettes share the same colour (the lattice is \emph{3-face-colourable}). Edges are also assigned colours according to the plaquettes they connect; for example a red edge connects a pair of red plaquettes. A key benefit that this code has over other error-correcting codes, such as the surface code, is the low connectivity requirements of the code: each qubit only needs to be connected to three nearest neighbours (the lattice is \emph{tri-valent}), and each measurement is only between pairs of adjacent qubits. This makes Floquet codes appealing for architectures with low connectivity, such as IBM's heavy hex lattice \cite{IBMHeavyHex}, as well as devices with native two-body measurements such as Majorana devices~\cite{Karzig2017MajoranaZeroModes, Paetznick2023PerformancePlanarFloquetCodes}.

\begin{figure}
    \centering
    \includegraphics[width=0.99\linewidth]{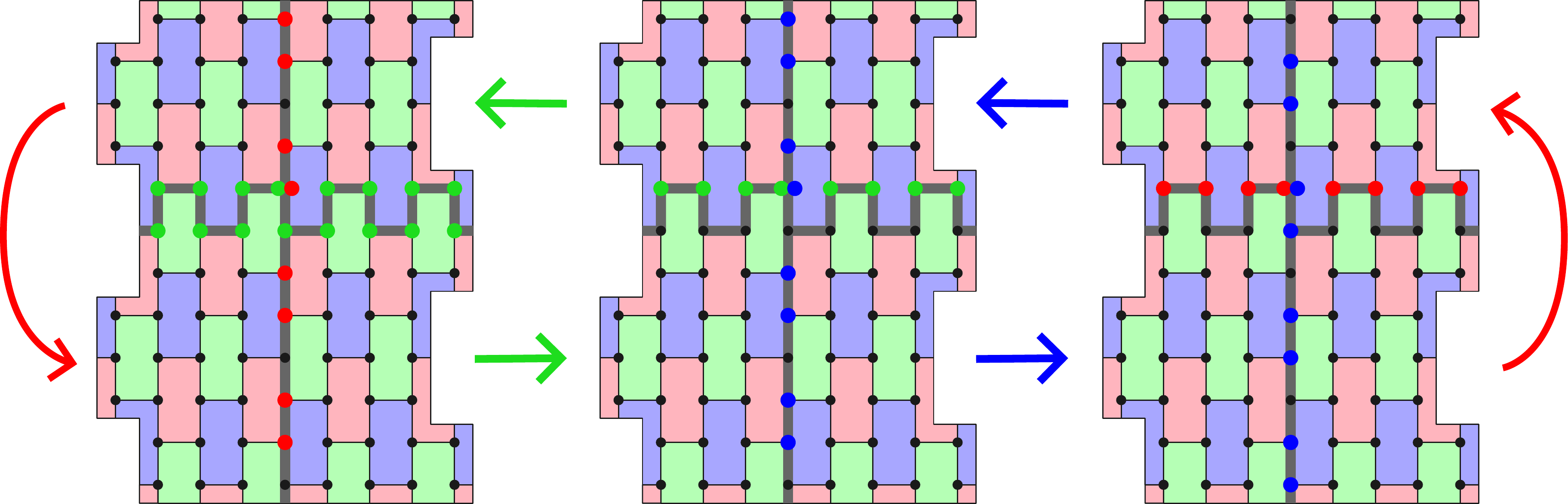}
    \caption{\label{fig:example-honeycomb} An example of the planar honeycomb code, with qubits on vertices and with plaquettes coloured according to their Pauli type (red for $X$, green for $Y$ and blue for $Z$). In the periodic honeycomb code, the horizontal and vertical boundaries wrap around on a torus. In the planar case shown, boundary ``edges'' involve single-qubit measurements \cite{Gidney2022benchmarkingplanar}. These form plaquettes that have non-deterministic outcomes at certain points in the measurement schedule. Logical strings at each point in the period-six measurement schedule are shown. Vertices coloured red/green/blue indicate where logicals have support as Pauli $X$/$Y$/$Z$ operators respectively. Arrows indicate the checks measured at that point in the sequence (red for $X$ checks, green for $Y$ checks and blue for $Z$ checks). Logical strings are only updated after $Y$ and $Z$ check measurements.}
\end{figure}

A QEC round using the periodic honeycomb code works by implementing two-body Pauli measurements between adjacent qubits according to the following measurement schedule: first, $XX$ measurements are applied to red edges, then $YY$ measurements are applied to green edges, and finally $ZZ$ measurements are applied to blue edges~\cite{Gidney2021faulttolerant}. If periodic boundary conditions are used, we then repeat the measurement schedule. Floquet codes may not store logical operators, or have stabiliser groups, when viewed as subsystem (i.e., static) codes~\cite{Hastings2021dynamically,Kesselring2022AnyonCondensation}, however this does not preclude them from having these properties when viewed dynamically. We can understand this through the instantaneous stabiliser groups, which depend on the measurement sub-round (i.e., the edge type just measured). For the honeycomb code, the instantaneous stabiliser group after each sub-round of QEC consists of six-body Pauli operators on each plaquette (these commute with all edge measurements), plus the two-body Pauli operators measured in the last sub-round. Each plaquette of a given colour is made up of the product of edges encircling that plaquette, which will be of the other two colours. For example, a red plaquette is made up of three green edges and three blue edges that encircle the plaquette. The $X^{\otimes 6}$ stabiliser that this red plaquette corresponds to is the product of these $YY$ and $ZZ$ edge operators, and the plaquette's measurement outcome can therefore be inferred from the two previous sub-rounds in which green and blue edges were measured, respectively. This process can also be followed for the green $Y^{\otimes 6}$ and blue $Z^{\otimes 6}$ plaquettes.

The logical operators of the honeycomb code have to change in each sub-round in order to ensure that they commute with measurements in future QEC sub-rounds. Specifically, some edge measurements on horizontal and vertical cycles across the code need to be multiplied in to the logical observables \cite{Hastings2021dynamically, Gidney2021faulttolerant}. The logical operators change each sub-round, repeating with period six. The measurement schedule is crucial for ensuring that the stabilisers can be measured while at the same time ensuring that these observables are preserved; if one were to instead implement measurements such that all stabilisers (of the subsystem code) were measured simultaneously, a code with no logical qubits would be formed \cite{Hastings2021dynamically}.

The planar honeycomb code functions in a similar way. Boundaries can be formed without increasing the connectivity requirements by cutting along edges of the same colour of a larger honeycomb lattice \cite{Gidney2022benchmarkingplanar}. The cut edges correspond to single-qubit Pauli measurements in the same basis as before. Note that if we use the same measurement schedule as before, the result will be that some measurements will anti-commute with the observable. Instead, the schedule is modified to measure edges in an $R \to G \to B \to R \to B \to G$ pattern~\cite{Haah2022boundarieshoneycomb,Dua_2024_rewinding}. The transformations of the logical operators through this measurement schedule are shown in \cref{fig:example-honeycomb}. Note that logical operators are not updated after red edges are measured.

The single-qubit measurements anti-commute with plaquettes along the boundary. The result of this is that some stabilisers are non-deterministic at times and therefore cannot detect errors at those times. However, this does not seem to significantly affect the code's performance when compared to the periodic honeycomb code~\cite{Gidney2022benchmarkingplanar}.

We form detectors---sets of measurements whose joint parities are deterministic when no errors have occurred---from two measurements of plaquette operators. For example, after the sequence $B \to R \to B$, we measure all bulk green plaquettes twice (once after $B \to R$, and a second time after $R \to B$); therefore detectors can be formed for each green plaquette by comparing the two measured values. An $X$ or $Z$ error occurring on a qubit in the plaquette's support between the first and last of these three measurement sub-rounds will trigger the detector. Meanwhile, measurement results from detectors are deterministic in the presence of zero noise. A decoder can use these results to identify and correct errors when noise is present. Boundary plaquettes are measured non-deterministically when the single-qubit edge measurements precede the two-qubit edge measurements. However, they are re-measured deterministically (assuming no error has occurred) when the two-qubit edge measurements precede the single-qubit measurements. For example, a green boundary plaquette in \cref{fig:example-honeycomb} is first measured non-deterministically when the blue edge measurements immediately precede the red edge measurements ($B\to R$). It is then measured deterministically for a second time after the sequence $R\to B$, at which point we can form a detector comparing these two measurement outcomes.

\subsection{Fabrication defects}

In this work we address the challenge of implementing Floquet codes when there are fabrication defects on the hardware device. A fabrication defect is a qubit which at the time of running the computation is known to be highly prone to errors, or otherwise unusable. Such faults are common in quantum processors: Google's 54-qubit Sycamore processor for example featured one qubit which did not function properly and therefore could not be used in their 2019 quantum computational supremacy experiment \cite{Arute2019QuantumSupremacy}.

Fabrication defects have been studied extensively in the context of surface codes, and can be accommodated by merging adjacent stabilisers together into ``super-stabilisers'' in order to remove the defective qubit(s) \cite{Stace2010DegeneracySurfaceCodes, Nagayama2017DefectiveLattice, Auger2017FabricationSurfaceCode, Strikis2023ScalableQC, GransSamuelsson2024PairwiseMeasurementsSurfaceCode,debroy2024lucisurfacecodedropouts}. These super-stabilisers are measured by introducing ``gauge operators'' around a defect, such that the super-stabiliser measurements can be formed by taking products of gauge measurements. This is similar to how edge measurements combine to form plaquette measurements in the honeycomb code, with the difference being that there is no dependency on the order in which gauge measurements are performed. However, translating this procedure to the honeycomb code is non-trivial: the adaptation of the lattice and the measurement schedule of the Floquet code must be performed carefully to make sure that detectors can still be formed near the defect locations. Meanwhile, although each sub-round of the honeycomb code can be mapped to an instance of a surface code \cite{Hastings2021dynamically}, there are several impediments and disadvantages to directly adapting the surface code strategy, including that gauge operators of the surface code would not always correspond to edge measurements in the honeycomb code, and that the surface code mapped to from the honeycomb code evolves with each sub-round. Meanwhile, the resulting honeycomb code strategy may not be optimal, and the applicability of this strategy beyond the honeycomb code may also be limited. Hence, we must go beyond a simple adaptation of existing strategies to accommodate fabrication defects in Floquet codes.

Aasen \etal have investigated ways of accommodating fabrication defects on honeycomb and 4.8.8 Floquet codes~\cite{aasen2023faulttolerant}. They propose several strategies, two of which involve removing all qubits in a plaquette adjacent to a defect. An important consideration for such methods is whether the device has extra connectivity available for the formation of new plaquettes, which is required for some methods proposed~\cite{aasen2023faulttolerant}. Respecting the constraints of hardware while removing the smallest number of qubits and edges from the lattice possible is important for maintaining a high percolation threshold, and presumably a high Pauli noise threshold under finite fabrication defect rates\footnote{One of the methods proposed in Ref.~\cite{aasen2023faulttolerant}, while requiring extra connectivity, only necessitates the removal of a single qubit, which could be useful for honeycomb code implementations. However, the properties of the resulting graph and the updated measurement schedule make this strategy potentially hard to generalise to other Floquet codes (including the honeycomb code with boundaries) or to multiple defects in the lattice.}.

\section{Strategy for removing defective qubits from the lattice}
\label{sec:strategy}

Here we outline the strategy for removing defective qubits from the honeycomb code's underlying lattice. The main idea is to create ``super-plaquettes" which avoid using the defective qubits, thereby removing them from the lattice. 
We do this by merging some plaquettes in the original lattice together, while other existing plaquettes are replaced by two-body plaquettes.
We will do all this while retaining the 3-face-colourability and tri-valence of the graph, for ease of defining the updated code and ease of generalisation to other Floquet codes. To achieve this, we need to remove a pair of qubits along an edge, as described below. There is some subtlety around choosing the extra removed qubit (or qubits if there is more than one defect) in such a way that maximises the distance of the resulting code---we shall refer to this as the ``edge choice'', since choosing a qubit to remove is the same as choosing an edge adjacent to the given fabrication defect. We also do not allow for extra connectivity between qubits in the updated Floquet code. In \Cref{app:extra_connectivity}, we describe how improved results can be obtained in the presence of extra connections in the underlying hardware. 

Below, we will outline the algorithm for generating the new code with super-plaquettes given a pattern of fabrication defects and an edge choice for each defect. There are certain situations in which the edge choice is obvious, but otherwise we need to use heuristics to make an edge choice. We discuss some heuristics for choosing suitable edges to remove in \cref{ssec:choosing-defect-edge}.

\subsection{The super-plaquette algorithm}
\label{ssec:superplaquette-algorithm}

We begin with a single fabrication defect in the bulk of the lattice and explain how the algorithm runs for this case. Let us denote an edge as $(E,B)$ where $E$ is a set of one or two qubits and $B$ is a Pauli basis. A plaquette will be a set of edges, each of which has basis not equal to the plaquette basis. The measurement outcome of this plaquette will be the product of measured values of all edge operators---the measurement of the qubit(s) in the basis specified---in its corresponding set. As mentioned above, the first step of the algorithm is to make an edge choice, giving us another qubit to remove from the lattice. We will call this selected edge the ``defect edge'', although note that there is not necessarily anything defective about this edge. Hence we begin with the following step:

\begin{enumerate}
    \item Choose a defect edge.
\end{enumerate}
\noindent
We call the basis of this defect edge the ``defect edge basis” ($D$).
An example is illustrated in \cref{fig:select-plaquettes}. The defective qubit is the red vertex and the other removed qubit is the blue vertex---the defect edge is the edge in orange connecting these two qubits. Note that this edge connects two red $X^{\otimes 6}$ plaquettes, so the defect edge basis is $D = X$.

We then proceed through the list of plaquettes and find all plaquettes that contain the defect edge---we will refer to these as the ``shrink plaquettes.'' For a single defect edge there will only be two such plaquettes and they will be of different types.

\begin{enumerate}
    \setcounter{enumi}{1}
    \item Find the shrink plaquettes adjacent to the defect edge.
\end{enumerate}
\noindent
In \cref{fig:select-plaquettes}, there is one green and one blue shrink plaquette incident to the defect edge. The green (resp. blue) shrink plaquette in this example corresponds to a $Y^{\otimes 6}$ (resp. $Z^{\otimes 6}$) stabiliser.

We then find all plaquettes of type $D$ that are adjacent to a shrink plaquette---we call these the ``merge plaquettes''. 

\begin{enumerate}
    \setcounter{enumi}{2}
    \item Find the merge plaquettes.
\end{enumerate}
\noindent
In \cref{fig:select-plaquettes}, there are four red merge plaquettes. In this example, the merge plaquettes correspond to $X^{\otimes 6}$ stabilisers.

\begin{figure}
    \centering
    \begin{subfigure}{0.3\linewidth}
        \includegraphics[width=\linewidth]{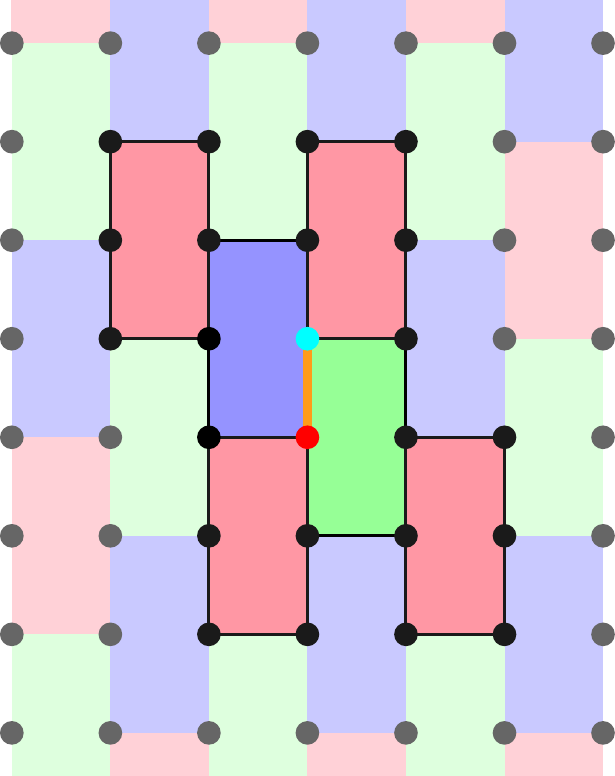}
        \caption{\label{fig:select-plaquettes}}
    \end{subfigure}\hfill
    \begin{subfigure}{0.3\linewidth}
        \includegraphics[width=\linewidth]{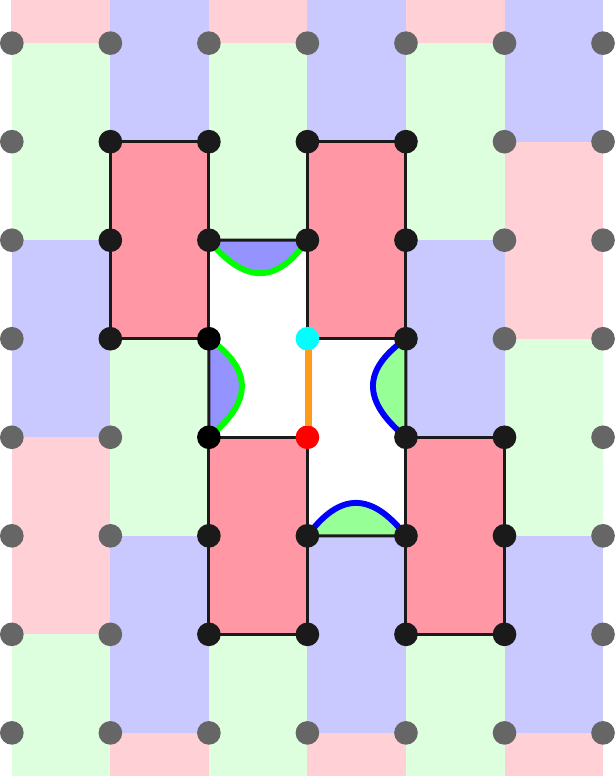}
        \caption{\label{fig:shrink-plaquettes}}
    \end{subfigure}\hfill
    \begin{subfigure}{0.3\linewidth}
        \includegraphics[width=\linewidth]{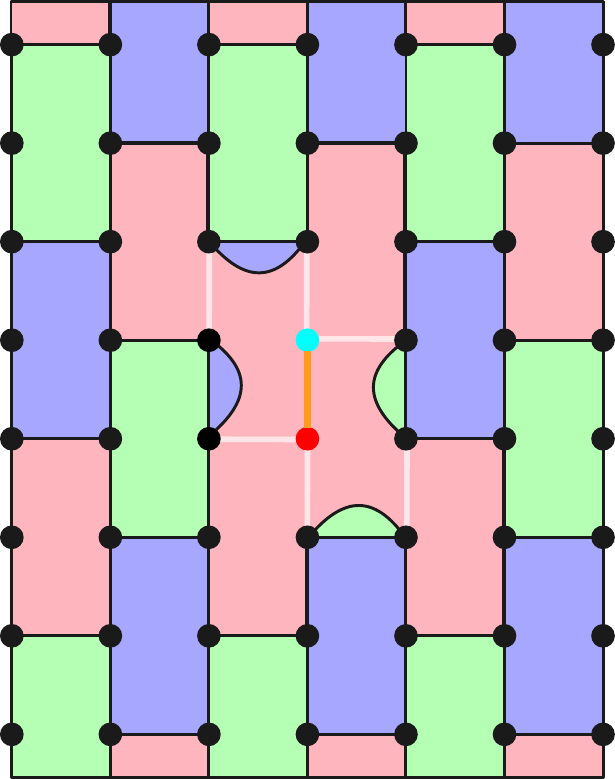}
        \caption{\label{fig:defect-removed}}
    \end{subfigure}

    \caption{\label{fig:example-defect} Removing a defective qubit from the bulk of the honeycomb code. (a) A defective qubit is highlighted in red, along with a neighbouring qubit in blue and the defect edge between them in orange. In this example, one green plaquette and one blue plaquette around the defect are selected as shrink plaquettes, and four surrounding red plaquettes are selected as merge plaquettes. These plaquettes have been highlighted. (b) The shrink plaquettes are replaced with weight-2 plaquettes, formed using newly-added $YY$ and $ZZ$ edges highlighted in green and blue, respectively. (c) The merge plaquettes are combined into a single super-plaquette. The super-stabiliser is supported on $22$ qubits, and it is composed of $9$ blue and $9$ green edges that were originally involved in one of the merge plaquettes, plus $2$ blue edges and $2$ green edges which were added in Step 4. Removed edges are shown in white, the defective qubit and its pair on the defect edge no longer contribute to any stabilisers and can be removed.}
\end{figure}

Now we form the super-plaquette and adjacent two-body plaquettes. For each shrink plaquette $S$, let $E_{S, D}$ be the set of edges in $S$ which are the same basis as the defect edge---in this case, the $XX$ edges---but do not contain the defect. We first add the two-body plaquettes to the list of plaquettes:

\begin{enumerate}
    \setcounter{enumi}{3}
    \item For each shrink plaquette $S$ (suppose it has Pauli type $P$) and for each edge $(E, D) \in E_{S, D}$, add a plaquette that involves two edges: $(E, D)$ and $(E,T)$, where $T \notin \{D, P\}$. Add edge $(E,T)$ to the list of edges and to a set of newly-added edges $C$.
    \item Delete all shrink plaquettes from the list of plaquettes.
\end{enumerate}
\noindent
In \cref{fig:shrink-plaquettes}, we show how Step 4 is applied to the example in \cref{fig:select-plaquettes}, with the new $YY$ and $ZZ$ edges in $C$ coloured green and blue, respectively. The green (resp. blue) weight-6 plaquettes have now been split into two weight-2 green (resp. blue) plaquettes.

Note that the set of newly-added edges $C$ connect the merge plaquettes in a cycle. We use $C$ to combine the merge plaquettes into a single new super-plaquette:

\begin{enumerate}
    \setcounter{enumi}{5}
    \item Include all edges in $C$ in the super-plaquette. For every merge plaquette $M$ and every edge $(E,B)$ in $M$, add $(E,B)$ to the super-plaquette if $(E,B)$ is not an element of any shrink plaquette.
    \item Delete all merge plaquettes from the list of plaquettes.
\end{enumerate}

The final result of this process is shown in \cref{fig:defect-removed}. Note that the defective qubit and its neighbour are no longer connected to any other qubits, nor do they contribute to any plaquettes in the code. This means that they can now be removed from the code.

The resulting super-plaquette is of the same type as the merge plaquettes: it is formed from a cycle of edges of alternating type not equal to $D$ (alternating $YY$ and $ZZ$ measurements for the example in \cref{fig:example-defect}). This follows from the fact that the cycle of edges that wraps around the super-plaquette is formed of edges from the set $C$ and edges from merge plaquettes, none of which have basis equal to the merge plaquette type. To see that the cycle of bases is alternating, consider the edges adjacent to $(E,B)\in C$ in \cref{fig:shrink-plaquettes}. There are five edges adjacent to $E$. One edge is $(E, D)$, the edge connecting these two vertices in the original Floquet code. This edge is of type $D$---the defect edge basis---and will not be included in the super-plaquette. Of the remaining four, two are of type $B$; these two will contribute to one of the shrink plaquettes, and will therefore be removed from the lattice during Steps 6 and 7. This leaves two edges incident to edge $(E,B)$, which by construction have to be a different type to $B$ or $D$, and these are the two edges that are included in the super-plaquette. Therefore the super-plaquette alternates before and after each edge in $C$. Furthermore, the cycle continues to alternate around the merge plaquette, by construction of the original code. In this way we see that the entire super-plaquette is formed from alternating edges of two bases. 

\subsubsection{Remarks on super-plaquette algorithm}

It can be seen that the algorithm introduced above results in an updated lattice with the required properties of tri-valence and 3-face-colourability. It does so while (at least for a defect within the bulk of the lattice) removing the smallest number of qubits, edges and plaquettes from the lattice possible. In \Cref{app:minimal-qubit-removal}, we prove the following:
\begin{theorem}
    The algorithm above produces a tri-valent, 3-face-colourable lattice. It does so while removing the minimal number of qubits, edges and plaquettes from the lattice (assuming a lattice without boundaries).
    \label{thm:validity-qubit-removal}
\end{theorem}

We also note that while there are new edges added to the graph, these edges are between vertices which were adjacent in the original graph. This means that all the new two-body measurements are performed between qubits which are already connected on the quantum processor, implying that no additional hardware connectivity is required. This algorithm can also be run using the same measurement schedule as the original code, as the new measurements and plaquettes follow the same rules of tri-valence and 3-face-colourability as the original Floquet code. Hence, assuming the need to create a static lattice with the same properties as the original, and which is a sub-lattice of the original, our algorithm results in optimal qubit, edge and plaquette removal.

While the example above is specific to the honeycomb code, we can see that the algorithm is generalisable to other two-dimensional Floquet codes. In \Cref{app:alternative-codes}, we discuss how this algorithm can be extended to Floquet codes on other tri-valent and 3-face-colourable graphs, and provide an explicit example of applying this strategy to the 4.8.8 Floquet code \cite{Paetznick2023PerformancePlanarFloquetCodes, aasen2023faulttolerant}.

An important consideration for maintaining a high threshold under fabrication defect noise is the number of \emph{connections} removed from the lattice: that is, some qubits are connected by edges in the original lattice but not in the updated lattice. These removed connections are shown in white in \cref{fig:defect-removed}. For uniform lattices (i.e., honeycomb, 4.8.8 or 4.6.12), the super-plaquette algorithm (again, applied in the bulk), results in the minimal removal of such connections without requiring additional connectivity in hardware. In \Cref{app:minimal-qubit-removal}, we prove the following:
\begin{theorem}
    Consider a defective qubit in the bulk of a Floquet code defined on a honeycomb, 4.8.8 or 4.6.12 lattice (without boundary); i.e., a uniform tiling of the plane. Under the constraint that we cannot form an edge between any two qubits not connected by an edge in the original lattice, the algorithm above (with a certain choice of defect edge), minimises the number of removed edges from the original lattice.
    \label{thm:minimal-edge-removal}
\end{theorem}

There are a few further remarks we can note about this strategy. First of all, we can verify that no additional logical degrees of freedom have been created. To see this, note that two physical qubits are removed for the defect. After applying the strategy in the example of \cref{fig:example-defect}, four merge plaquettes are combined into a single super-plaquette, but two shrink plaquettes are split into two weight-2 plaquettes, meaning that the number of plaquettes has only decreased by one. We have also removed seven edges (one red, three blue, three green) and added four new edges (two blue, two green). Overall this means that we have removed one edge of each type from the graph, so after each sub-round the number of stabilisers formed by the edges of that sub-round also decreases by one. As a result, the size of the instantaneous stabiliser group after each sub-round has decreased by two generators, meaning that the number of logical qubits---the number of physical qubits minus the size of the instantaneous stabiliser group during each sub-round---remains unchanged.

We can also see that for multiple defects this algorithm can be implemented iteratively, by choosing a defect edge for each defective qubit. If multiple defects are close to each other, this can mean the super-plaquettes are themselves merged together to create even larger plaquettes. An example of two super-plaquettes merged together is shown in \cref{fig:double-defect}. Note that in this example the two defect edges selected are in the same basis, however this does not need to be true in general, and in certain cases it can be beneficial to choose different bases (see Appendix~\ref{app:edge-selection}).

\begin{figure}
    \centering
    \begin{subfigure}{0.3\linewidth}
        \includegraphics[width=\linewidth]{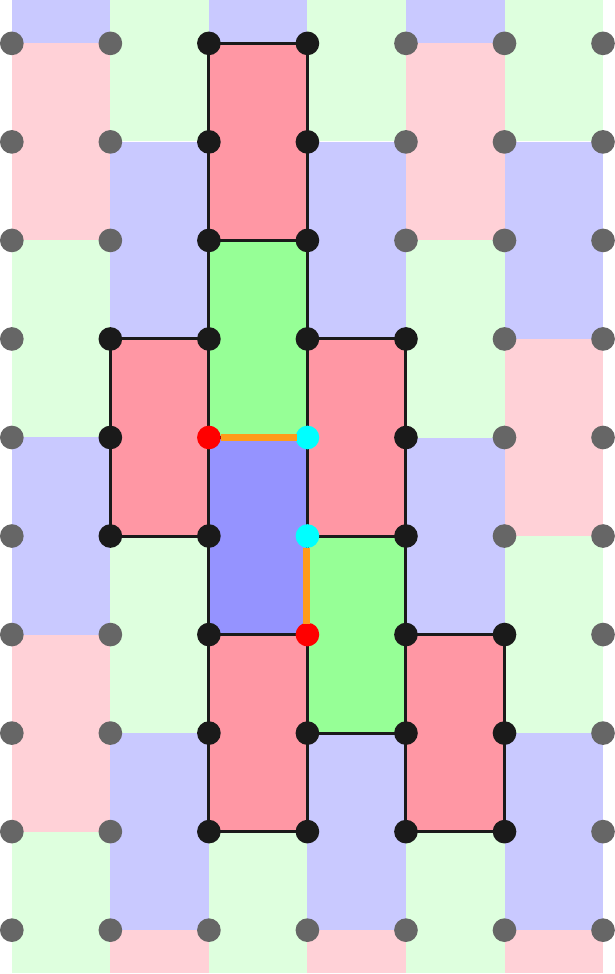}
        \caption{\label{fig:double-defect-initial}}
    \end{subfigure}\hfill
    \begin{subfigure}{0.3\linewidth}
        \includegraphics[width=\linewidth]{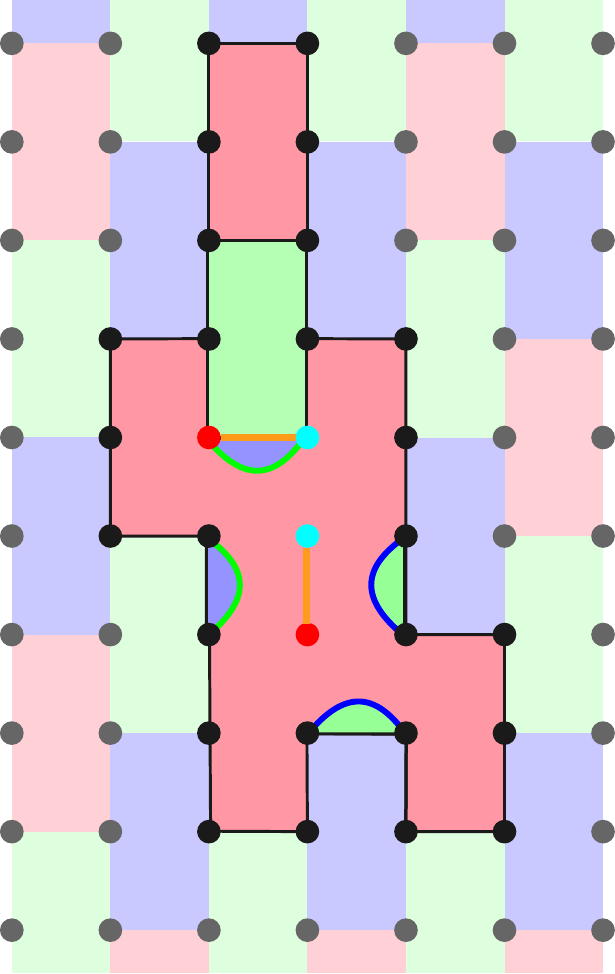}
        \caption{\label{fig:double-defect-first-removed}}
    \end{subfigure}\hfill
    \begin{subfigure}{0.3\linewidth}
            \includegraphics[width=\linewidth]{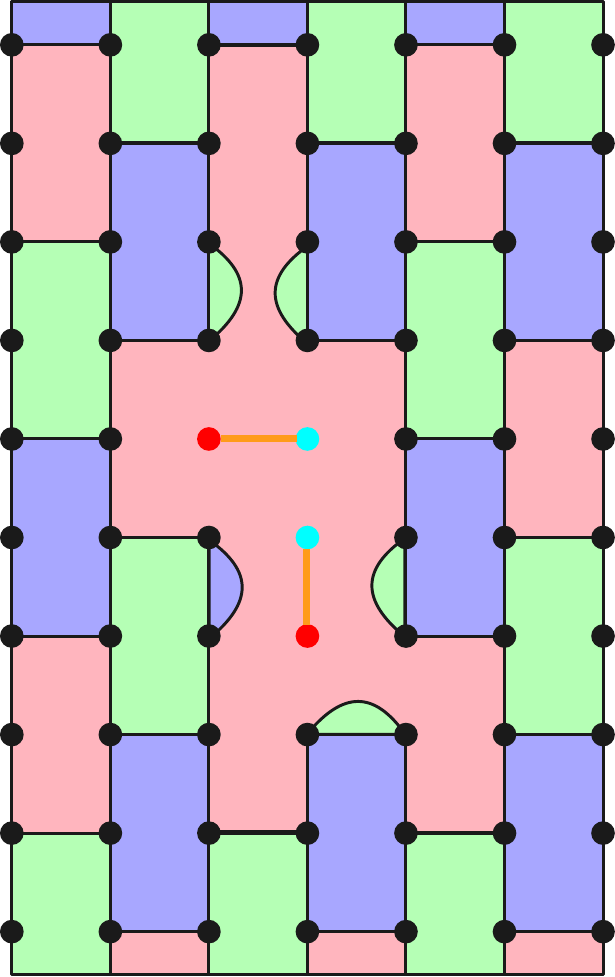}
        \caption{\label{fig:double-defect-both-removed}}
    \end{subfigure}
    \caption{\label{fig:double-defect} Removing two defective qubits from the bulk of the honeycomb code. (a) Defective qubits are in red, defect edges in orange and the neighbouring qubits to remove in blue. Merge and shrink plaquettes are highlighted for both defects, with shrink plaquettes in blue \& green and merge plaquettes in red. (b) The modified lattice after the bottom defect edge is removed, cf. \cref{fig:example-defect}. (c) The modified lattice after both defect edges are removed.}
\end{figure}

It is interesting to note the similarities between the resulting weight-two stabilisers shown in \cref{fig:defect-removed} and the two-gon method for applying boundaries to the honeycomb code presented in Refs.~\cite{Haah2022boundarieshoneycomb, Paetznick2023PerformancePlanarFloquetCodes}. Hastings and Haah showed that implementing a boundary in this way provides a non-trivial cycle of two types of edge, meaning that the measurement schedule needs to be modified in order to avoid measuring the logical observable \cite{Hastings2021dynamically, Haah2022boundarieshoneycomb}. Here, however, we \emph{do} want to measure the cycle generated by this algorithm, in order to measure a super-stabiliser around the defective qubit(s). Note also that these weight-two plaquettes can be measured directly, in addition to the way they are measured in the Floquet code (through two rounds of edge check measurements). Doing so would require an extra idling layer for the remaining qubits in the code, which would likely reduce the overall logical performance, at least for low defect densities. However, experimenting with this could result in useful improvements, and we leave this to future work.

\subsection{Defective qubits on the boundary}
\label{ssec:boundary-defects}

So far we have only considered fabrication defects which appear in the bulk of the code. Whilst this will work for periodic boundary conditions, we also need to consider what happens to defects which lie on the boundary of the code in order for this technique to be practical to implement on current quantum hardware.

In \cite{Gidney2022benchmarkingplanar}, Gidney \etal showed how boundaries for the honeycomb code can be formed by cutting a planar patch out of a larger honeycomb lattice. We can also use this technique to handle fabrication defects on the planar honeycomb code: first, we embed the patch back into the larger lattice; second, we apply the algorithm described in \cref{ssec:superplaquette-algorithm} to modify the honeycomb code; and third, we cut the patch out along the same boundaries again. We give some examples of applying this method to defects on the horizontal and vertical boundaries of the code in \cref{fig:defect-boundary}. Just like the periodic boundary conditions case, this approach requires no extra connectivity and can be implemented using the same measurement schedule as the defect-free planar honeycomb code.

\begin{figure}
    \centering
    \begin{subfigure}{0.3\linewidth}
        \includegraphics[width=\linewidth]{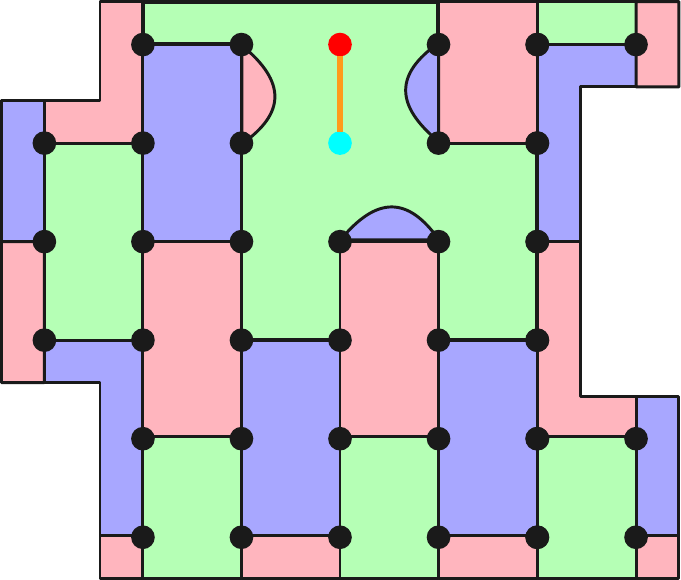}
        \caption{\label{fig:defect-boundary-good-edge}}
    \end{subfigure}\hfill
    \begin{subfigure}{0.3\linewidth}
        \includegraphics[width=\linewidth]{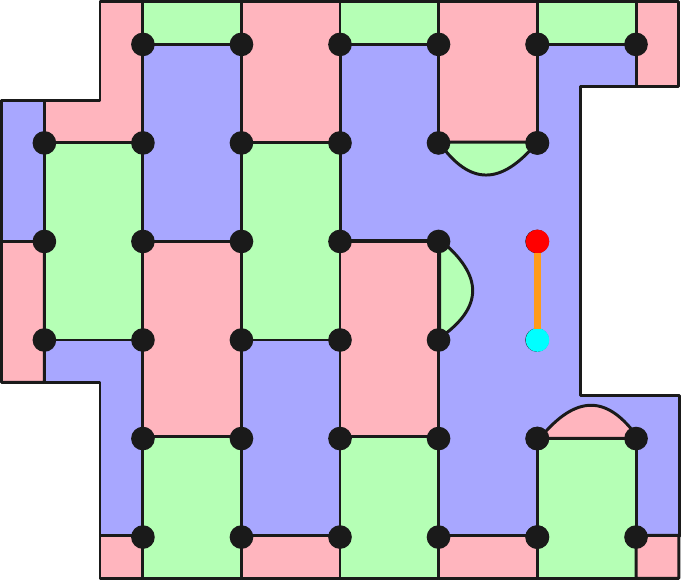}
        \caption{\label{fig:defect-boundary-vertical}}
    \end{subfigure}\hfill
    \begin{subfigure}{0.3\linewidth}
        \includegraphics[width=\linewidth]{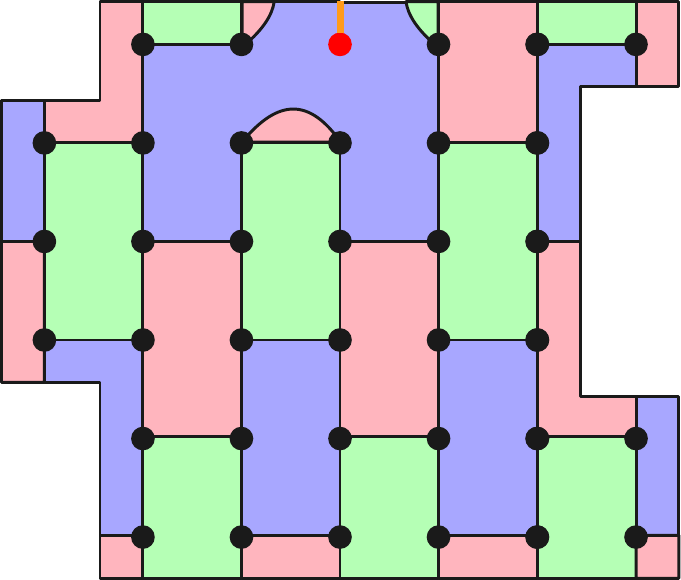}
        \caption{\label{fig:defect-boundary-bad-edge}}
    \end{subfigure}
    \caption{\label{fig:defect-boundary} Removing defective qubits on the (a) horizontal and (b) vertical boundaries of the honeycomb code. (c) A boundary edge is selected as a defect edge, which has the unintended effect of introducing two corner plaquettes (of weight $1$) in red and green. Thus choosing this edge as the defective edge is not possible.}
\end{figure}

It is worth noting that in this situation there are constraints on which edge we can choose to be the defect edge. In particular, the defect edge basis must not match the basis of the boundary. As shown in \cref{fig:defect-boundary-bad-edge}, choosing the boundary edge as the defect edge leads to weight-one plaquettes being defined on the boundary. These plaquettes are formed by two non-commuting single-qubit measurements, and are therefore non-deterministic throughout the whole computation. Introducing these non-deterministic plaquettes will create additional logical operators, potentially reducing the code distance or logical fidelity. Similar behaviour can be seen in the corner plaquettes on the defect-free planar honeycomb code \cite{Gidney2022benchmarkingplanar}.

The fact that boundary plaquettes on planar Floquet codes are non-deterministic also leads to other considerations we might want to take into account when selecting a defect edge. This is because while detectors can be formed for plaquettes in the bulk on average every three sub-rounds, detectors for plaquettes on the boundaries can only be formed once every six rounds, and can never be formed for plaquettes on the corners. As a result, our choice of defect edge could mean introducing a large super-plaquette that is only deterministically measured once every six sub-rounds, or even not measured at all. We describe the means of avoiding such cases, when possible, in Appendix~\ref{app:edge-selection}.

\subsection{Changing the logical operators}
\label{ssec:logical-observables}
The planar honeycomb code with the boundaries indicated above stores a single logical qubit, with logical operators (evolving in time) shown in \cref{fig:example-honeycomb}. Strings that cross the lattice non-trivially, between vertically or horizontally opposite boundaries, serve as our logical operators~\cite{Haah2022boundarieshoneycomb, Gidney2022benchmarkingplanar, Paetznick2023PerformancePlanarFloquetCodes}. In the periodic honeycomb code there are two logical qubits, with each having logical operators made up of one horizontal cycle and one vertical cycle such that the two anti-commute \cite{Hastings2021dynamically, Gidney2021faulttolerant}.

If the super-plaquettes formed by the algorithm in \cref{ssec:superplaquette-algorithm} are not incident to these strings then we do not need to make any changes to the logical operators. However, if the strings use an edge removed by the algorithm, or if the strings cross a defective qubit, then we can use a path-finding algorithm to construct alternative logical operators. If we cannot find strings that connect opposite boundaries in the modified graph, the defects cannot be resolved using the above algorithm because the graph does not percolate.

It is important to ensure that the new operators commute with the measurements as they evolve, in the same way that the logical operators in the original code transform in order to commute with each round of measurements. We give examples of how the logical operators can evolve around a newly formed super-plaquette on both the periodic and planar honeycomb codes in \Cref{app:observable-evolution}.

\subsection{Choosing a defect edge}
\label{ssec:choosing-defect-edge}

In order to run the algorithm described in \cref{ssec:superplaquette-algorithm}, we need to pick suitable defect edges incident to the defective qubits. For $n$ defective qubits, the number of possible defect edge assignments scales as $O(3^n)$, making it infeasible to exhaustively search over all possible edges to find an optimal selection.

Instead we use a number of heuristics to find a good set of defect edges. Our aim with these heuristics are two-fold. First, we want to keep super-plaquettes small, to ensure that removing the defects does not reduce the distance too significantly. And second, we want to keep super-plaquettes in the bulk of the code where possible, to ensure they are deterministically measured throughout the whole experiment. These heuristics were employed for the numerical simulations described in \cref{sec:simulation-results}. We give full details on the heuristics used in \Cref{app:edge-selection}.

\section{Simulation results}
\label{sec:simulation-results}

We shall now discuss our simulations of the planar honeycomb code with fabrication defects. Our simulations used multi-qubit Pauli product measurement (MPP) gates, which are native to Majorana devices \cite{Karzig2017MajoranaZeroModes, Paetznick2023PerformancePlanarFloquetCodes}, as well as single-qubit reset and measurement gates. For a Pauli noise parameter $p_{\textrm{Pauli}}$, we assume two-qubit depolarising noise with probability $p_{\textrm{Pauli}}$ after each MPP gate, single-qubit depolarising noise with probability $p_{\textrm{Pauli}}/10$ after each single-qubit reset or measurement gate, and that a measurement result flips with probability $p_{\textrm{Pauli}}$. This is similar to the standard depolarising with entangling measurements (SDEM3) noise model considered by Gidney \etal \cite{Gidney2022benchmarkingplanar}.

Our simulations are run as follows. For each ``target distance'' $d$ (the code distance for that lattice size in the absence of defects) we construct a $2d \times 3d$ planar honeycomb patch (see \cref{fig:example-honeycomb} for a $d=4$ example). We then sample defective data qubits from this patch with probability $p_{\textrm{defect}}$, and adapt the honeycomb code according to the algorithm discussed above. We then generate Stim circuits for running a $3d$-sub-round quantum memory experiment using this code for both logical observables \cite{Gidney2021Stim}. Finally, the measurement results for these experiments are decoded using PyMatching \cite{Higgott2023PyMatching}.

Because each sample of fabrication defects requires a different Stim circuit, providing a complete set of Stim circuits generated in these experiments would be impractical. Instead, we provide some example Stim circuits in Ref.~\cite{StimCircuitsZenodo}.

\subsection{Percolation threshold}

First, we benchmark how well this method preserves a logical qubit under fabrication defects. As a reminder, the logical observable on the honeycomb code is a product of Pauli operators on qubits in a non-trivial cycle for the periodic case, or in a path between two boundaries of the same type for the planar case. Such a cycle or path can still be created when using this method to account for fabrication defects as long as a super-plaquette does not form a non-trivial cycle or path between two boundaries itself. Hence we are interested in the probability for the modified graph to percolate at various fabrication defect probabilities.

We estimate the probability of successfully preserving the logical observables using Monte Carlo sampling. First, we generate instances of the honeycomb code lattice for target code distances ranging from $d=3$ to $d=7$. We then select fabrication defect probabilities up to $14\%$. For each distance and defect rate, we generate up to 10,000 lattice samples, and use the method described in \cref{sec:strategy} to adjust the code for them. We then estimate the probability that, when the algorithm is applied, the resulting graph does not percolate. Finally, we estimate a percolation threshold using second-order approximation methods, as described in Chapter 4 of Ref.~\cite{Harrington2004Thesis}. Since this is only a property of the graph and the algorithm described above, the resulting percolating probabilities should apply to all Floquet codes defined on this lattice, not just the honeycomb code.

The results are presented in \cref{fig:percolation-threshold}, with close-to-threshold data in \cref{fig:close_thresh_perc}. We find that applying these techniques to the honeycomb code achieves a percolation threshold of approximately $13.1 \pm 0.1\%$, which is comparable to the percolation threshold found for the surface code of just under $16\%$ \cite{Auger2017FabricationSurfaceCode}.

\begin{figure}
    \centering
    \begin{subfigure}{0.5\linewidth}
        \includegraphics[width=\linewidth]{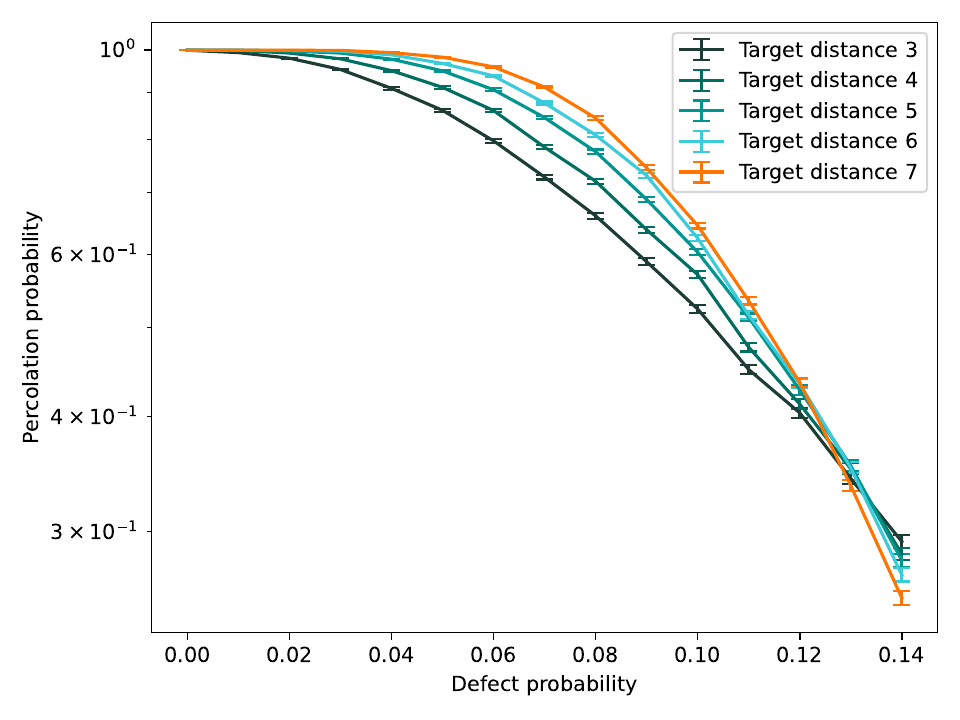}
        \caption{}
    \end{subfigure}\hfill
    \begin{subfigure}{0.5\linewidth}
    \includegraphics[width = \linewidth]{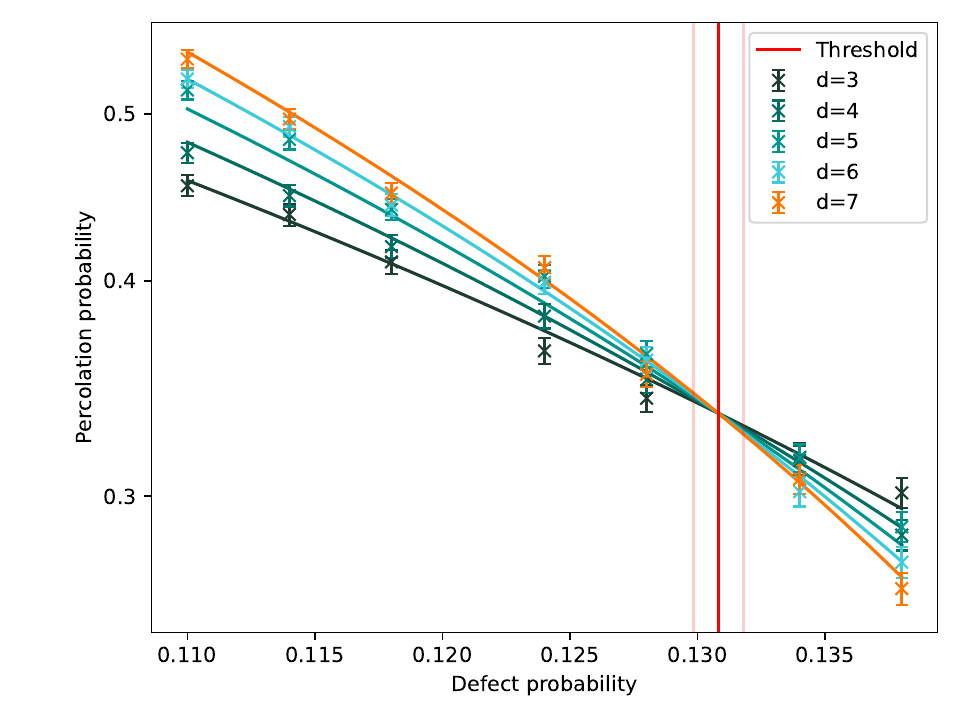}
    \caption{\label{fig:close_thresh_perc}}
    \end{subfigure}
    \caption{Probabilities for the graph to percolate when the super-plaquette strategy is used to remove fabrication defects for varying defect probabilities and target distances $d$. Error bars indicate standard error of the mean. (a) Unfitted data averaged over 10,000 realisations. (b) Curve fits to the close-to-threshold data (averaged over 6,000 realisations) are shown. The critical defect probability is around $13.1 \pm 0.1\%$, as indicated by the red lines.}
    \label{fig:percolation-threshold}
\end{figure}

\subsection{Average effective Pauli distance}

Next, we benchmark the effective Pauli distance for each sample. As before, we start with target Pauli distances ranging from 3 to 7, defect probabilities ranging from $1\%$ to $14\%$, and generate up to 10,000 defect samples for each target distance and defect probability. Each circuit is then implemented in Stim with Pauli noise inserted, and the distances for both the horizontal and vertical observables are estimated using Stim's functionality for finding the shortest graph-like error \cite{Gidney2021Stim}. Note that if a code graph fails to percolate, then the effective distance is taken to be zero.

The results are presented in \cref{fig:effective-distance}. By plotting the average effective distance on a log scale we find that for a code with a fixed target distance, the effective distance decreases exponentially with the defect probability up to the percolation threshold. However, for small defect probabilities the expected distance grows linearly with the target distance, indicating the possibility of a non-zero Pauli threshold. We will show below that such a threshold exists for defect probabilities up to roughly the percolation threshold in \cref{ssec:pauli-error-threshold}. This means that if the probability of a qubit being defective is sufficiently small, achieving a desired Pauli distance is simply a case of creating a sufficiently large patch.

\begin{figure}
    \centering
    \includegraphics[width=0.6\linewidth]{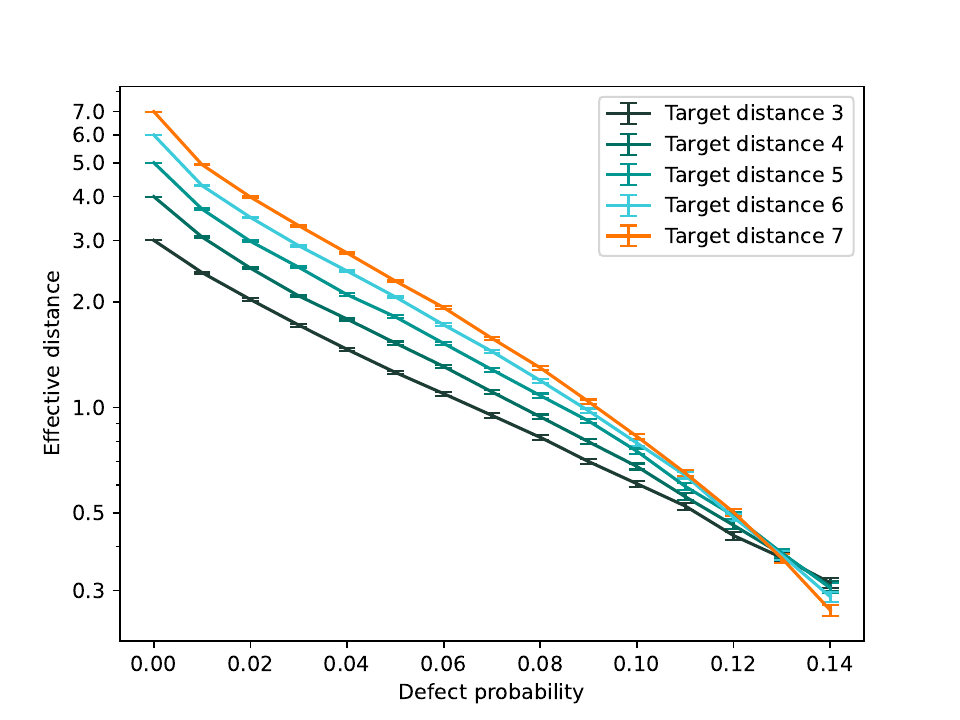}
    \caption{\label{fig:effective-distance} Average effective Pauli distance found for given target distances and defect probabilities. Error bars indicate standard error of the mean.}
\end{figure}

\subsection{Pauli error threshold}
\label{ssec:pauli-error-threshold}

Finally, we investigate how the Pauli error threshold changes depending on the fabrication defect probability. For each target distance $d$ and fabrication defect probability $p_{\textrm{defect}}$, we generate 1,000 defect samples and adapt the code to these defects. For each code instance, if the resulting graph percolates, we next generate Stim \cite{Gidney2021Stim} circuits for a $3d$-sub-round quantum memory experiment and insert Pauli noise with varying error probabilities $p_{\textrm{Pauli}}$. We then collect 100,000 quantum memory experiment samples, for each code instance and for each observable (both horizontal and vertical). We then decode each experiment. For horizontal and vertical logical error probabilities $p_H$ and $p_V$ we estimate an overall logical error probability for a given experiment as

$$p_L = 1 - (1 - p_H)(1 - p_V),$$

\noindent which is the estimated probability of either logical error occuring. We average $p_L$ across all generated defect samples whose resulting code graphs percolate. Since this results in an estimate to the logical failure probability given that the graph percolates, we can use this to estimate the overall failure probability $p_{F}$ for a given defect probability and target distance as

$$p_{F} = p_{\textrm{perc}}p_L + (1 - p_{\textrm{perc}}),$$

\noindent where $p_{\textrm{perc}}$ is the probability that the code manages to successfully percolate. In the above, we consider any non-percolating case as a logical failure.

Finally, from these logical error probabilities a threshold is estimated. This is done by approximating logical error rates close to threshold as quadratically dependent on $x = (p_{\textrm{Pauli}}-p_0)d^{1/\nu_0}$, where $p_0$ is the Pauli error threshold and $\nu_0$ is a critical exponent for the transition (see Chapter 4 of Ref.~\cite{Harrington2004Thesis}). An example fit at $p_\text{defect} = 6\%$ is shown in \cref{fig:rescaled-harrington-fit-6-percent}, with the resulting threshold presented alongside Pauli error rates in \cref{fig:harrington-fit-6-percent}. Further data at other defect probabilities are provided in \Cref{app:additional_numerics}.

\begin{figure}
    \centering
    \begin{subfigure}{0.5\linewidth}
        \includegraphics[width=\linewidth]{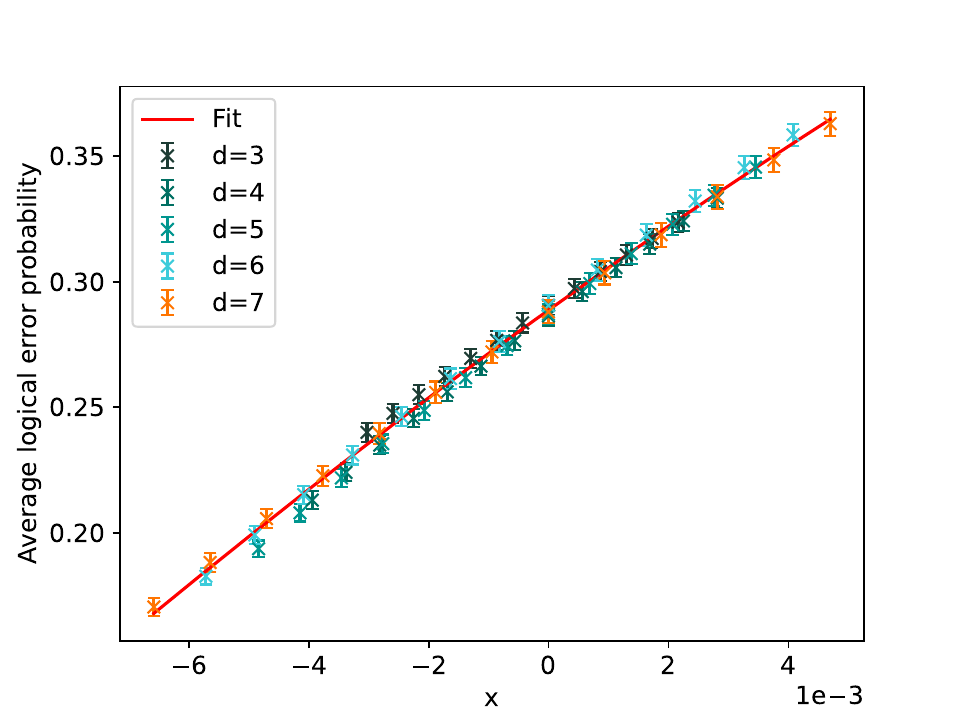}
        \caption{}
        \label{fig:rescaled-harrington-fit-6-percent}
    \end{subfigure}\hfill
    \begin{subfigure}
    {0.45\linewidth}
        \includegraphics[width=\linewidth]{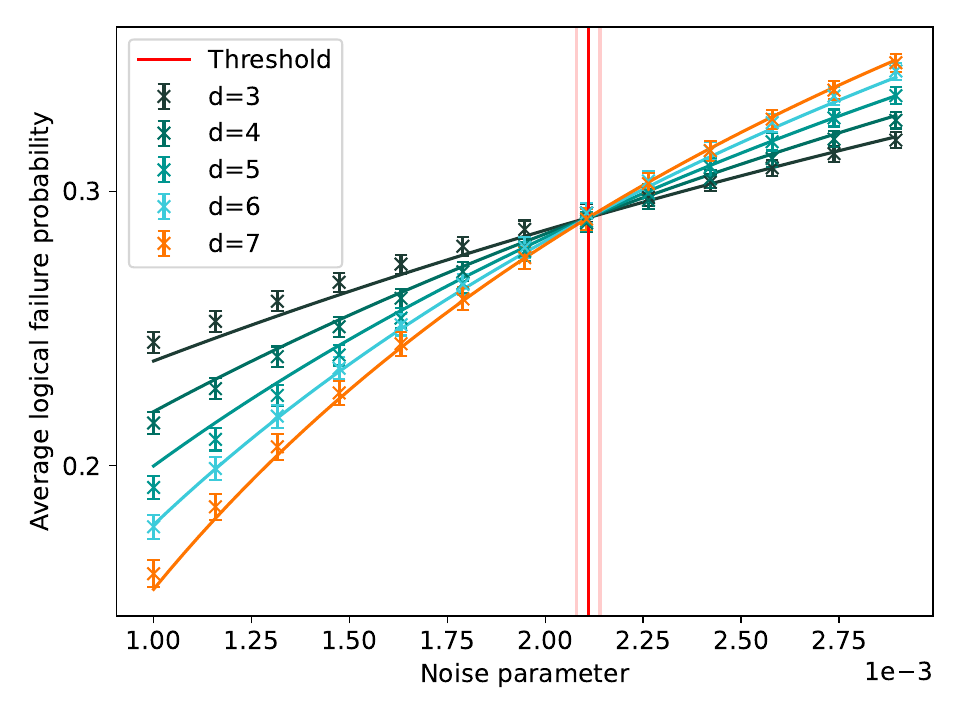}
        \caption{}
        \label{fig:harrington-fit-6-percent}
    \end{subfigure}
    \caption{Estimating a Pauli error threshold at $p_\text{defect} = 6\%$. (a) Rescaled data $x = (p-p_0) d^{1/\nu_0}$ is plotted against average logical error probability, along with a quadratic best fit. Error bars indicate standard error of the mean for the 1,000 $p_L$ values sampled; the errors associated with each $p_L$ estimate, taken from 100,000 quantum memory experiment samples, is comparatively negligible. (b) Average logical error probability for various Pauli noise parameters, with estimated Pauli threshold (and error) in red.}
    \label{fig:harrington-fit}
\end{figure}

\begin{figure}
    \centering
        \includegraphics[width=0.6\linewidth]{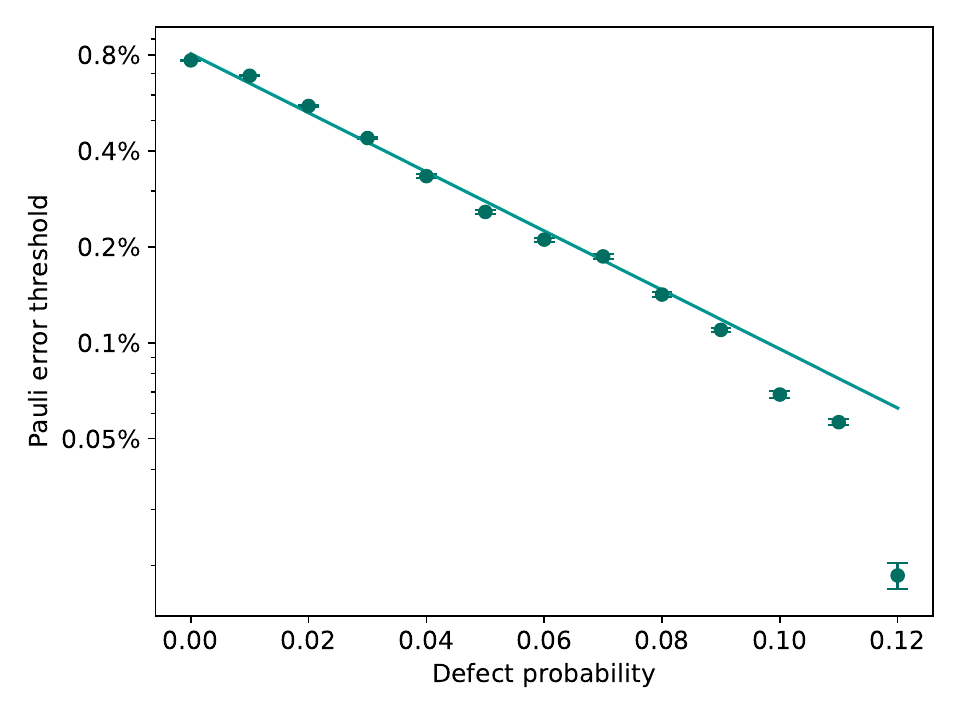}
    \caption{\label{fig:pauli-thresholds} Pauli error thresholds for varying qubit defect probabilities. The points are fitted to the exponential function $A \exp(B p_{\textrm{defect}})$, where $p_{\textrm{defect}}$ is the defect probability, $A=(8.1\pm 0.2) \times 10^{-3}$ and $B=-21.4\pm 0.9$. As can be seen, the accuracy of the exponential fit is significantly reduced close to the percolation threshold. At $p_\text{defect} \approx 13\%$, no Pauli error threshold exists.}
\end{figure}

The results are presented in \cref{fig:pauli-thresholds}. We fit the Pauli error threshold to an exponential function of $p_\textrm{defect}$, where $p_{\textrm{defect}}$ is the independent probability of a qubit being defective. This exponential function provides a reliable fit up to a defect probability of $9\%$. At $p_{\textrm{defect}} \geq 10\%$, all target distances simulated had an effective Pauli distance $\lesssim 1$ (see \cref{fig:effective-distance}), which is likely contributing to this drop-off in threshold. However, note that the exponential fit is guaranteed to be inaccurate around the percolation threshold, where the Pauli error threshold goes to zero exactly.  

The performance found is similar to that seen in the surface code, where Pauli thresholds have been shown to scale as $7\times 10^{-3} \exp(-22p_{\textrm{defect}})$~\cite{Auger2017FabricationSurfaceCode}. Note that the surface code simulations in Auger \etal \cite{Auger2017FabricationSurfaceCode} use a slightly different circuit-level noise model to that used here. The parameters we find for the exponential fit may change if one were to replace MPP gates with measurement circuits. It is notable, however, that for hardware with native MPP gates, such as Majorana Fermions \cite{Paetznick2023PerformancePlanarFloquetCodes, Karzig2017MajoranaZeroModes}, the code performance with finite $p_\text{defect}$ is very close to that of the surface code.

Despite the small average effective distances for $p_\textrm{defect} \geq 10\%$, we find clear evidence for sub-threshold behaviour at $p_\textrm{defect} = 10\%$ (see \Cref{app:additional_numerics}). The evidence for $p_\textrm{defect} > 10\%$ is less clear, although our fits, and the increase in effective distance with target distance for $p_\textrm{defect} > 10\%$ (\cref{fig:effective-distance}), indicate a positive threshold.

\section{Hardware-specific improvements}
\label{sec:hardware-improvements}

Throughout this paper, we have made minimal assumptions about the underlying hardware: the constraints we have for implementing our modified Floquet codes are the same as the constraints for implementing the original codes in a defect-free lattice.

In this section we will now discuss what extensions can be made when we know more about the quantum hardware. First, we propose an improvement that can be made if the Floquet code is being implemented on a square lattice---which is common in superconducting quantum computers \cite{Krinner2022SurfaceCodeSuperconducting, GoogleSuppressing, RigettiAnkaa}---by taking advantage of extra connectivity which is not used in the original code. And second, we propose an improvement that can be made on the heavy hex lattice used in IBM quantum processors \cite{IBMHeavyHex}, by treating defective auxiliary qubits as a defective connection rather than a defective data qubit. Further details on using additional connectivity and accommodating defective connections can be found in \Cref{app:extra_connectivity,app:defective-connections}, respectively.

\subsection{Accommodating defective qubits on a square lattice}
\label{ssec:square-grid-connectivity}

One key benefit of Floquet codes compared to other error-correcting codes is the low connectivity requirements: each data qubit only needs to be connected to three other data qubits, and each measurement involves only two data qubits. But what if the device does have extra connectivity? Many quantum devices, for instance, have connections based on a square lattice rather than a hexagonal lattice \cite{Krinner2022SurfaceCodeSuperconducting, GoogleSuppressing, RigettiAnkaa}.

It is possible to implement Floquet codes such as the honeycomb code on a square lattice, through only using a subset of available connections on the device. An example of this is given in \cref{fig:initial-defect-square-lattice}, with unused connections shown in white. We are also therefore able to accommodate defective qubits without these additional connections using the strategy described in \cref{sec:strategy}, as shown in \cref{fig:defect-removed-square-lattice}. As before, this techique involves reducing the shrink plaquettes to weight-2 and combining four merge plaquettes into one super-plaquette.

However, we can improve on this strategy by using these additional connections in the quantum hardware. In \cref{fig:defect-removed-square-lattice-extra-connections}, we instead reduce the shrink plaquettes to weight-4, and combine only two merge plaquettes together. This can be done by using two additional connections available on the square lattice. The result is a smaller super-plaquette, leading to a higher distance: the horizontal distance is reduced by one and the vertical distance is unchanged, compared to both being reduced by one when not using extra connectivity.

It is worth noting that this ability to take advantage of additional connectivity is not unique to the square lattice, and can in fact be used to better accommodate defects in a number of different hardware layouts. We provide further details on how to take advantage of extra connections in hardware in \Cref{app:extra_connectivity}.

\begin{figure}
    \centering
    \begin{subfigure}{0.3\linewidth}
        \includegraphics[width=\linewidth]{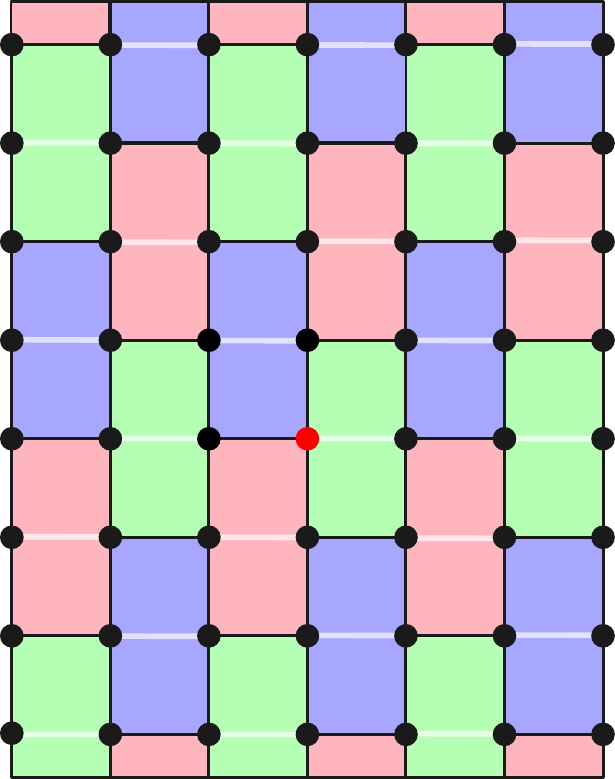}
        \caption{\label{fig:initial-defect-square-lattice}}
    \end{subfigure}\hfill
    \begin{subfigure}{0.3\linewidth}
        \includegraphics[width=\linewidth]{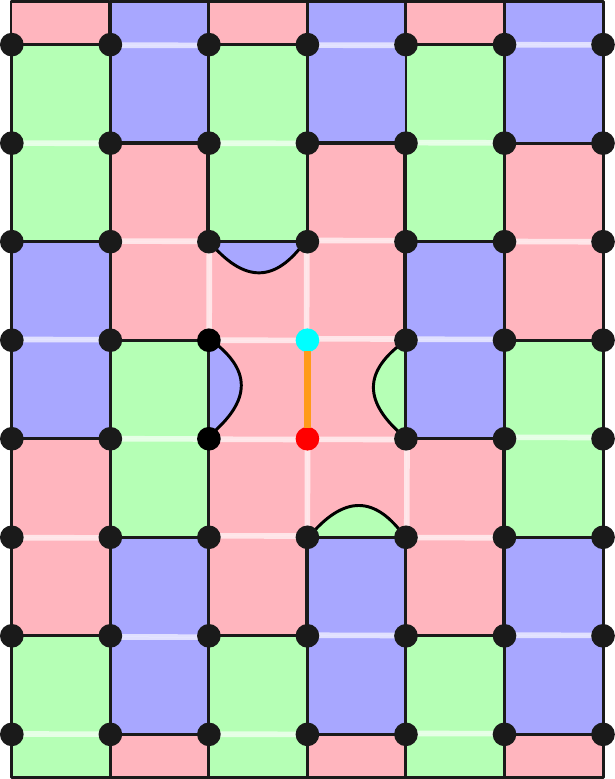}
        \caption{\label{fig:defect-removed-square-lattice}}
    \end{subfigure}\hfill
    \begin{subfigure}{0.3\linewidth}
        \includegraphics[width=\linewidth]{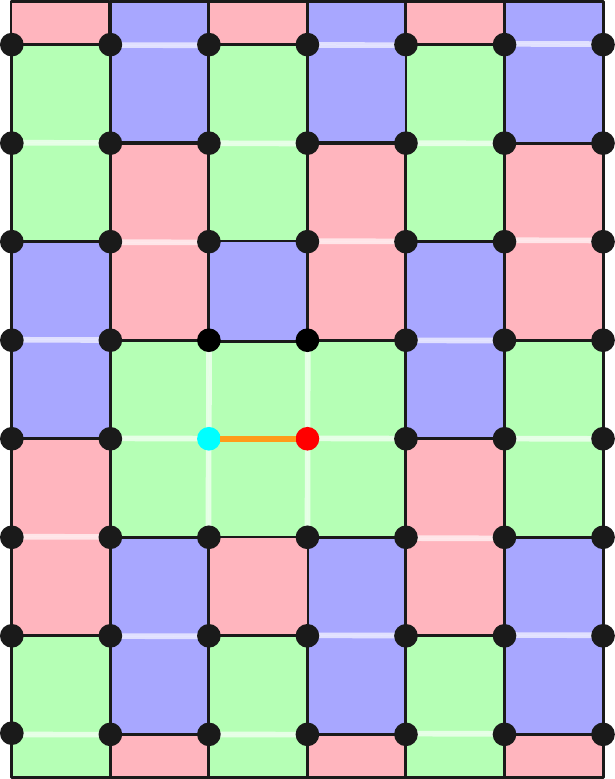}
        \caption{\label{fig:defect-removed-square-lattice-extra-connections}}
    \end{subfigure}
    \caption{\label{fig:square-lattice} Removing a defective qubit from a Floquet code implemented on a square lattice. (a) An implementation of the honeycomb code on a square lattice. White lines denote connections which are not used in the experiment. A defective qubit is shown in red. (b) Removing the defective qubit using the algorithm described in \cref{sec:strategy}. This method uses the same connections as used by the original honeycomb code. (c) Removing the defective qubit using two edges which are not used in the honeycomb code but are available in the square lattice to reduce the size of the super-plaquette.}
\end{figure}

\subsection{Accommodating defective qubits on a heavy hex lattice}
\label{ssec:defective-connections}

So far we have only considered defective data qubits on the honeycomb code. But what if a device has faulty connections rather than defective qubits? This is especially relevant for implementing Floquet codes on, for example, the heavy hex lattice used in IBM quantum processors \cite{IBMHeavyHex}. The heavy hex lattice is similar to the hexagonal lattice used for the honeycomb code, but with ``auxiliary qubits'' on the edges of the lattice. An example of a heavy hex lattice with the honeycomb code colouring is presented in \cref{fig:initial-defect-heavy-hex}. Two-body measurements between ``data qubits'' on the vertices can then be measured using an auxiliary syndrome extraction circuit \cite{Shor1996SyndromeExtraction}, thus allowing us to implement the honeycomb code on the heavy hex lattice \cite{Gidney2021faulttolerant, Gidney2022benchmarkingplanar}. Defective data qubits can be removed using the same strategy as in \cref{sec:strategy}, and therefore will not be the focus of this section.

If an auxiliary qubit is found to be defective, we can accommodate this by treating the neighbouring data qubits as defective and removing them using the algorithm in \cref{ssec:superplaquette-algorithm}. The result of this approach is presented in \cref{fig:defect-removed-heavy-hex-qubit-strategy}. This involves removing two data qubits and seven auxiliary qubits. We also split two plaquettes into weight-two plaquettes and merge four plaquettes into one super-plaquette.

Alternatively, we are able to treat the defective auxiliary qubit as a defective connection. In \Cref{app:defective-connections}, we present an approach for removing defective connections from Floquet codes while maintaining the same properties of tri-valence and 3-face-colourability. Our strategy works by selecting one plaquette incident to the defective connection to be the shrink plaquette, and the other to be a merge plaquette. Other plaquettes which are adjacent to the shrink plaquette and of the same colour as the merge plaquette are also selected to be merge plaquettes. With the shrink plaquette and merge plaquettes chosen, Steps 4-7 of the algorithm in \cref{ssec:superplaquette-algorithm} can then be followed, with $D$ being the basis of the defective edge. In \Cref{app:defective-connections}, we prove the following:

\begin{theorem}[for defective auxiliary qubits]
    The above process produces a lattice such that (a) each data qubit is incident to three auxiliary qubits, (b) each auxiliary qubit is incident to two data qubits, and (c) the lattice is 3-face-colourable.
    \label{thm:validity-defective-connections}
\end{theorem}

By using this approach to remove just the defective connection rather than the data qubits incident to that connection, we are able to create the Floquet code shown in \cref{fig:defect-removed-heavy-hex-connection-strategy}. This approach removes no data qubits, and only removes two additional auxiliary qubits alongside the defective one. In \Cref{app:defective-connections}, we show that this is also optimal for uniform heavy lattices without requiring additional connectivity:

\begin{theorem}[for defective auxiliary qubits]
    Consider a defective auxiliary qubit in the bulk of a Floquet code defined on a heavy hex, heavy 4.8.8 or heavy 4.6.12 lattice (without boundary); i.e., a uniform tiling of the plane. Under the constraint that we cannot form an edge between any two qubits not connected by an edge in the original lattice, the algorithm described above removes the defective auxiliary qubit with the minimal number of removed auxiliary qubits from the original lattice.
    \label{thm:minimal-auxiliary-qubit-removal}
\end{theorem}

This technique can also be used to accommodate a defective connection between a data qubit and an auxiliary qubit, by treating the auxiliary qubit incident to the connection as defective. It is interesting to compare this to the surface code, where defective connections are mapped to a defective data qubit rather than a defective auxiliary qubit \cite{Auger2017FabricationSurfaceCode}.

\begin{figure}
    \centering
    \begin{subfigure}{0.3\linewidth}
        \includegraphics[width=\linewidth]{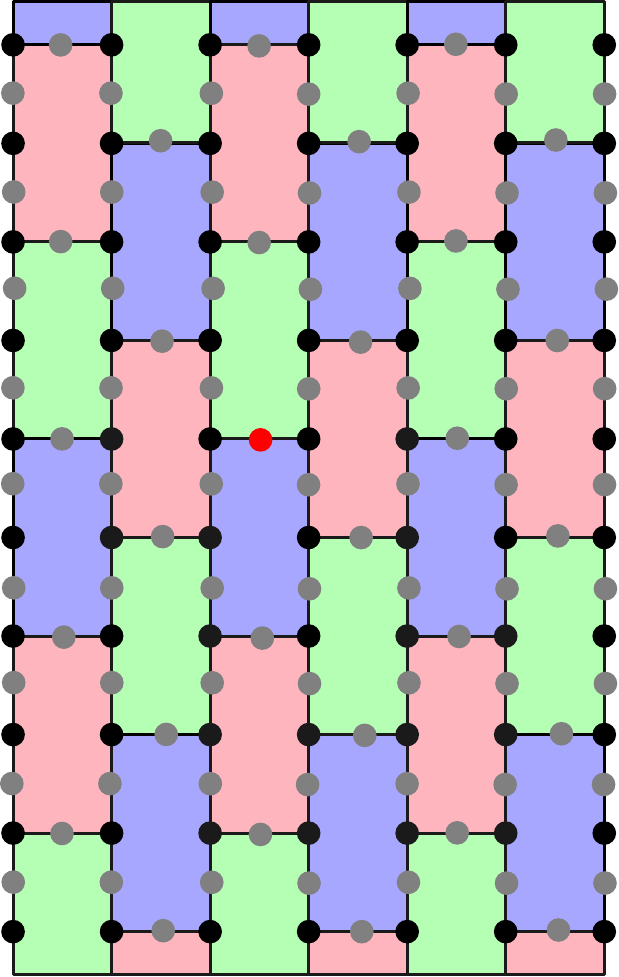}
        \caption{\label{fig:initial-defect-heavy-hex}}
    \end{subfigure}\hfill
    \begin{subfigure}{0.3\linewidth}
        \includegraphics[width=\linewidth]{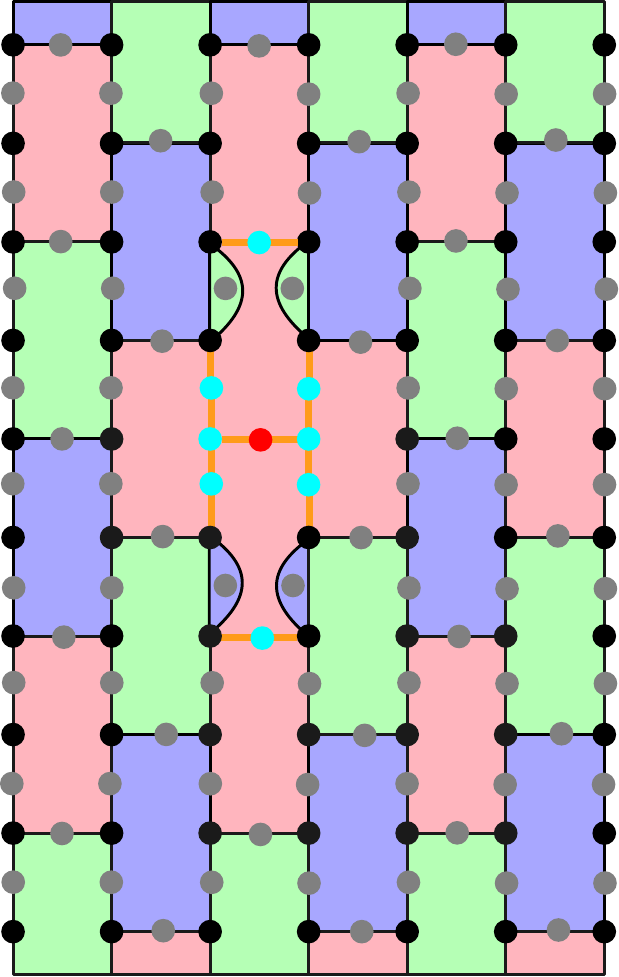}
        \caption{\label{fig:defect-removed-heavy-hex-qubit-strategy}}
    \end{subfigure}\hfill
    \begin{subfigure}{0.3\linewidth}
        \includegraphics[width=\linewidth]{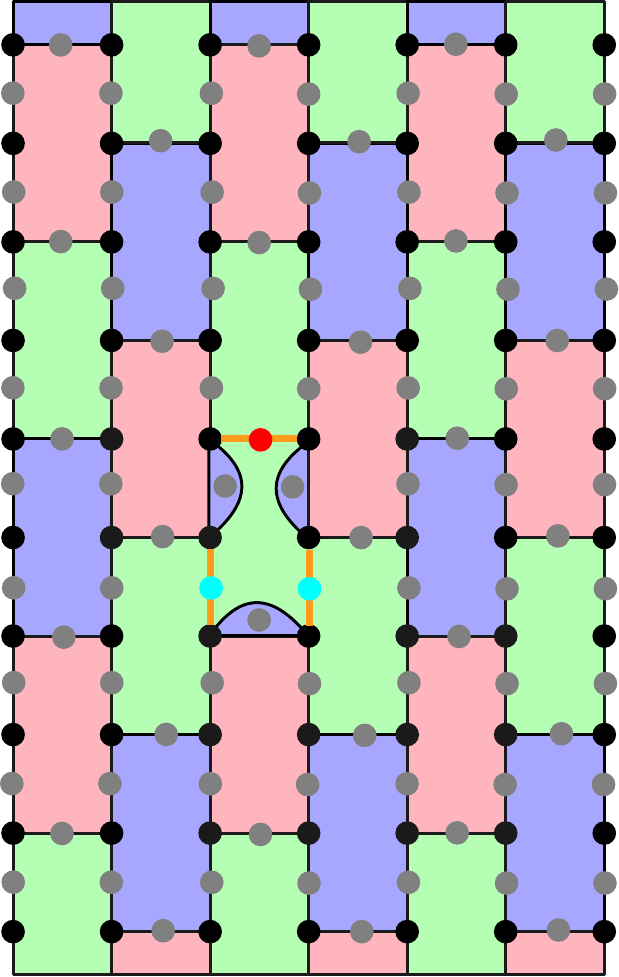}
        \caption{\label{fig:defect-removed-heavy-hex-connection-strategy}}
    \end{subfigure}
    \caption{\label{heavy-hex} Removing a defective auxiliary qubit from a Floquet code implemented on a heavy hex lattice. (a) An implementation of the honeycomb code on the heavy hex lattice. Auxiliary qubits on the edges of the hexagonal lattice are shown in grey, with one defective auxiliary qubit shown in red. (b) Removing the defective auxiliary qubit by using the algorithm described in \cref{ssec:superplaquette-algorithm} to remove the two adjacent data qubits. (c) Removing the defective auxiliary qubit by treating it as a defective connection and following the approach described in \Cref{app:defective-connections}. Removed qubits alongside the qubit are shown in blue, with their connections shown in orange. For weight-2 plaquettes, the auxiliary qubits are placed in the middle of the plaquettes to indicate that the same auxiliary qubit can be used to measure both edges; this requires no modifications to the heavy hex lattice.}
\end{figure}

\section{Conclusion}
\label{sec:conclusion}

In this work, we have provided new methods for accommodating fabrication defects on the honeycomb code. In doing so, we do not require any additional connectivity beyond what is required for the original code, nor do we require any modifications to the measurement schedule. We have also shown that our method is optimal in terms of the number of qubits, stabilisers, and connections removed from the lattice. Through simulations, we have shown the code to be scalable up to a defect probability of $13\%$. Finally, we have shown that further improvements can be made depending on the underlying hardware, by taking advantage of extra connectivity or defective connections.

Our work is generalisable beyond the honeycomb code, owing to the adapted lattices maintaining tri-valence and 3-face-colourability and imposing no additional requirements on the measurement schedule of a Floquet code implemented. Thus, any Floquet code which can be implemented on such a lattice can be adapted for implementation on a lattice with defects removed (see \Cref{app:alternative-codes}).
However, moving beyond these codes considered would require adapting our method. Examples of Floquet codes outside of this class are the 2-D ``Floquet colour code'' of Ref.~\cite{Dua_2024_rewinding}, all 3-D Floquet codes \cite{Davydova2023CSSFloquet, Zhang2023Floquet3D,Dua_2024_rewinding} and Floquet codes with twist defects \cite{Ellison2023FloquetTwist}. There are also other static codes which can be defined on a tri-valent, 3-face-colourable graph, with the best-known example being the colour code \cite{Bombin2006ColourCode, Gidney2023ColourCode}. It is likely possible to adapt our method to such static codes.

We assumed the need to create a fixed lattice with key properties on which the updated Floquet code is defined. There may be improvements to be obtained by relaxing this assumption~\cite{aasen2023faulttolerant}. In particular, the adaptation of surface code techniques~\cite{Nagayama2017DefectiveLattice,Strikis2023ScalableQC,Siegel2023adaptivesurfacecode,debroy2024lucisurfacecodedropouts} to Floquet codes may prove possible. However, there are difficulties associated with this. The equivalence between a Floquet code at a single time step and the surface code does not allow one to easily determine how to form detectors, or how to define the measurement schedule to allow for boundaries and a low logical error rate. There are significant challenges to adapting surface code techniques that our strategy sidesteps, allowing for easy application to a range of Floquet codes. 

While we have discussed how to implement this strategy when there are hardware defects known at the start of the experiment, one area that needs further work is accommodating defects which occur in the middle of an experiment, such as via high-energy radiation, qubit loss or erasure noise. Qubit loss in the middle of the experiment is common in both photonic and neutral atom systems \cite{Barrett2010MBQCErasure, Whiteside2014Loss, Auger2018FusionFailures, Sahay2023BiasedErasure}. Different proposals for correcting erasure events (which can also model qubit loss) in Floquet codes have been studied, through either modifying the stabilisers around the erased qubits \cite{Paesani2023CSSFloquetFoliate}, or through introducing additional Pauli error mechanisms \cite{Gu2023ErasureHoneycomb}. A comparison of these methods to our code modifications around defects would be useful for Floquet code implementation in hardware. Aasen \etal also have proposals for identifying defective qubits mid-experiment \cite{aasen2023faulttolerant}. We believe our work can also be adapted to accommodate defects which occur mid-experiment, but leave this for future work.

There are also other more damaging noise models one can consider. Two such models are leakage, where a qubit enters a state outside of the computational subspace \cite{Brown2019Leakage}, and high-energy events where large numbers of qubits can be impacted at the same time \cite{McEwen2022CosmicRays}. Both of these forms of error can be accommodated through a combination of identifying the error locations and adapting the codes accordingly, for which our method could be applied \cite{Miao2023OvercomingLeakage, Siegel2023adaptivesurfacecode}. Leakage can also be tackled using dedicated operations known as leakage reduction units~\cite{Aliferis2005Leakage,Fowler2013Leakage}. Future work could adapt these techniques to correct for leakage in Floquet codes.

\section*{Acknowledgements}

We thank Matthew P.\ Stafford, whose helpful explanations of how fabrication defects are handled on the surface code inspired this work. We also thank Setiawan for helpful feedback on an earlier version of this manuscript. We thank Benjamin B\'{e}ri, Dan Browne and our colleagues at Riverlane for insightful discussions. Finally, we thank the reviewers and editor at Quantum for helpful comments during the review process.

\bibliography{references}

\appendix

\section{Improvements with extra connectivity}\label{app:extra_connectivity}

We now demonstrate the improvements that are possible given extra connectivity in the quantum processor: that is, given the ability to measure check operators between pairs of qubits not connected in the honeycomb lattice.
We show how extra connectivity allows us to achieve optimal edge removal for a single defect.
This is interesting as the underlying quantum hardware on which one intends to perform a honeycomb code experiment, for example, may be a square grid \cite{GoogleSuppressing, Krinner2022SurfaceCodeSuperconducting, RigettiAnkaa}. 
\Cref{fig:defects_extra_con} shows how the size of super-plaquettes can be reduced when extra connectivity is usable. 
In \cref{fig:defects_extra_con_hor}, we shrink the green and blue hexagonal plaquettes to the weight-4 plaquettes shown and merge two red plaquettes into a weight-10 plaquette. This is possible with a square-grid-connected hardware.
Alternatively, as shown in \cref{fig:defects_extra_con_ver}, we can form the same plaquettes around the defect but choose a vertical defect edge.

\begin{figure}
    \centering
    \begin{subfigure}{0.45\linewidth}
        \includegraphics[width=\linewidth]{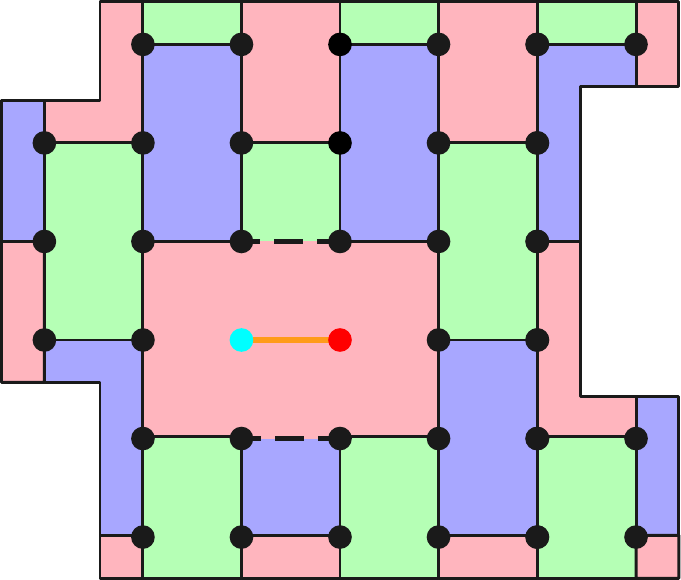}
        \caption{\label{fig:defects_extra_con_hor}}
    \end{subfigure}\hfill
    \begin{subfigure}{0.45\linewidth}
        \includegraphics[width=\linewidth]{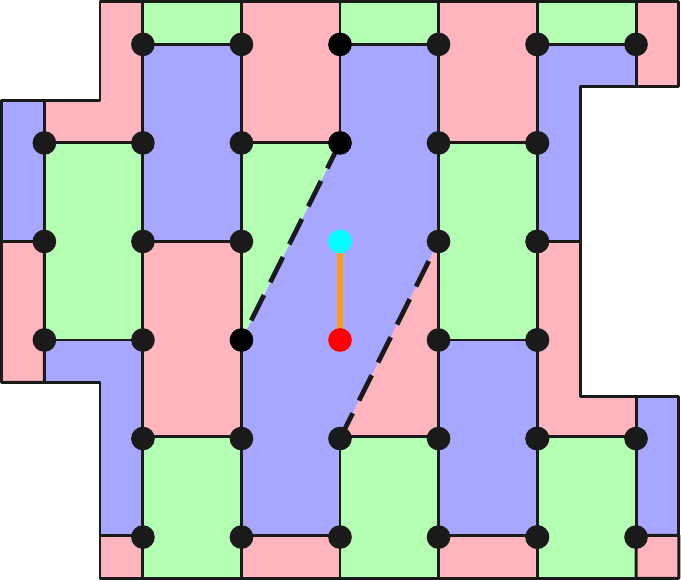}
        \caption{\label{fig:defects_extra_con_ver}}
    \end{subfigure}
    \caption{Super-plaquettes and weight-4 plaquettes for a single defect when extra connectivity (dashed lines) is present in the quantum hardware. (a) A horizontal defect edge is chosen. The resulting check operators are measurable for a square-grid connectivity. (b) A vertical defect edge is chosen.}
    \label{fig:defects_extra_con}
\end{figure}

With the above strategy, a single defect reduces \textit{either} the horizontal or vertical Pauli distance by one, while leaving the other unchanged.
Moreover, we can accommodate the defect while removing five connections from the code lattice, and adding two, resulting in a net loss of three connections. We show in \Cref{app:minimal-qubit-removal} that any strategy for removing a defective qubit must remove at least three edges, meaning that this approach is optimal in this regard.
Given these improvements, we therefore expect that the performance of Floquet codes with defective qubits, given this extra connectivity, would be improved: the percolation threshold would likely be higher, and the performance under Pauli noise likely better.

\section{Super-plaquette algorithm and minimal removal of qubits, plaquettes and edges}
\label{app:minimal-qubit-removal}

In this Appendix we prove certain features of the algorithm presented in \cref{ssec:superplaquette-algorithm}. We begin by commenting on why the algorithm preserves the important features of the lattice for defining the updated Floquet code: tri-valence and 3-face-colourability. We do this for the bulk of the code and note that the tri-valence and 3-face-colourability result follows for the boundary by cutting out the patch used, as explained in the main text. We also prove optimality, in terms of removed qubits, edges and plaquettes, in the bulk. There may be subtleties associated with optimality at the boundary (e.g., one could redefine the boundary around the defect), which we do not deal with here.

\begin{proof}[Proof of \Cref{thm:validity-qubit-removal}]
The first three steps of the algorithm simply involve selecting the defect edge, and the shrink and merge plaquettes. There will be exactly two shrink plaquettes\footnote{If $K_1$ (resp. $K_2$) is the weight of the first (resp. second) shrink plaquette, then the total number of merge plaquettes is $(K_1 + K_2)/2 - 2$. The first (resp. second) shrink plaquette is adjacent to $K_1/2$ (resp. $K_2/2$) merge plaquettes, with two merge plaquettes being incident to both shrink plaquettes.}.
Now notice that the edges in shrink plaquettes of type $T\neq D$ (where $D$ is the defect edge basis) do not feature in any plaquette after the algorithm is applied: they are removed from the lattice (see \cref{fig:example-defect}). However, every qubit incident to such an edge in the original lattice (apart from the two removed qubits) is incident to an \textit{additional} edge introduced in Step 4 of the algorithm. Hence, each vertex in the updated lattice remains incident to three edges as required.

Now, for 3-face-colourability, note that the two-body plaquettes are guaranteed to be adjacent to the super-plaquette, call it $S_1$, and one other plaquette from the original lattice, call it $S_2$. By 3-face-colourability of the original lattice, since a two-body plaquette inherits the colouring of the parent shrink plaquette, it is coloured differently to $S_1$ (which, by construction, is the same colour as a merge plaquette---the merge plaquettes are coloured differently to both shrink plaquettes) and to $S_2$ (which was also adjacent to the parent shrink plaquette). Meanwhile, the super-plaquette is adjacent to two-body plaquettes (which we have seen must be coloured differently to it) and to plaquettes of the original lattice adjacent to a merge plaquette, which we assume are coloured differently to it. Hence both types of newly added plaquette are adjacent to no other plaquettes of the same colour in the updated lattice. This is sufficient to conclude that the resulting lattice remains 3-face-colourable.

We will now show that the algorithm in \cref{ssec:superplaquette-algorithm} is optimal in terms of qubits, plaquettes and edges removed.
We assume (as above) that the lattice has periodic boundary conditions.
Since our scheme maintains tri-valence and 3-face-colourability, the number of removed plaquettes and edges are constrained by the number of removed qubits.
Tri-valence on the torus requires that $E = 3V/2$, where $E$ is the number of edges and $V$ is the number of vertices in the graph. 
This constraint is incompatible with an odd number of vertices, and hence we cannot remove just a single qubit.
Our algorithm is therefore optimal in terms of the number of removed qubits, removing as it does two qubits. The rest follows by necessity. 
If two qubits are removed, the number of vertices becomes $\Tilde{V} = V-2$, and tri-valence dictates that the number of edges becomes
$\Tilde{E} = 3\Tilde{V}/2 = 3V/2 - 3$. 
Removing more vertices results in more edges being removed: hence our algorithm is optimal in terms of the number of removed edges as well.

If the removal of the defect does not change the number of logical qubits then the Euler characteristic of the graph must not change. The Euler characteristic on a torus is $\chi = V - E + F = 0$, where $F$ is the number of plaquettes. Having removed two qubits and three edges, the Euler characteristic becomes

\begin{equation}
\begin{split}
    \Tilde{\chi} &= \Tilde{V} - \Tilde{E} + \Tilde{F} \\
    &= (V - 2) - (E - 3) + \Tilde{F} \\
    &= V - E + \Tilde{F} + 1,
\end{split}
\end{equation}

\noindent meaning that in order to maintain $\Tilde{\chi} = 0$, we must have $\Tilde{F} = F - 1$. Once again, removing more vertices requires the removal of more faces.
In summary, a scheme for removing a defective vertex from the lattice while maintaining 3-face-colourability and tri-valence must remove at least two vertices, three edges and one face. Therefore, the algorithm outlined in \cref{ssec:superplaquette-algorithm} is optimal in terms of qubit, edge and plaquette removal.
\end{proof}

In our scheme, we remove certain edges from the honeycomb lattice and add extra edges along existing connections in the hardware, meaning that we effectively remove some connections from the lattice (see white highlighted edges in \cref{fig:defect-removed}).
This is necessary if we do not rely on extra connectivity in the hardware (i.e., we do not have edges between qubits that are not connected in the honeycomb lattice).
Indeed, our scheme is also optimal in terms of the connections that are removed from the lattice, at least when considering a single defect removed from the honeycomb, 4.8.8 or 4.6.12 lattices (the tri-valent, 3-face-colourable lattices based on a uniform tiling of the plane by regular polygons).

\begin{figure}
    \centering
    \begin{subfigure}{0.4\linewidth}
        \includegraphics[width=\linewidth]{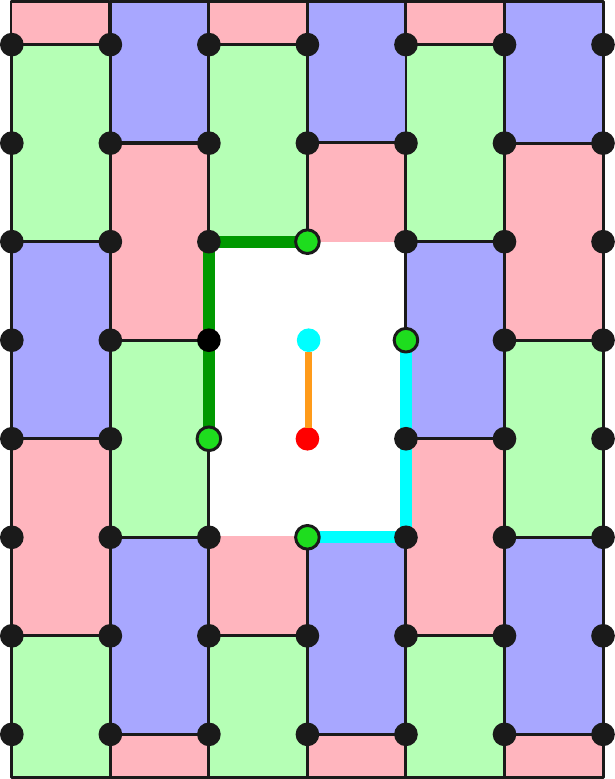}
        \caption{\label{fig:degree-two-vertices}}
    \end{subfigure}\hfill
    \begin{subfigure}{0.4\linewidth}
        \includegraphics[width=\linewidth]{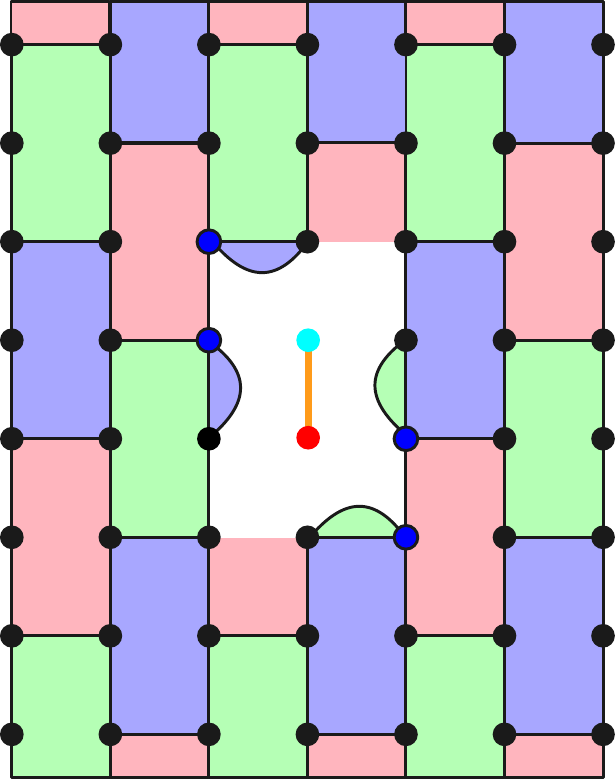}
        \caption{\label{fig:degree-four-vertices}}
    \end{subfigure}
    \caption{Vertices which violate the tri-valence requirements of Floquet codes when removing a defective qubit. (a) After removing edges incident to the defective qubit and its neighbour, four non-adjacent vertices are degree-2, thus requiring additional edges. These vertices are coloured green. These are connected in pairs by shortest paths in the lattice (coloured dark green and light blue respectively). (b) After adding more edges to correct the degree-2 vertices, four vertices in blue are degree-4, meaning that edges incident to these vertices must be removed.}
    \label{fig:vertex-degrees}
\end{figure}

\begin{proof}[Proof of \Cref{thm:minimal-edge-removal}]
From the above, we must remove two qubits per defect in the lattice. We first note that we must remove edges that are adjacent to these two removed qubits. An example is shown in \cref{fig:degree-two-vertices}. Note that we now have four vertices of degree two (coloured green in \cref{fig:degree-two-vertices}).
There are four of these because there are exactly five removed edges (one of which is the defect edge) and each pair of qubits is connected by at most one edge in the original lattice. 
Call the set of these degree-2 vertices $S$.
To maintain tri-valence after this edge removal, we must include extra edges incident to these vertices in $S$. We assume the edges in the original lattice, before the defect removal, account for all possible available connections in the device.
Hence, the extra edges we add must lie along existing connections in the lattice in order to avoid requiring additional connections in hardware.

In order to avoid introducing any vertices of degree-4 upon adding extra edges (see the blue vertices in \cref{fig:degree-four-vertices}), we need to remove some more edges from the graph. The way this must work is as follows. Firstly, without loss of generality, let the defect edge be coloured red. Then two of the vertices in $S$ were adjacent to a green removed edge and two were adjacent to a blue removed edge (see \cref{fig:degree-two-vertices}; the former lie on green plaquettes in the figure and the latter on blue plaquettes). Let these pairs of vertices belong to subsets $S_G$ and $S_B$ respectively. Notice that the vertices in $S_G$ ($S_B$) are not incident to any green (blue) edge (see \cref{fig:degree-two-vertices}). Thus we must add more green (blue) edges to the graph. 
To do so while maintaining tri-valence and 3-face-colourability (and to avoid self-loops in the graph), we must draw a path of edges between vertices in $S_G$ ($S_B$), where the path passes through non-green (non-blue) then green (blue) edges in an alternating fashion. We will add extra edges to the non-green (non-blue) parts of this path, while the green (blue) edges in this path will need to be removed in order to maintain tri-valence. 
For example, in \cref{fig:degree-two-vertices}, the dark green path connects vertices in $S_G$ and the light blue path connects vertices in $S_B$. The vertices in $S_G$ are not adjacent to a green edge, and hence we must add in green edges at these vertices, shown in \cref{fig:degree-four-vertices}. But this necessitates removing a green edge, and so on. The paths drawn in this figure are of length three: they each pass through two added edges and one removed edge.
Finding the shortest such path minimises the number of removed connections.

Now consider how the algorithm from \cref{ssec:superplaquette-algorithm} forms such a path: it constructs it from the edges around a shrink plaquette. 
The two vertices in $S_G$ $(S_B)$ will be separated by 3 edges that pass through the removed qubits.
Hence, the length of the path constructed by the algorithm will be $K - 3$ for a shrink plaquette of weight $K$. 
Hence the number of removed edges along the path will be $(K-4)/2$. 
In the honeycomb lattice, both shrink plaquettes will be weight 6, so the total number of additional removed edges (on top of the five adjacent to a removed qubit pair) will be 2.
In the 4.8.8 lattice we also achieve 2 additional removed edges by choosing a defect edge such that one shrink plaquette has weight $K_1 = 8$ and another has weight $K_2 = 4$ (see \cref{fig:488-green-merge,fig:488-blue-merge} in \Cref{app:alternative-codes}).
In the 4.6.12 lattice we can achieve only 1 additional removed edge ($K_1 = 6$, $K_2 = 4$).
In all cases it is obvious that there is no shorter path of edges (that does not pass through the removed qubits) that connects the vertices in $S_G$ and similarly for $S_B$. For example, in \cref{fig:degree-two-vertices}, a more optimal path would involve a direct connection between a pair of valence-2 vertices: no such edge can be formed. In the 4.8.8 case, a shorter path would imply a length-6 cycle of edges, which does not exist in the original lattice (assuming it is large enough that the non-trivial cycles are longer than six edges).
\end{proof}

In the above proof we have focused on uniform tilings on the plane, however this work can apply more generally to other surfaces as well. Below we remark how this optimality can also be shown for the 8.8.8 lattice, which is used in hyperbolic Floquet codes \cite{Higgott2023HyperbolicFloquetCodes, Fahimniya2023Hyperbolic}.

\begin{remark}
    Consider a defective qubit in the bulk of a Floquet code defined on an 8.8.8 lattice (without boundary), which is a particular uniform tiling of a hyperbolic surface. Under the constraint that we cannot form an edge between any two qubits not connected by an edge in the original lattice, the algorithm described in \cref{ssec:superplaquette-algorithm} removes the defective qubit with the minimal number of removed edges from the original lattice.
\end{remark}
\begin{proof}
This proof follows from the fact that the proof for \Cref{thm:minimal-edge-removal} does not rely on the lattice being on a plane. In the 8.8.8 lattice, the two shrink plaquettes will always be of weight-8, meaning that there will be 5 vertices in the paths connecting the vertices in $S_G$ and $S_B$, and a total of 4 edges are removed from the lattice (in addition to the five incident to the removed qubit pair). A more optimal path in the 8.8.8 lattice would require two degree-2 vertices being connected by a path shorter than length 5 (that does not pass through the removed qubits). This is not possible under the constraint that additional edges between non-connected vertices cannot be formed.
\end{proof}

Hence, our scheme is also optimal in terms of the connections that are removed from the lattice.
Note that we do not claim optimality for non-trivial combinations of defects, since then one must, for example, carefully choose the defect edges to achieve the optimal result. We do not have an optimal algorithm for such a choice, rather we use simplifying heursitics towards optimality (see \Cref{app:edge-selection}).
Moreover, it may be true that there exists a path between vertices in $S_G$ or $S_B$ that is shorter than that constructed from the corresponding shrink plaquette, if that shrink plaquette is very large-weight (see \cref{fig:defect-on-superplaquette-bad} in \Cref{app:edge-selection}).
We expect that finding optimality for non-trivial combinations of defects would not greatly improve code performance for small values of the defect probability.

It is worth noting that the method in Sec.~IIC of Ref.~\cite{aasen2023faulttolerant} does correct for a defective qubit on the honeycomb code while only removing one vertex, one edge, and no faces. This is achieved by violating the tri-valence requirement. The resulting Floquet code requires both extra connectivity and a more complex measurement schedule. This is an intriguing method, but the necessary modifications of the lattice and measurement schedule make its applicability to other Floquet codes, or to the planar honeycomb code, non-trivial. It is also unclear how to accommodate multiple defects using that method.

\section{Applicability to other Floquet codes}
\label{app:alternative-codes}

In this appendix, we discuss how the techniques described in \cref{sec:strategy} can be applied to other 2-D Floquet codes. Many 2-D Floquet codes can be defined on a colour code lattice, including the 4.8.8 Floquet code, the hyperbolic Floquet codes and the CSS Floquet code~\cite{Paetznick2023PerformancePlanarFloquetCodes, Higgott2023HyperbolicFloquetCodes, Fahimniya2023Hyperbolic,Davydova2023CSSFloquet,Kesselring2022AnyonCondensation}. These codes can be adapted to accommodate defective qubits by selecting a defect edge incident to the defective qubit, along with selecting the merge and shrink plaquettes in the same way as in Steps 2 and 3 as described in \cref{ssec:superplaquette-algorithm}, with the rest of the super-plaquette algorithm following naturally. An example of applying this approach to the 4.8.8 Floquet code is given in \cref{fig:488}. The only step that may change is Step 4, wherein we choose bases for the added edges. However, if one can make an identification of edge basis with edge colour at any given measurement sub-round (see, e.g., Ref.~\cite{Kesselring2022AnyonCondensation}), then this modification is simple. 
For example, in the CSS Floquet code, one measures operators along coloured edges in the cyclical pattern $R\rightarrow G \rightarrow B$ with the basis of the measurement operator alternating cyclically as $X \rightarrow Z$. Hence, if we use our algorithm to adapt this Floquet code, the measurement schedule remains unchanged since we can consistently colour the newly added edges in the lattice.

In cases where edge measurements are defined using an orientation rather than a colouring (see Ref.~\cite{Hastings2021dynamically}), one simply needs to define the orientation of the newly added lattice edges consistently such that the formed super-plaquette commutes with all edge measurements. This can involve edges of a mixed Pauli basis (e.g., $XY$ measurements), but the code adaptation in such cases is not complicated; just as in the original algorithm, there will only be one possible choice for these edge bases.
For the schedule of the modified code, we assume again that there is an identification between the edges measured in each sub-round and the colour of those edges (e.g., red $XX$ edge measurements performed in the first sub-round of the honeycomb code \cite{Gidney2021faulttolerant}). This ensures that qubits are only involved in one measurement during each sub-round. If this assumption holds, adapting the schedule to the new lattice is straight-forward.

\begin{figure}
    \centering
    \begin{subfigure}{0.3\linewidth}
        \includegraphics[width=\linewidth]{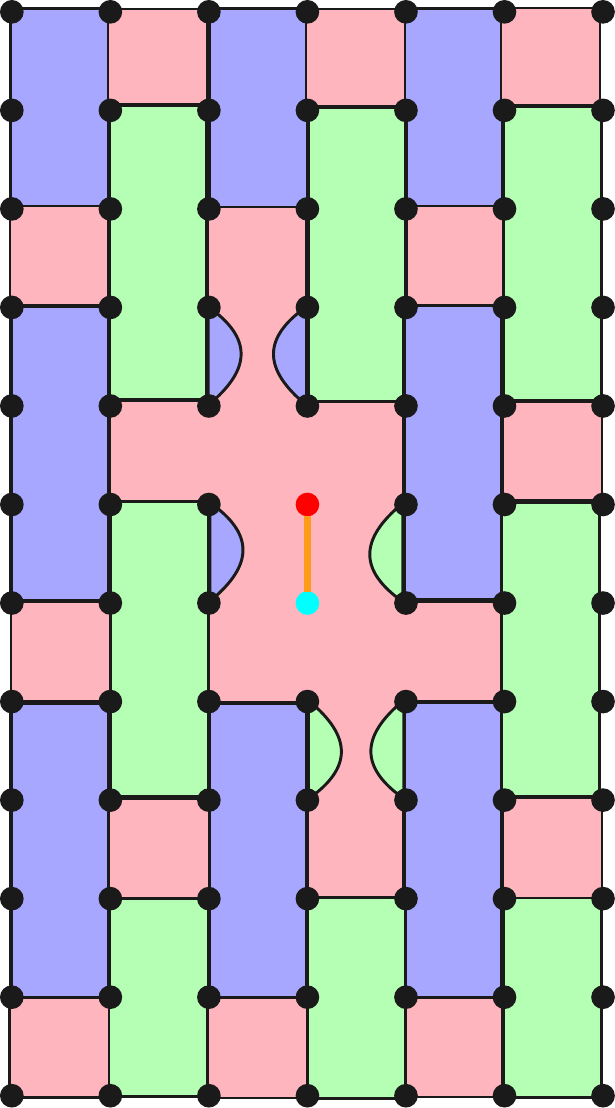}
        \caption{\label{fig:488-red-merge}}
    \end{subfigure}\hfill
    \begin{subfigure}{0.3\linewidth}
        \includegraphics[width=\linewidth]{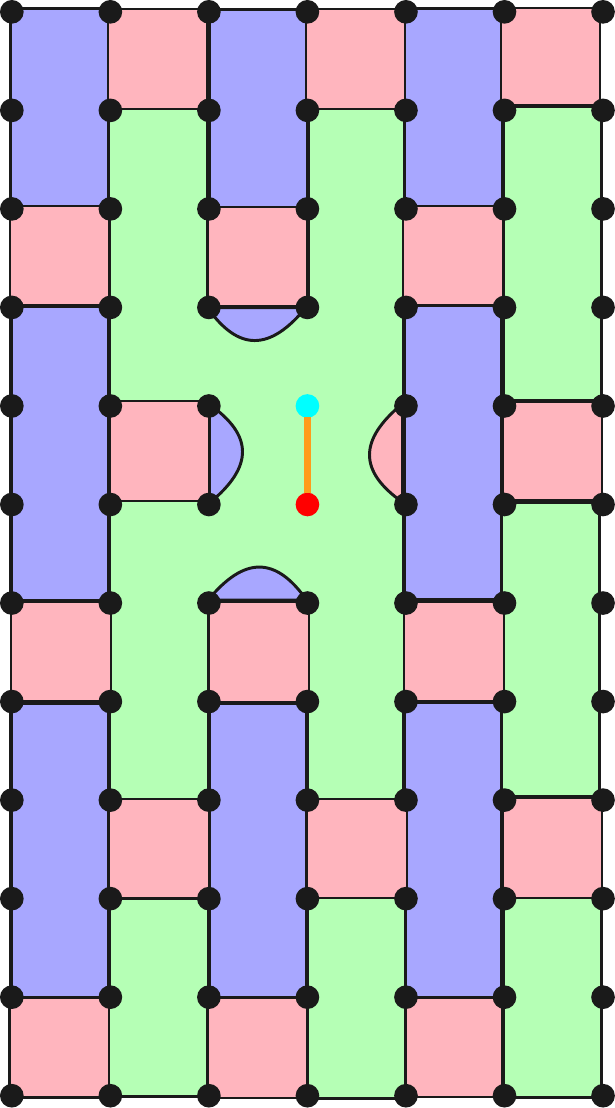}
        \caption{\label{fig:488-green-merge}}
    \end{subfigure}\hfill
    \begin{subfigure}{0.3\linewidth}
        \includegraphics[width=\linewidth]{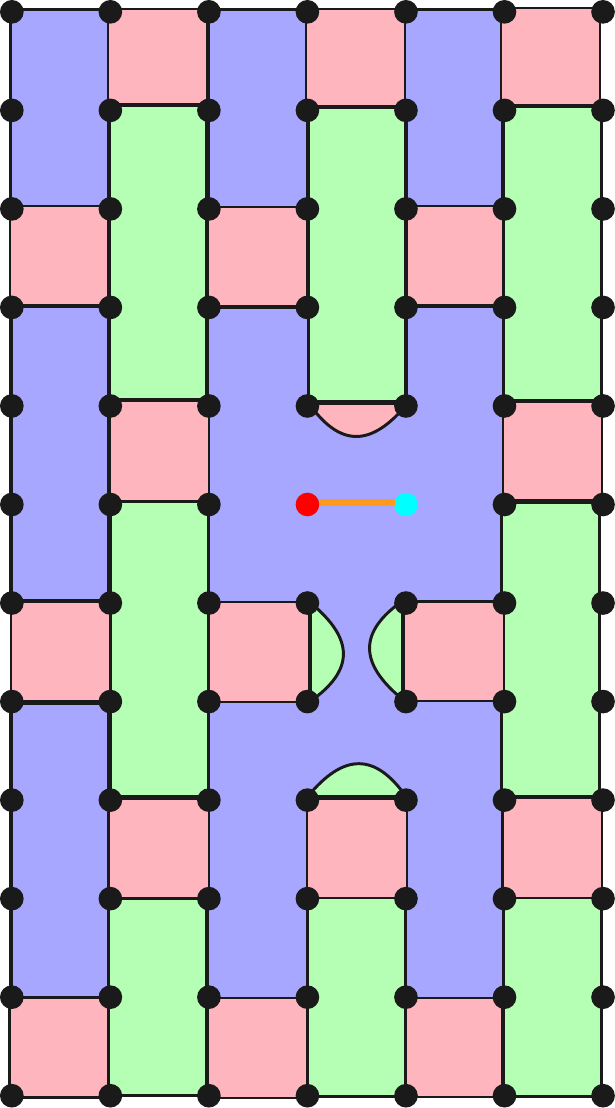}
        \caption{\label{fig:488-blue-merge}}
    \end{subfigure}
    \caption{\label{fig:488} Applying the algorithm in \cref{ssec:superplaquette-algorithm} to remove a (a) red, (b) green, and (c) blue defect edge incident to a defective qubit in the bulk of the 4.8.8 Floquet code.}
    
\end{figure}

Even if the bulk of the lattice does not change there are different ways of implementing the boundaries of Floquet codes. In this paper we have focused on the boundary method described by Gidney, Newman and McEwen \cite{Gidney2022benchmarkingplanar}. However, there is also an alternative method for implementing boundaries used in Refs.~\cite{Haah2022boundarieshoneycomb, Paetznick2023PerformancePlanarFloquetCodes}, by introducing additional weight-2 plaquettes. The strategy can be similarly employed to accommodate defects on these instances of the planar honeycomb code through the same process as before: embedding the patch in the bulk of a larger honeycomb code, removing the defects via the algorithm in \cref{ssec:superplaquette-algorithm}, and then re-introducing the boundaries to the code.

\section{Details on defect edge selection}
\label{app:edge-selection}

We will now provide further details on our choice of heuristics for the defect edge selection and why they are beneficial. As a reminder, our aim with these heuristics is to keep the super-plaquettes small and in the bulk of the code where possible.

The first and simplest heuristic is that if two defective qubits are adjacent to each other, it is worth choosing the defect edge to be the edge connecting them. This allows us to remove two defective qubits with the same overhead as just removing a single defect.

Our second heuristic concerns defective qubits which lie on already-formed super-plaquettes. An example of such a defect is provided in \cref{fig:defect-on-superplaquette-initial}. In \cref{fig:defect-on-superplaquette-good}, the red edge incident to the defective qubit is selected, meaning that the super-plaquette is selected as a merge plaquette. If instead the green edge incident to the defective qubit is selected, as in \cref{fig:defect-on-superplaquette-bad}, the original super-plaquette is selected as a shrink plaquette. The result is a green super-plaquette which is larger than the red super-plaquette formed in \cref{fig:defect-on-superplaquette-good}, as well as several pairs of functioning qubits which are now disconnected from the rest of the code. As a result, if a defective qubit is on an already-formed super-plaquette we choose the defect edge to be the same colour as said super-plaquette.

\begin{figure}
    \centering
    \begin{subfigure}{0.3\linewidth}
        \includegraphics[width=\linewidth]{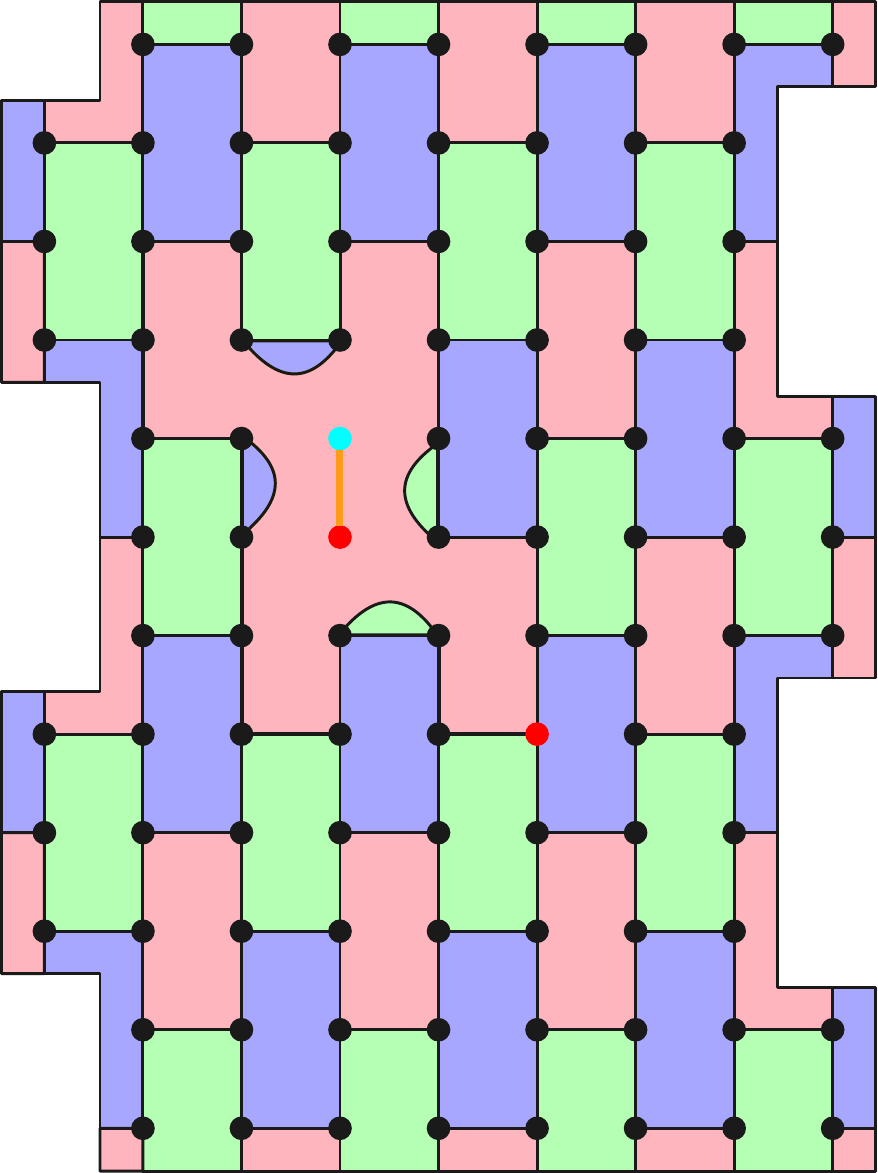}
        \caption{\label{fig:defect-on-superplaquette-initial}}
    \end{subfigure}\hfill
    \begin{subfigure}{0.3\linewidth}
        \includegraphics[width=\linewidth]{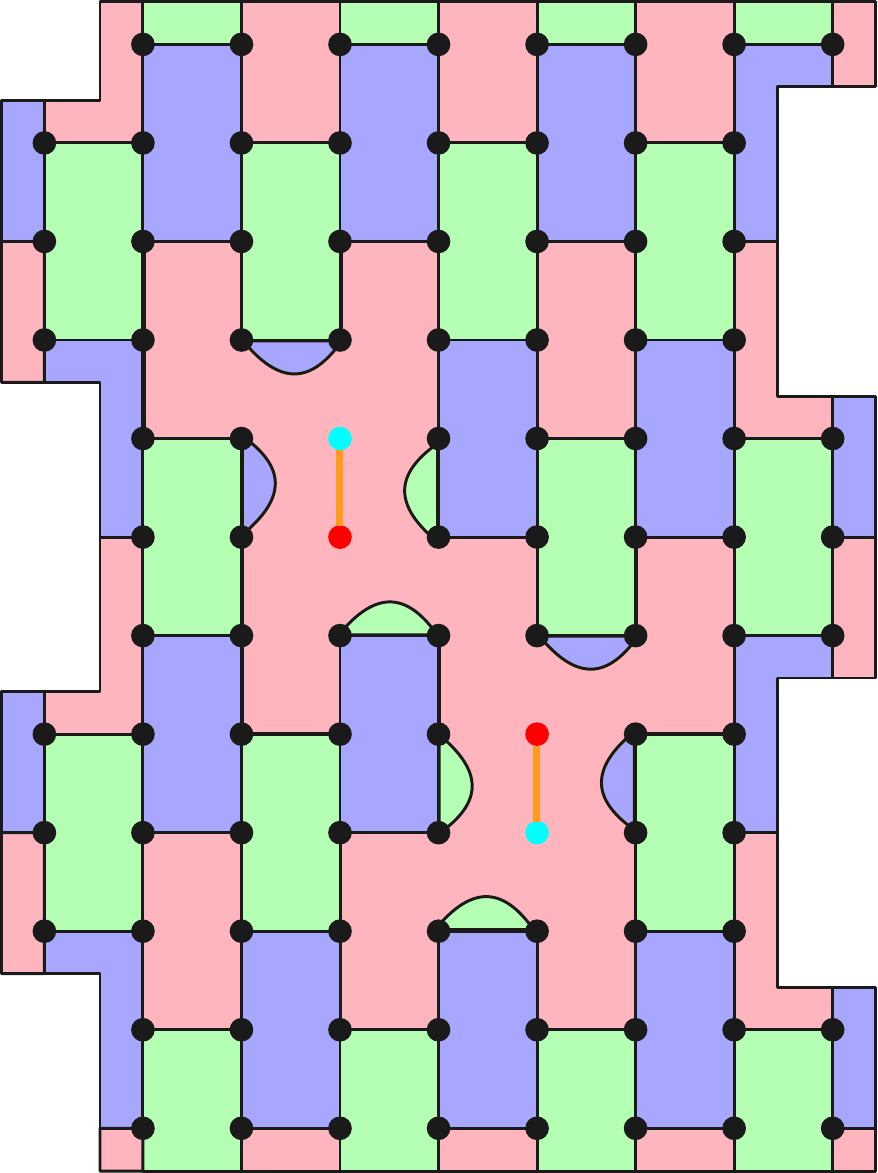}
        \caption{\label{fig:defect-on-superplaquette-good}}
    \end{subfigure}\hfill
    \begin{subfigure}{0.3\linewidth}
            \includegraphics[width=\linewidth]{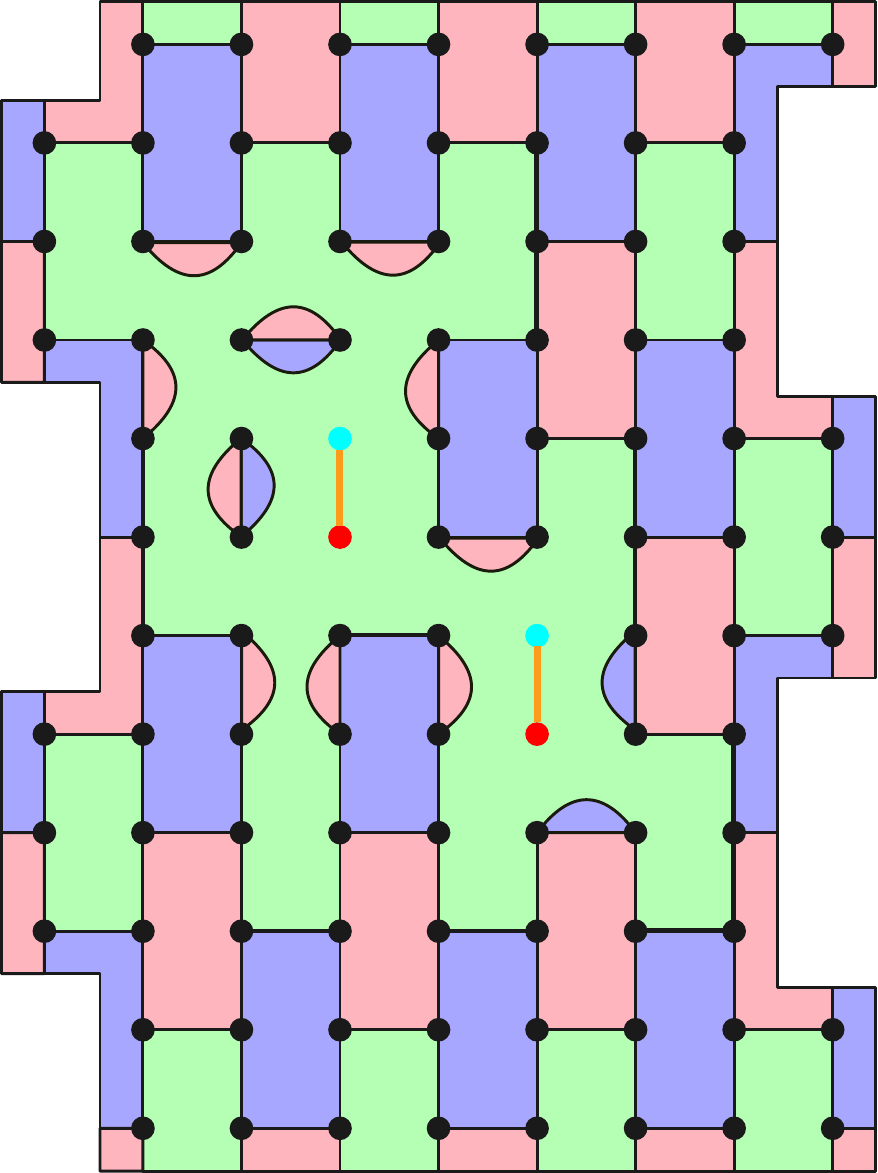}
        \caption{\label{fig:defect-on-superplaquette-bad}}
    \end{subfigure}
    \caption{\label{fig:defect-on-superplaquette} Removing a defective qubit which is on an already-formed super-plaquette. (a) The initial layout, with the left defective qubit already removed, and the right defective qubit on the resulting super-plaquette. (b) The defect edge is selected to be an edge that connects two red plaquettes, matching the colour of the already formed super-plaquette. (c) The defect edge is selected to be an edge that connects two green plaquettes, meaning that the previously formed red super-plaquette now needs to be shrunk.}
\end{figure}

Our next heuristic concerns defective qubits which lie near but not directly on already-formed super-plaquettes. In \cref{fig:defect-on-support-initial}, a defective qubit is directly below an already-formed super-plaquette. In \cref{fig:defect-on-support-good}, we have selected the defect edge to not be incident to the already-formed super-plaquette, resulting in two super-plaquettes being formed in blue and green. Compare with \cref{fig:defect-on-support-bad}, where the defect edge selected is incident to the already-formed super-plaquette. This results in the blue super-plaquette being selected as a merge plaquette, leading to a single larger super-plaquette. A small number of larger super-plaquettes will reduce the code distance faster than a large number of smaller super-plaquettes, hence we prioritise selecting edges which are not incident to already-formed super-plaquettes.

\begin{figure}
    \centering
    \begin{subfigure}{0.3\linewidth}
        \includegraphics[width=\linewidth]{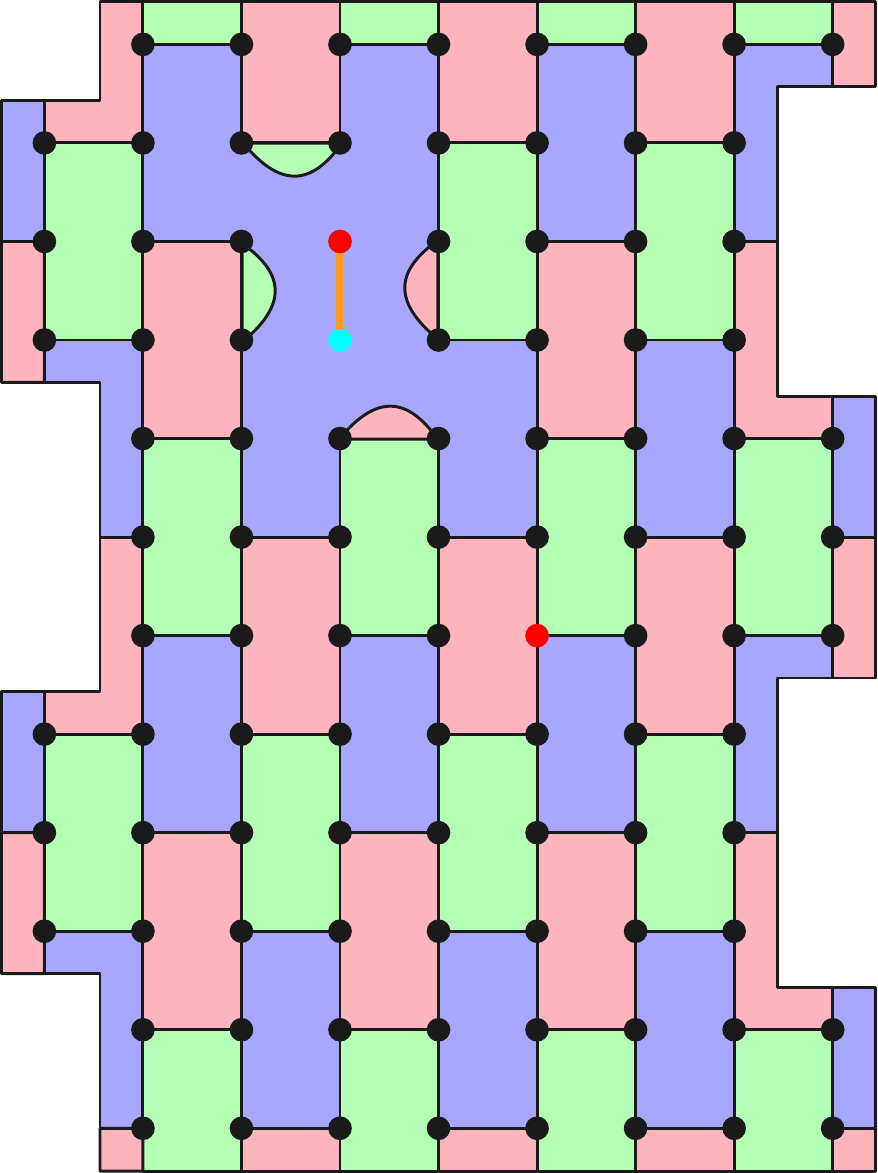}
        \caption{\label{fig:defect-on-support-initial}}
    \end{subfigure}\hfill
    \begin{subfigure}{0.3\linewidth}
        \includegraphics[width=\linewidth]{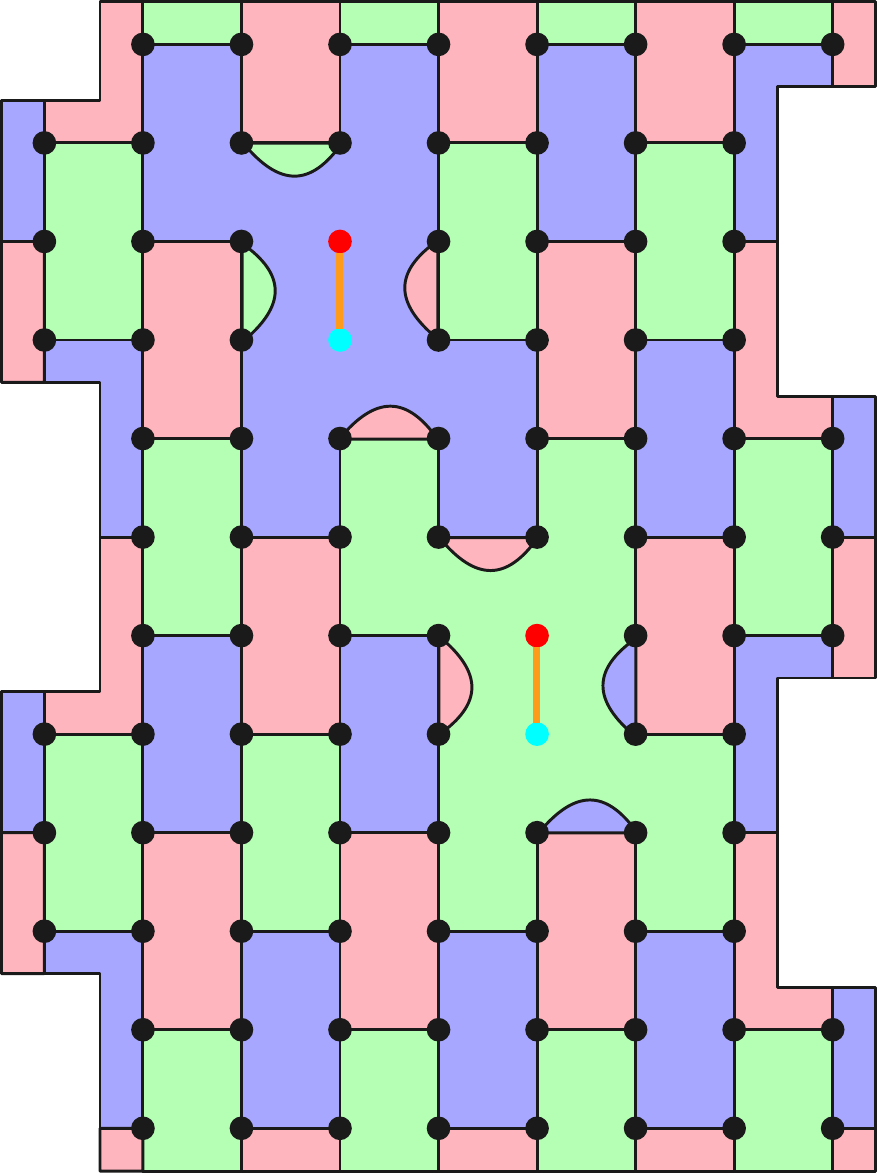}
        \caption{\label{fig:defect-on-support-good}}
    \end{subfigure}\hfill
    \begin{subfigure}{0.3\linewidth}
            \includegraphics[width=\linewidth]{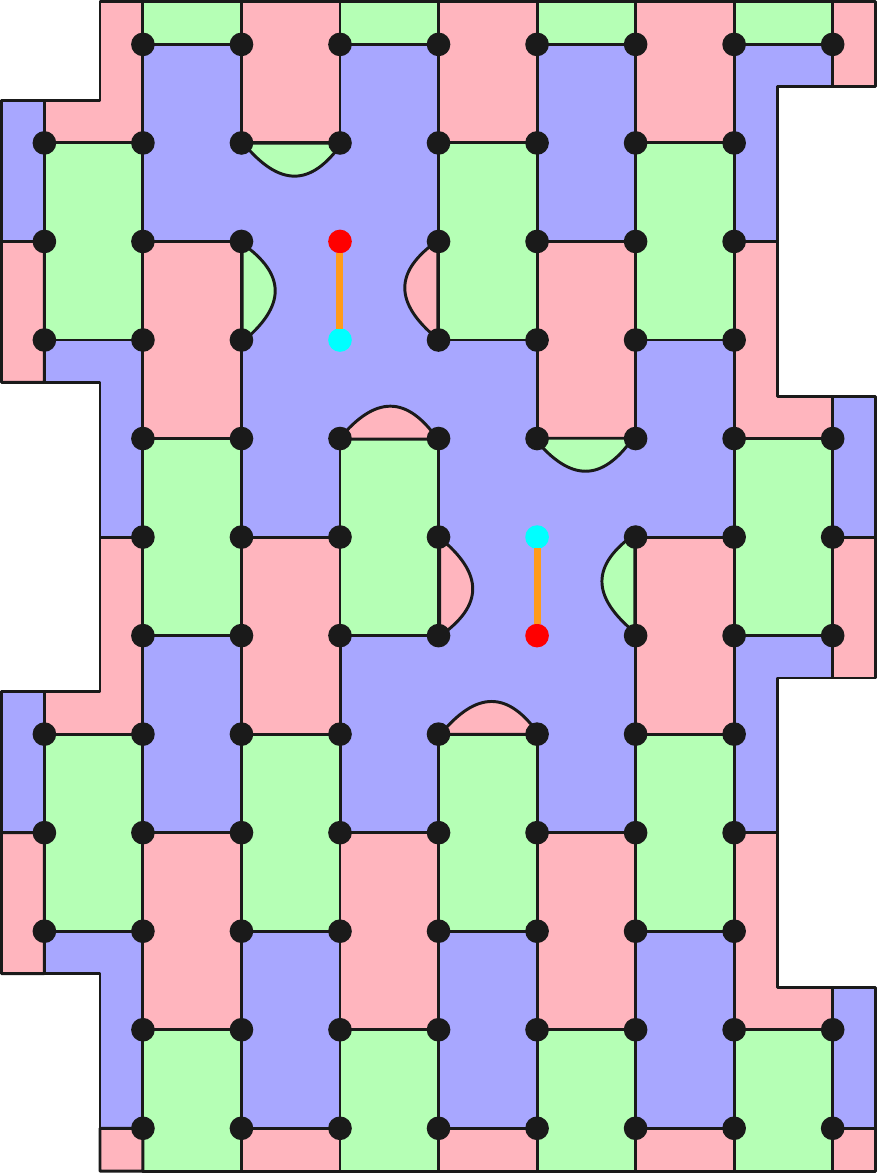}
        \caption{\label{fig:defect-on-support-bad}}
    \end{subfigure}
    \caption{\label{fig:defect-on-support} Removing a defective qubit near an already-formed super-plaquette. (a) The initial layout, with the top defective qubit already removed from the code. (b) A defect edge is selected which is not incident to the already-formed super-plaquette, leading to one blue super-plaquette and one green super-plaquette being formed. (c) A defect edge is selected which is incident to the already-formed super-plaquette, resulting in a single larger blue super-plaquette.}
\end{figure}

If none of the previous cases apply, we consider the orientation of the defect edge. \Cref{fig:edge-orientation} shows two examples of removing defective qubits by selecting an incident red edge as the defect edge. In \cref{fig:edge-orientation-vertical}, the red edge is a vertical edge, whereas in \cref{fig:edge-orientation-horizontal}, the red edge is a horizontal edge. Also shown in red are edges which denote the lowest-weight horizontal and vertical logical errors: if two-qubit $ZZ$ or $YY$ errors were to occur on these edges after measuring the red edges, the result could be a logical error (see \cref{fig:example-honeycomb} for the positions in the measurement schedule at which these errors would be logical). We can note that the horizontal Pauli distance is reduced by one in both cases. However, if a vertical edge is selected, as in \cref{fig:edge-orientation-vertical}, the vertical Pauli distance is only reduced by one. On the other hand, if a horizontal edge is selected, as in \cref{fig:edge-orientation-horizontal}, the vertical Pauli distance is reduced by two. This is because when a horizontal defect edge is selected the merge plaquettes align with the vertical logical operators, whereas if a vertical defect edge is selected the merge plaquettes lie diagonally across the code. As a result, we prefer selecting vertical edges over horizontal edges where possible.

\begin{figure}
    \centering
    \begin{subfigure}{0.3\linewidth}
        \includegraphics[width=\linewidth]{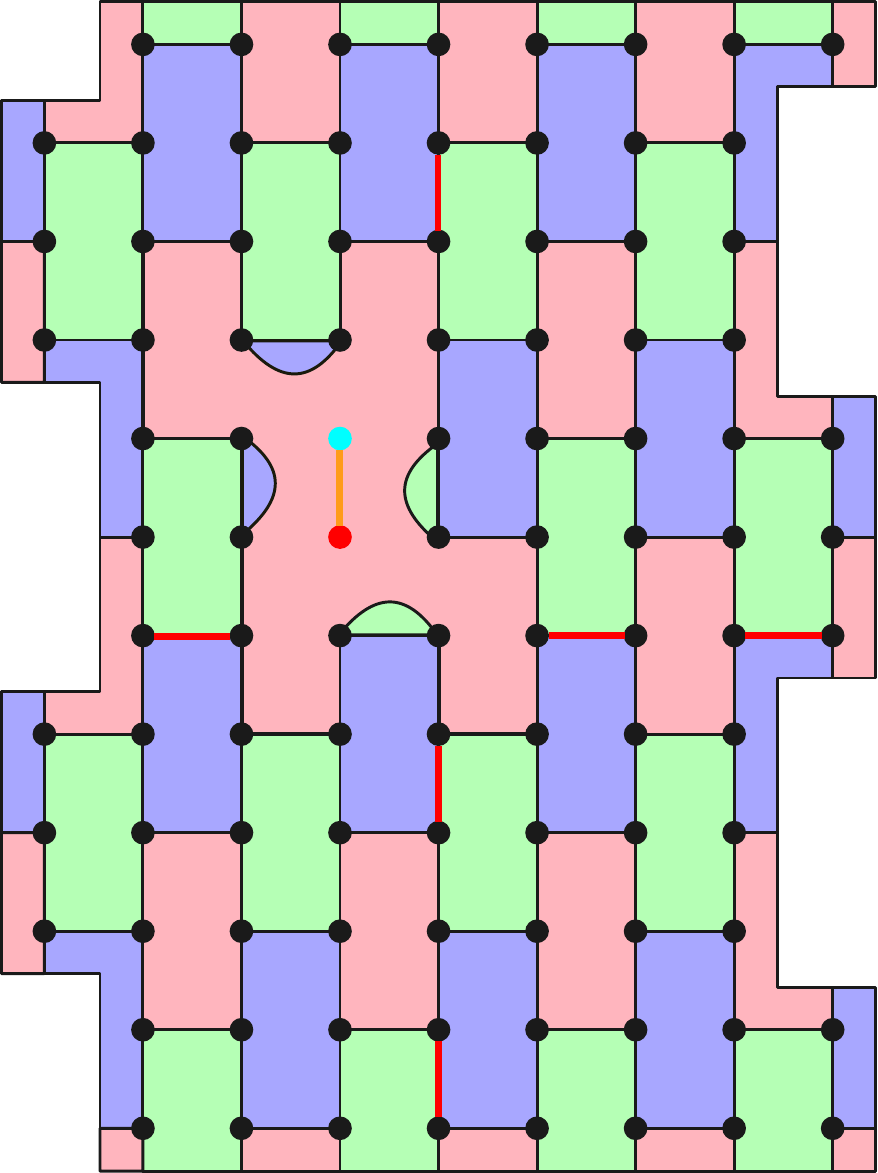}
        \caption{\label{fig:edge-orientation-vertical}}
    \end{subfigure}\hspace{20pt}
    \begin{subfigure}{0.3\linewidth}
            \includegraphics[width=\linewidth]{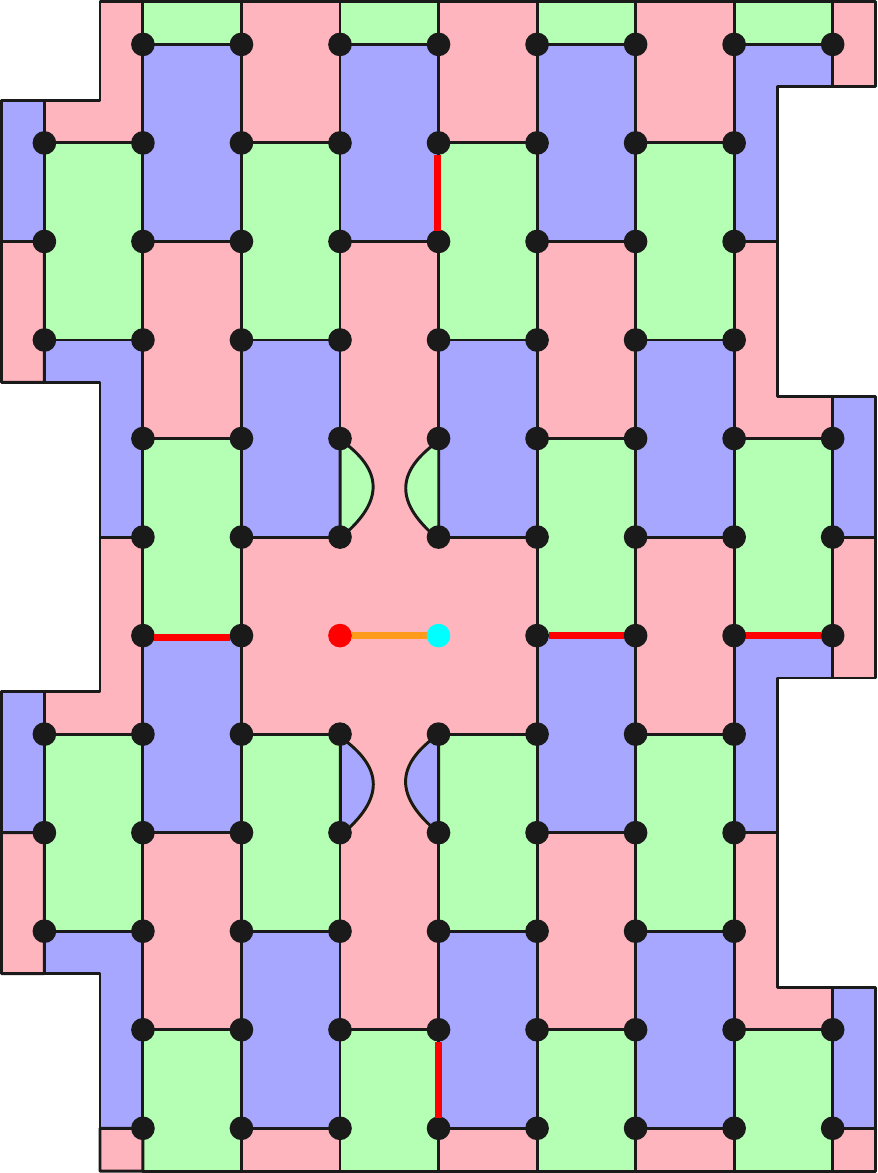}
        \caption{\label{fig:edge-orientation-horizontal}}
    \end{subfigure}
    \caption{\label{fig:edge-orientation} Super-plaquettes formed by selecting a (a) vertical or (b) horizontal defect edge. Red lines denote two-qubit Pauli errors which, depending on the basis of the error and the measurement sub-round, could cause a logical error.}
\end{figure}

Our final heuristic concerns defects which are close to boundaries. \Cref{fig:boundary-heuristic} illustrates two examples of removing a defect close to a horizontal boundary of the honeycomb code. In \cref{fig:boundary-heuristic-deterministic}, the defect edge is selected such that the resulting super-plaquette is in the bulk of the code. As this super-plaquette is in the bulk of the code, it commutes with the measurements in every round and is therefore deterministically measured twice every six sub-rounds. In comparison, \cref{fig:boundary-heuristic-nondeterministic} shows a defect edge selected such that the super-plaquette lies on the boundary of the code. The single-qubit measurements on the boundary anti-commute with the stabiliser formed by this plaquette, meaning that for this stabiliser is deterministically measured only once every six sub-rounds. To avoid having large non-deterministic stabiliser measurements, we prioritise selecting defect edges which are pointing away from the nearest code boundaries where possible.

\begin{figure}
    \centering
    \begin{subfigure}{0.3\linewidth}
        \includegraphics[width=\linewidth]{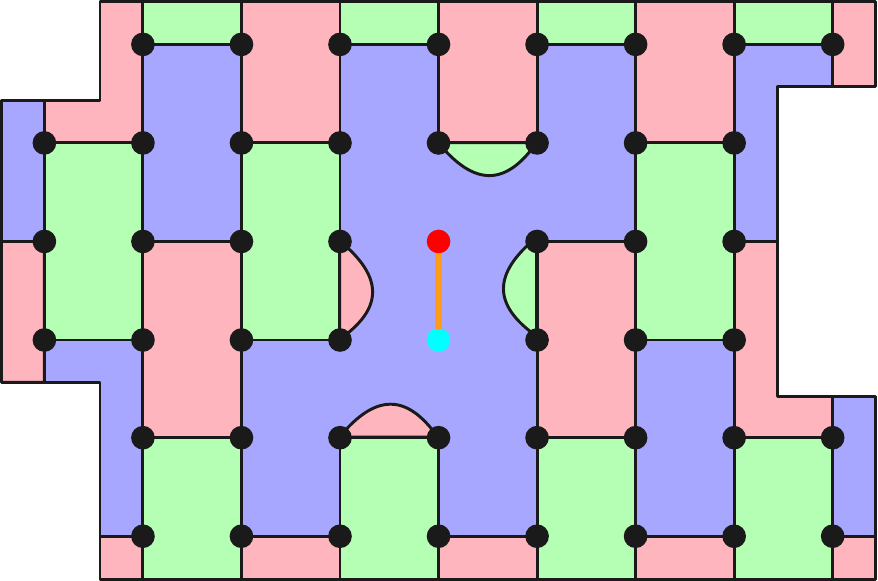}
        \caption{\label{fig:boundary-heuristic-deterministic}}
    \end{subfigure}
    \begin{subfigure}{0.3\linewidth}
            \includegraphics[width=\linewidth]{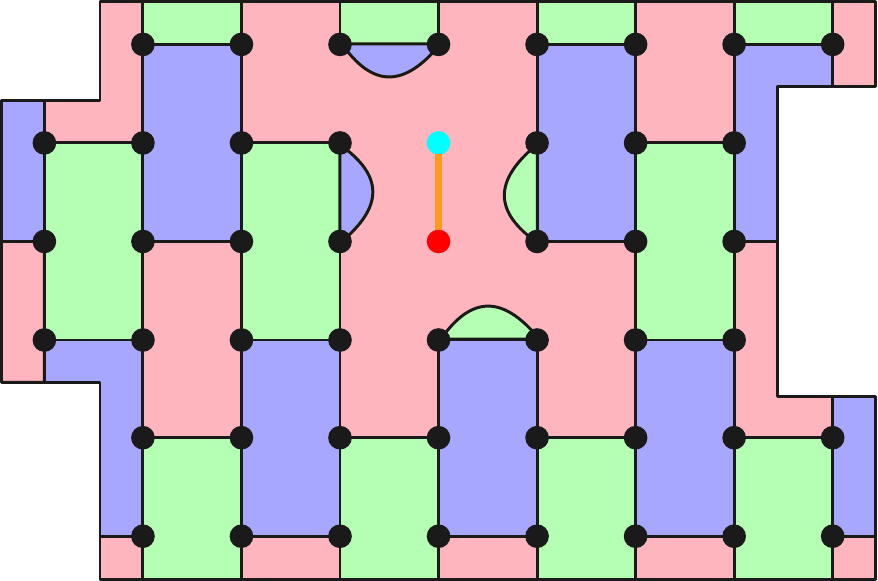}
        \caption{\label{fig:boundary-heuristic-nondeterministic}}
    \end{subfigure}
    \caption{\label{fig:boundary-heuristic} Removing a defective qubit close to the boundary. (a) The defect edge is pointing away from the boundary, resulting in a super-plaquette formed in the bulk of the code. (b) The defect edge is pointing towards the boundary, resulting in a super-plaquette on the boundary which is non-deterministic during certain QEC sub-rounds.}
\end{figure}

\section{Evolution of the logical operators}
\label{app:observable-evolution}

As mentioned in \cref{ssec:logical-observables}, the logical operators are multi-qubit Pauli operators along non-trivial strings in the lattice. If a qubit or edge along a string is removed because of the algorithm in \cref{ssec:superplaquette-algorithm}, a path-finding algorithm can be used to find an alternative string.

An important aspect of Floquet codes is that with each sub-round of measurements, the results of measurements between pairs of qubits on an observable string are multiplied in to the result of the logical observable \cite{Hastings2021dynamically, Gidney2021faulttolerant}. This is to ensure that the logical observable commutes with each subsequent QEC sub-round.

\Cref{fig:periodic-observables} shows how the operator strings evolve around a defective qubit in the periodic honeycomb code. This allows us to verify that during each QEC sub-round the logical operators commute with the measurements in the next sub-round and anti-commute with each other.

\begin{figure}
    \centering
    \includegraphics[width=\linewidth]{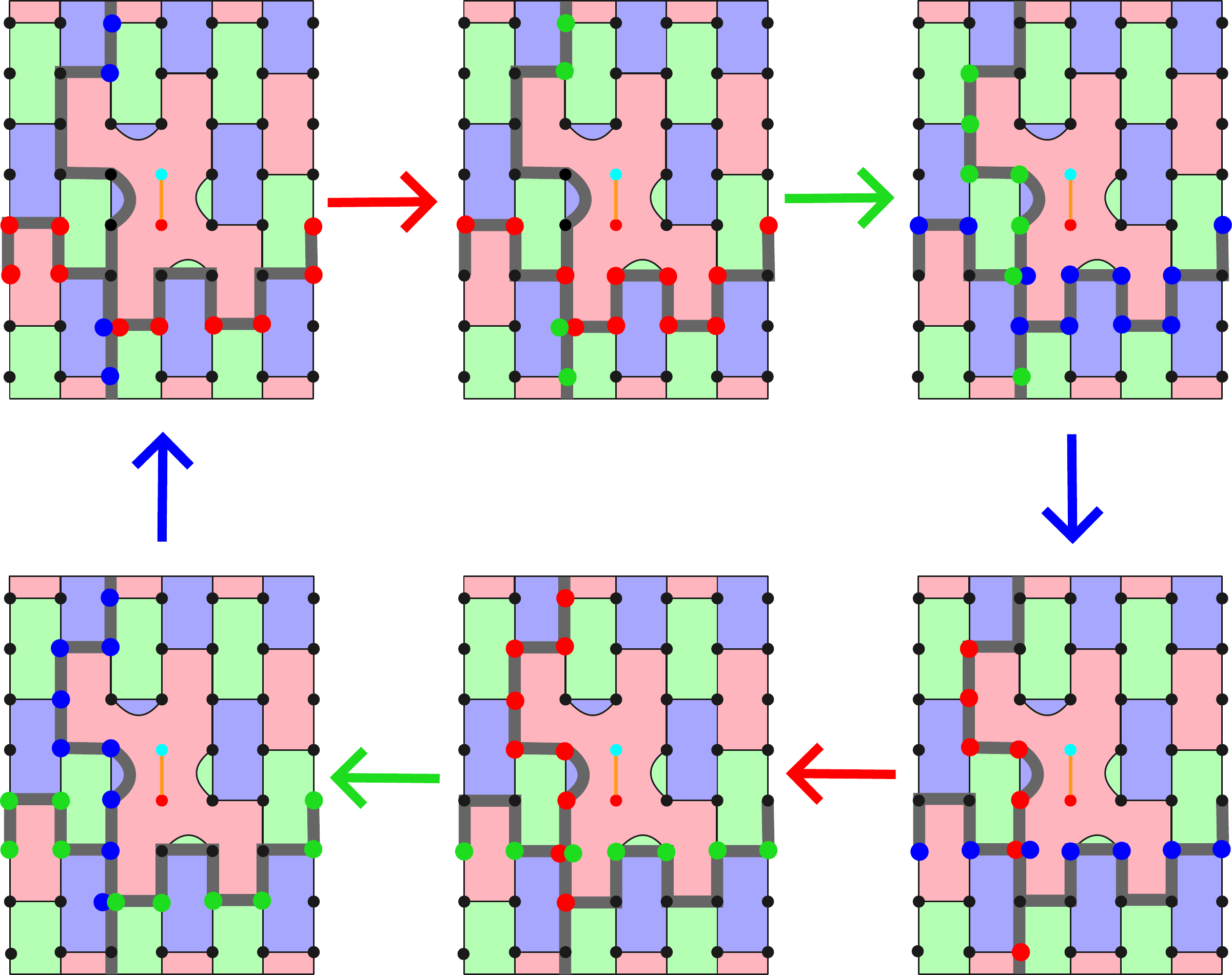}
    \caption{\label{fig:periodic-observables} Evolution of the logical operators around a defective qubit in the periodic honeycomb code. White lines denote the logical strings. The vertical string uses a newly-added edge incident to a weight-2 plaquette, whilst the horizontal string uses an edge that existed in the original code. Circles denote qubits that form the logical operators, with the colour indicating basis. Arrows denote the next round of measurements (red for $X$ edges, green for $Y$ and blue for $Z$). Note that there are two other logical operators along the same paths but offset by three sub-rounds, these have not been drawn here for simplicity.}
\end{figure}

The operators in the honeycomb code with planar boundary conditions evolve differently compared to the honeycomb code with periodic boundary conditions, as the introduction of boundaries leads to measurements which anti-commute with the observable if the periodic schedule is followed \cite{Hastings2021dynamically}. To avoid this, two modifications are made to the code: first, the measurement schedule is modified to become a period-6 schedule of $R \rightarrow G \rightarrow B \rightarrow R \rightarrow B \rightarrow G$; and second, the results of red measurements are not multiplied in to the logical operators \cite{Haah2022boundarieshoneycomb, Gidney2022benchmarkingplanar, Paetznick2023PerformancePlanarFloquetCodes}. \Cref{fig:planar-observables} illustrates how the logical operators evolve on a planar honeycomb code with a defective qubit on the boundary. Again, we can note that the operators commute with each QEC sub-round, and that the two operators anti-commute with each other.

\begin{figure}
    \centering
    \includegraphics[width=\linewidth]{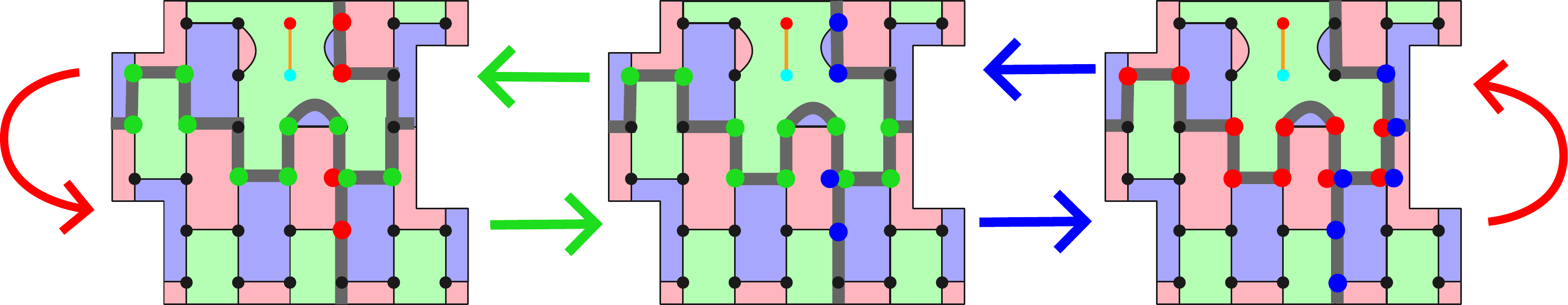}
    \caption{\label{fig:planar-observables} Evolution of the logical operators around a defective qubit in the planar honeycomb code. The horizontal string uses a newly-added edge incident to a weight-2 plaquette, whilst the vertical string uses an edge that existed in the original code. Measurements along red edges do not contribute to the logical operators in the planar case.}
\end{figure}

It is curious to note in \cref{fig:periodic-observables,fig:planar-observables} that when the strings are incident to a weight-2 plaquette there is a choice in which edge to include, as both edges connect the same pair of qubits. For instance, in \cref{fig:periodic-observables} (resp.\ \cref{fig:planar-observables}), the vertical (resp.\ horizontal) logical operator includes one of the newly added edges. As long as measurement results from only one of the two edges is included in the operator, the logical qubit can be preserved. Indeed, the two operator strings that can be formed based on edge choice are equivalent up to the weight-2 stabiliser. This shows that the path-finding algorithm can use both edges in the original code and newly-added edges when constructing the new logical strings.

\section{Additional numerical results}\label{app:additional_numerics}

\begin{figure}
\centering
\begin{subfigure}{0.24\textwidth}
    \includegraphics[width=\textwidth]{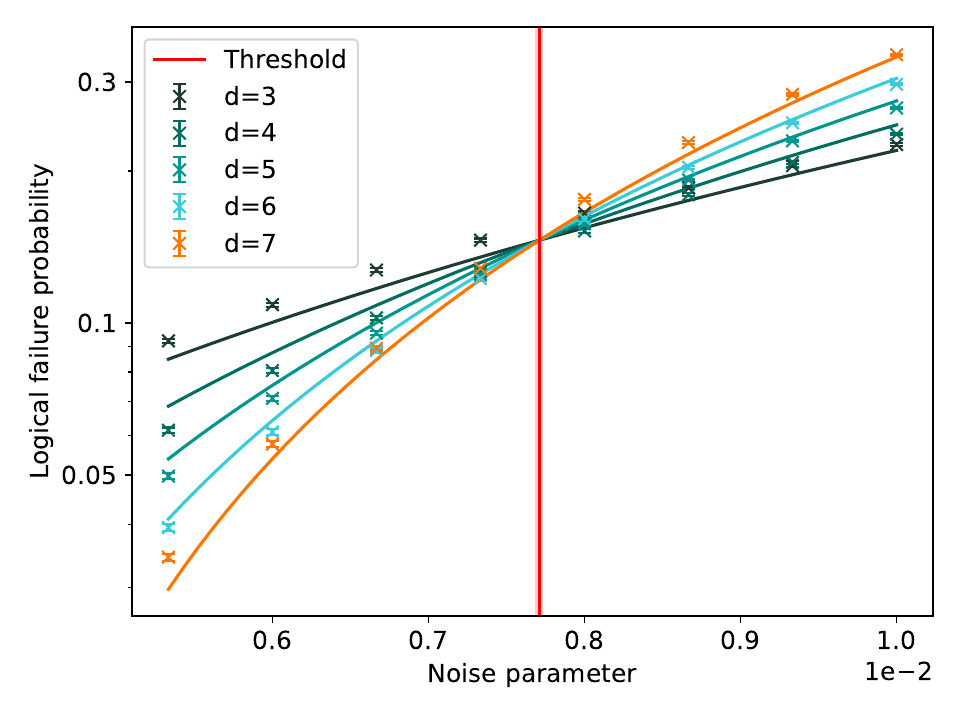}
    \caption{}
\end{subfigure}\hfill
\begin{subfigure}{0.26\textwidth}
    \includegraphics[width=\textwidth]{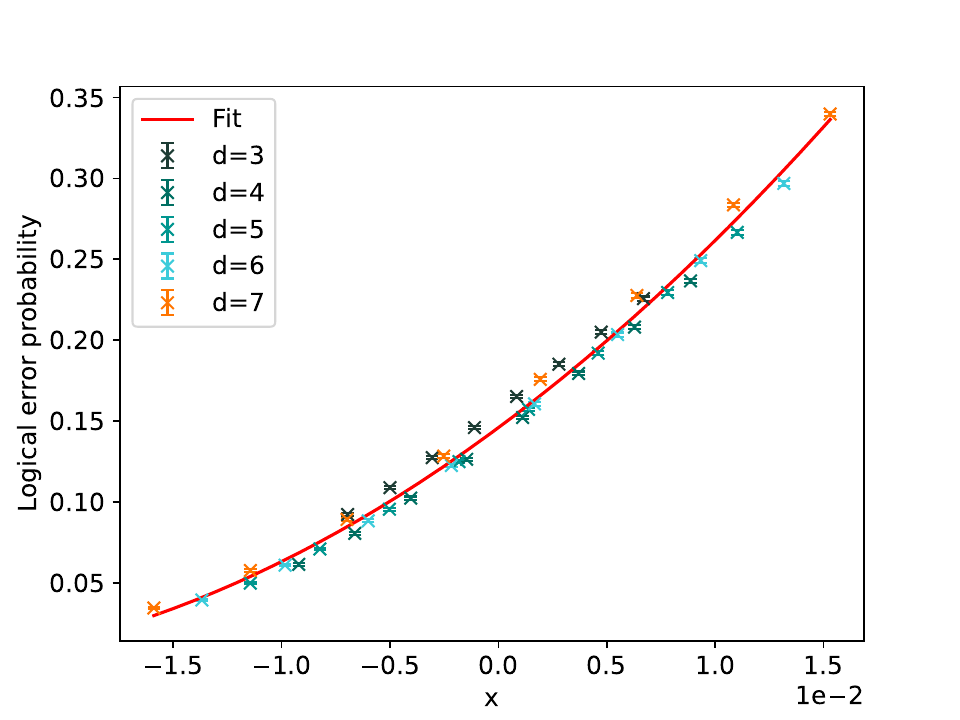}
    \caption{}
\end{subfigure}\hfill
\begin{subfigure}{0.24\textwidth}
    \includegraphics[width=\textwidth]{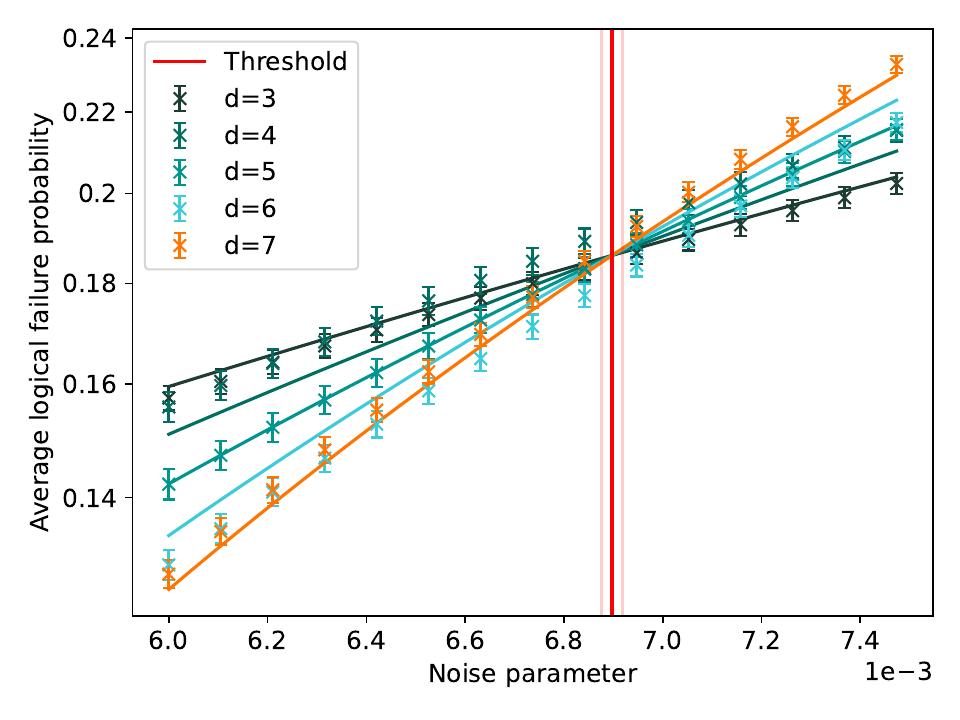}
    \caption{}
\end{subfigure}\hfill
\begin{subfigure}{0.26\textwidth}
    \includegraphics[width=\textwidth]{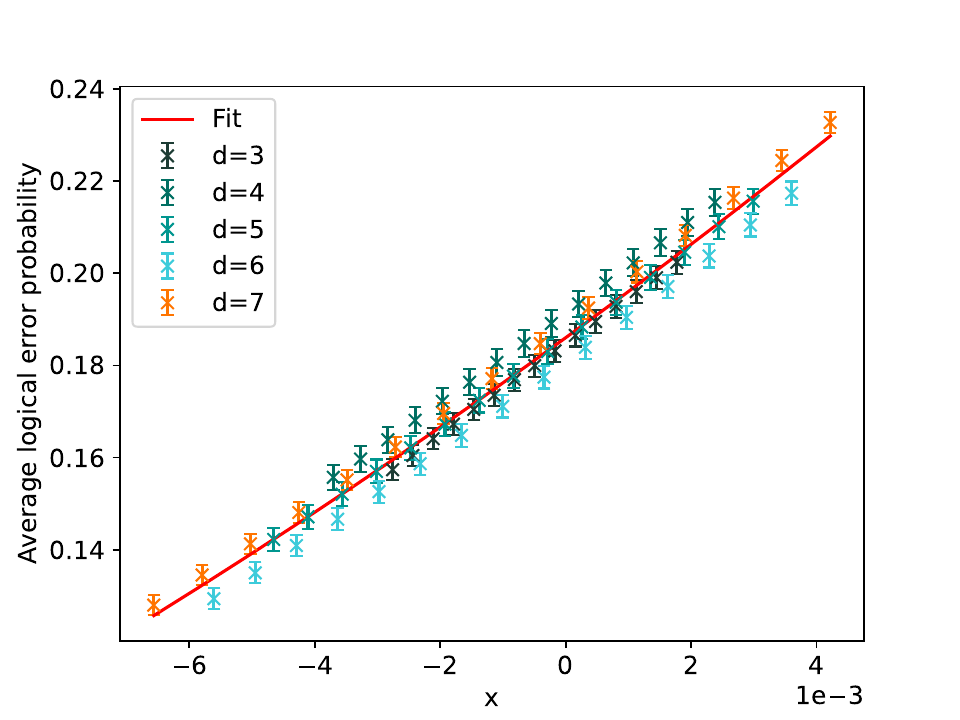}
    \caption{}
\end{subfigure}\\
\begin{subfigure}{0.24\textwidth}
    \includegraphics[width=\textwidth]{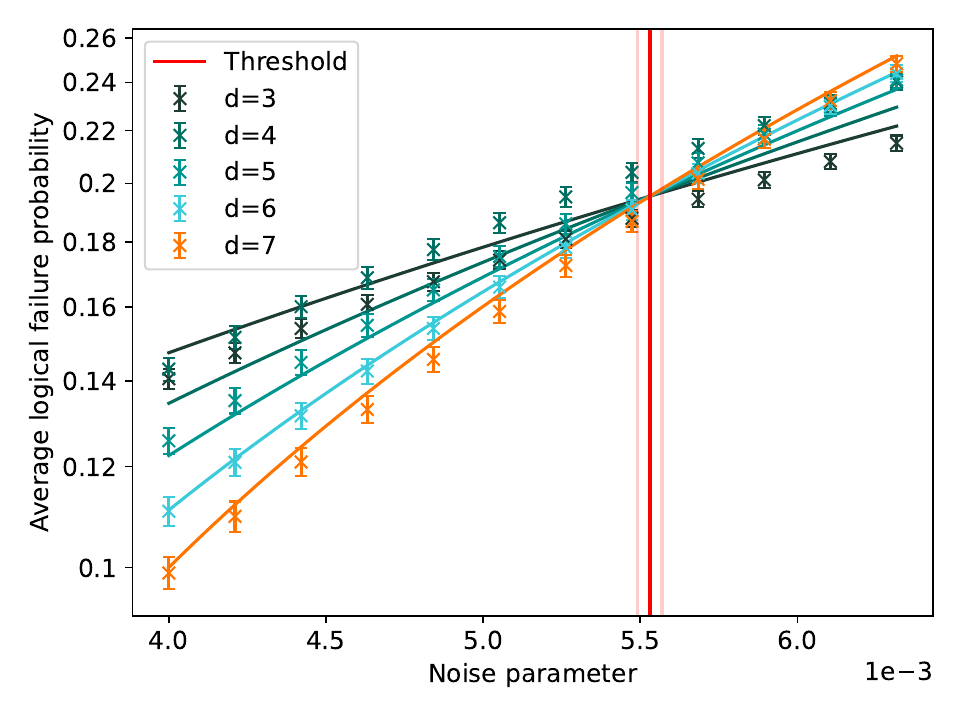}
    \caption{}
\end{subfigure}\hfill
\begin{subfigure}{0.26\textwidth}
    \includegraphics[width=\textwidth]{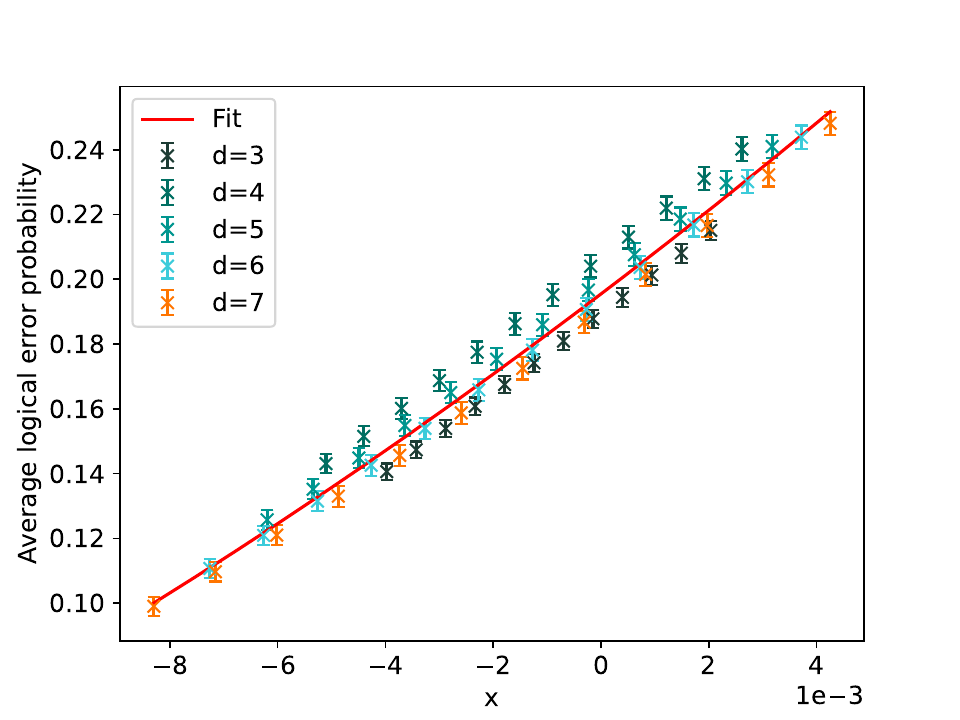}
    \caption{}
\end{subfigure}\hfill
\begin{subfigure}{0.24\textwidth}
    \includegraphics[width=\textwidth]{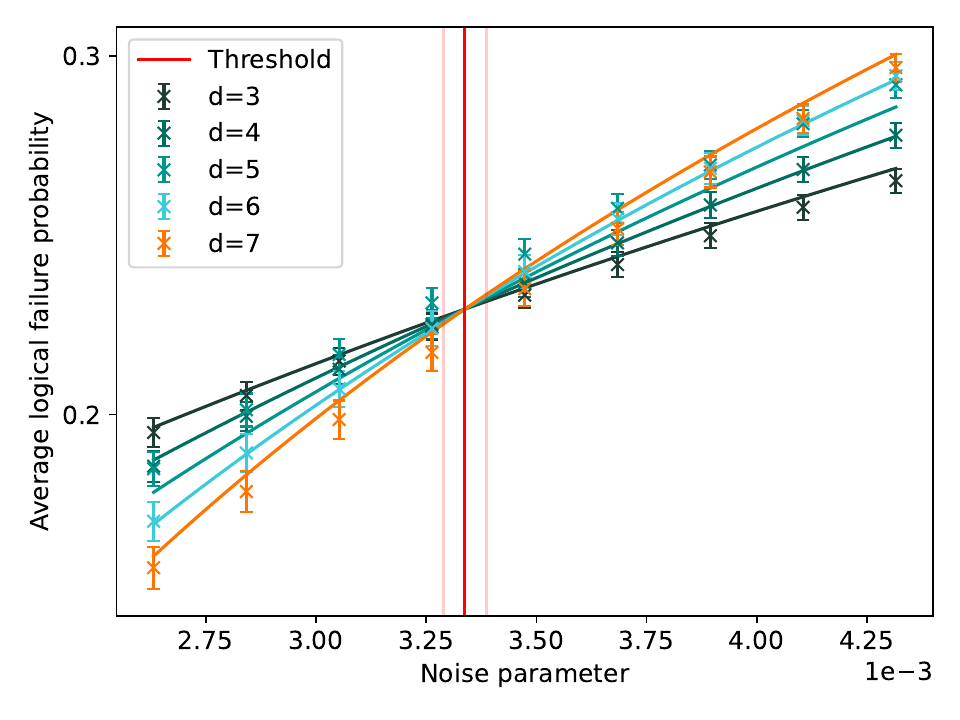}
    \caption{}
\end{subfigure}\hfill
\begin{subfigure}{0.26\textwidth}
    \includegraphics[width=\textwidth]{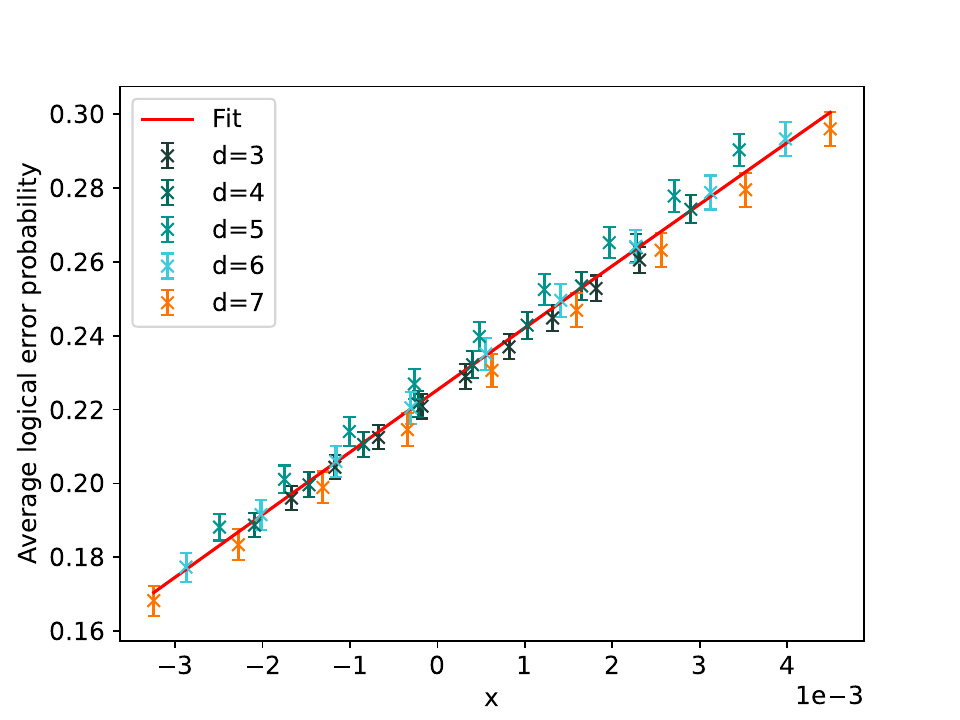}
    \caption{}
\end{subfigure}\\
\begin{subfigure}{0.24\textwidth}
    \includegraphics[width=\textwidth]{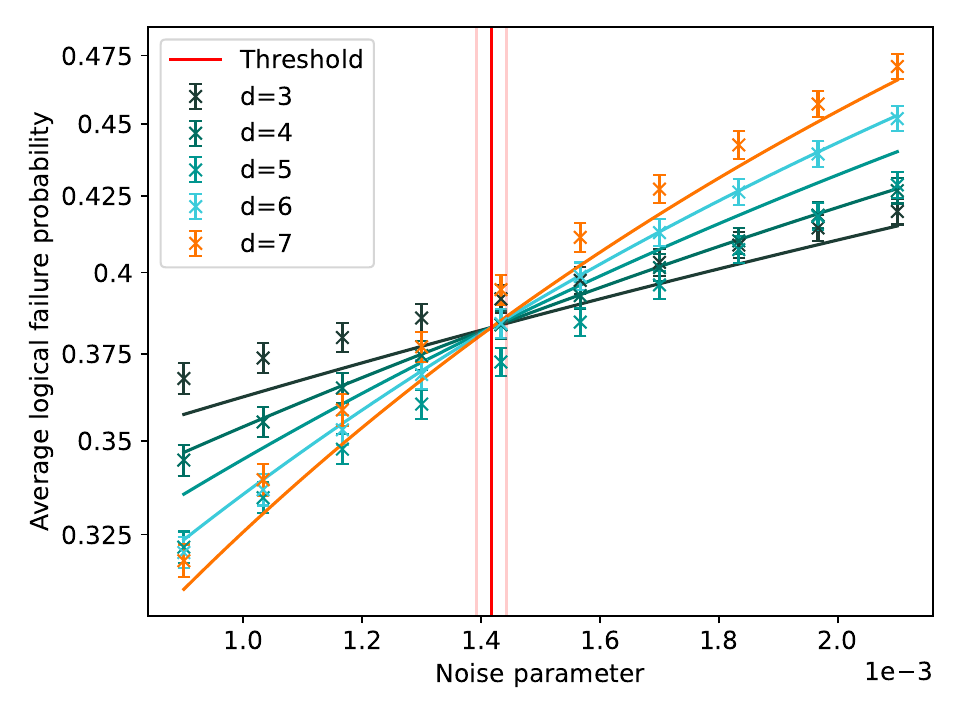}
    \caption{}
\end{subfigure}\hfill
\begin{subfigure}{0.26\textwidth}
    \includegraphics[width=\textwidth]{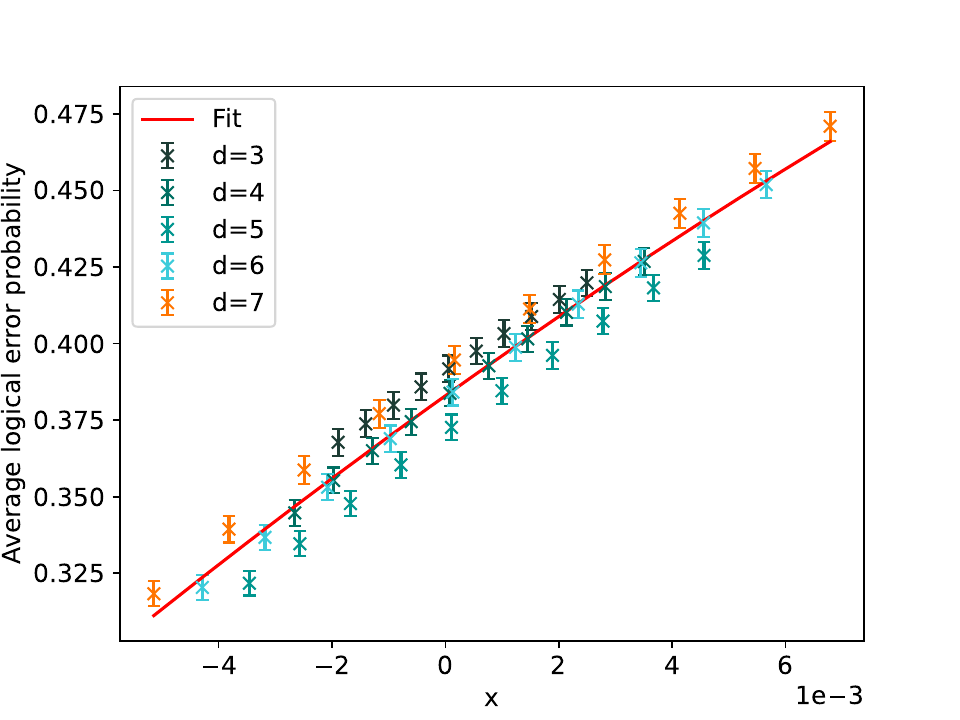}
    \caption{}
\end{subfigure}\hfill
\begin{subfigure}{0.24\textwidth}
    \includegraphics[width=\textwidth]{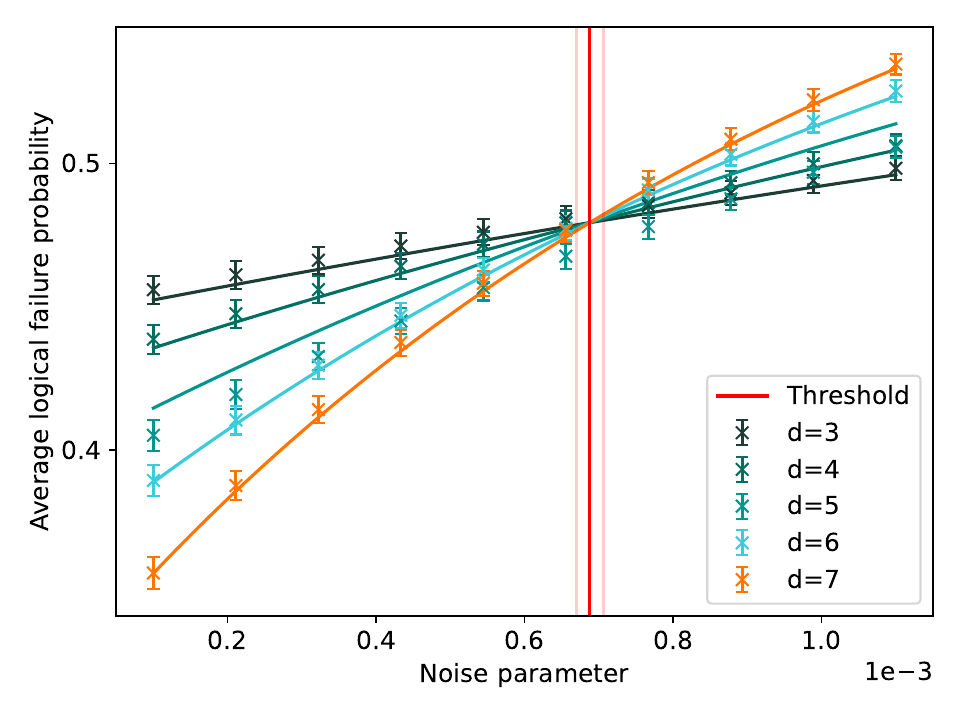}
    \caption{\label{subfig:10_perc_noise}}
\end{subfigure}\hfill
\begin{subfigure}{0.26\textwidth}
    \includegraphics[width=\textwidth]{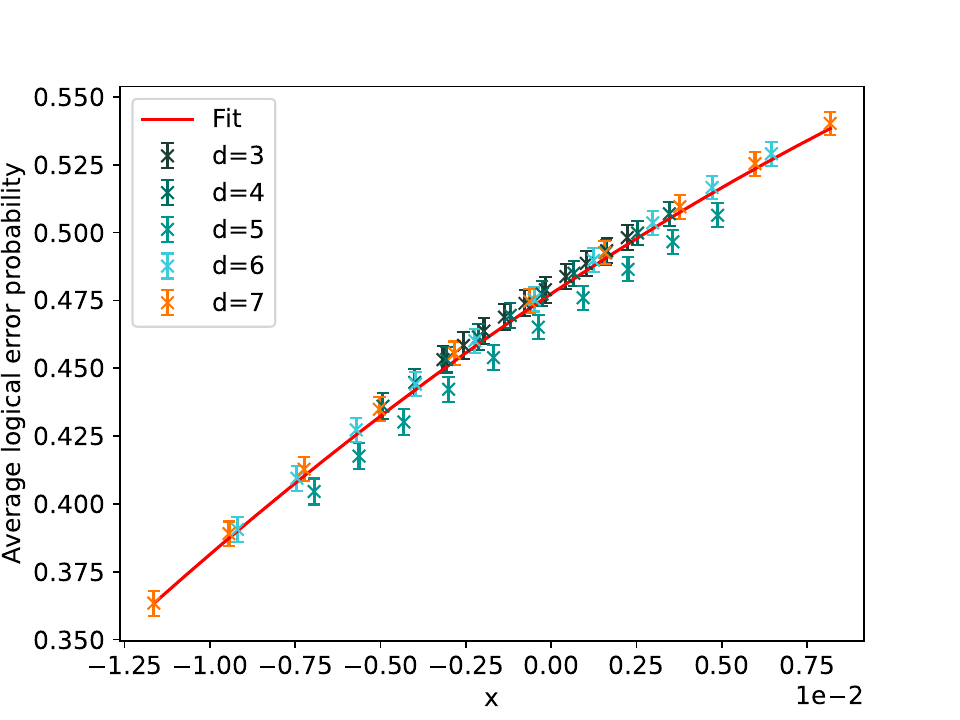}
    \caption{}
\end{subfigure}\\
\begin{subfigure}{0.24\textwidth}
    \includegraphics[width=\textwidth]{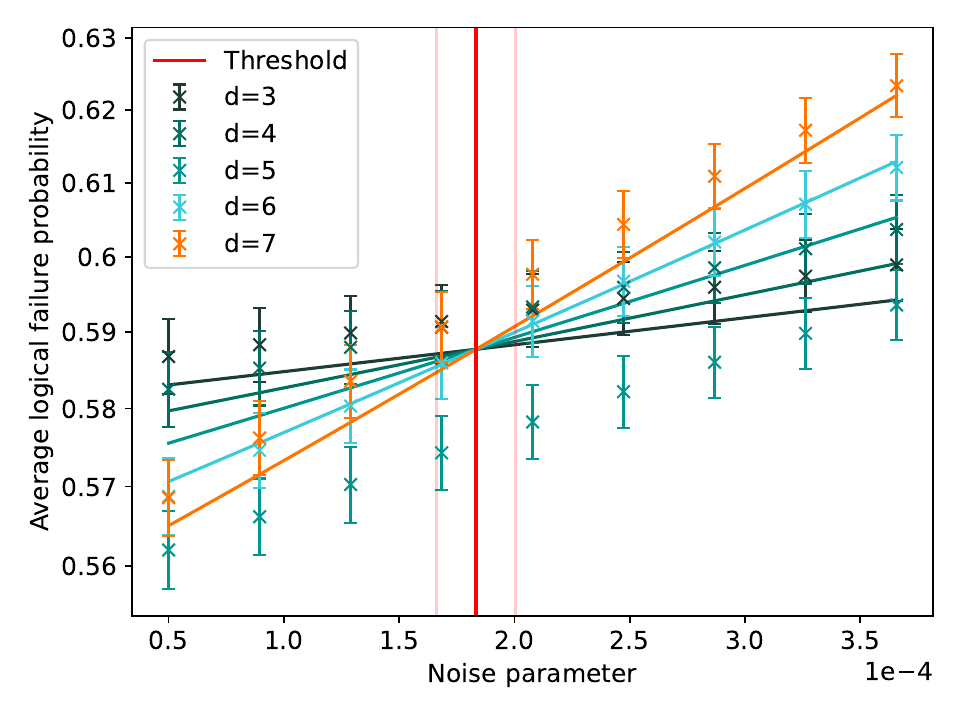}
    \caption{}
\end{subfigure}\hspace{20pt}
\begin{subfigure}{0.26\textwidth}
    \includegraphics[width=\textwidth]{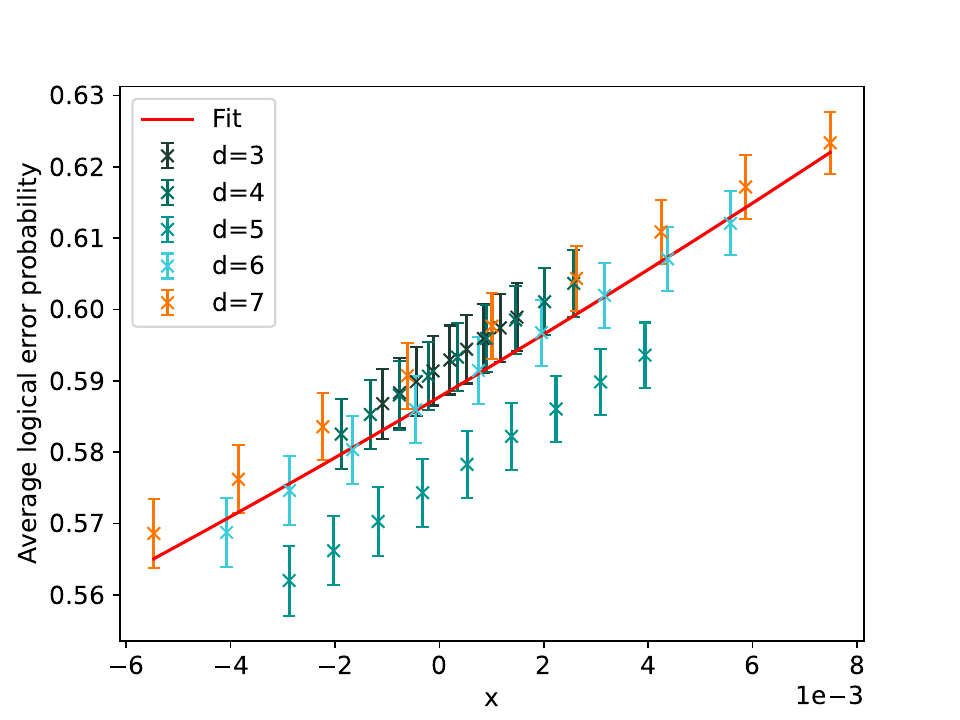}
    \caption{}
\end{subfigure}
\caption{Close-to-threshold data of average logical failure probability for fabrication defect probabilities: (a)-(b) 0\%, (c)-(d) 1\%, (e)-(f) 2\%, (g)-(h) 4\%, (i)-(j) 8\%, (k)-(l) 10\% and (m)-(n) 12\%. Thresholds are indicated with red vertical lines, with light red lines indicating the standard error of the fitted threshold.}
\label{fig:close_threshold_data}
\end{figure}

In this Appendix we present additional data from which we extracted the Pauli thresholds presented in \cref{fig:pauli-thresholds}. In \cref{fig:close_threshold_data}, we display close-to-threshold average logical failure probabilities, averaged over 1,000 realisations with defects randomly sampled, for $p_\text{defect}>0\%$, and 100,000 error correction shots for each realisation. We display both the raw data and the scaling collapse of the data. For the latter, we plot $x = (p_{\textrm{Pauli}}-p_0) d^{1/\nu_0}$ against the (average) logical failure probability, where $p_0$ and $\nu_0$ are the threshold and critical exponent taken from a best fit. We also plot a quadratic best fit for each of these scaling collapse plots~\cite{Harrington2004Thesis}. The scaling collapse plots indicate the presence of some finite-size effects---evident from systematic deviations from the best fit lines---which become particularly noticeable close to the percolation threshold of $p_\text{defect} \approx 13\%$. As such, the standard errors on the thresholds indicated may under-estimate the true error, owing to our inability to simulate larger code distances, since the error incorporates the errors on the data points but not a systematic error due to limited applicability of the scaling assumptions. 

However, it is clear that thresholds exist up to at least $p_\text{defect} = 10\% $ (see \cref{subfig:10_perc_noise}), where a monotonic decrease in average logical failure probability with distance is seen at low $p_{\textrm{Pauli}}$. For $p_\text{defect} = 12\%$, we do not see such a clean demonstration of sub-threshold behaviour: larger system sizes and/or more realisations would likely be required to clearly demonstrate the existence of a threshold. However, we do see some increase in the code performance with target distance.

\section{Strategy for removing defective connections}
\label{app:defective-connections}

We shall now discuss how to adapt a Floquet code when the qubits themselves are functioning properly but there are defective connections between pairs of qubits. As before, we assume that these defective connections are identified in advance through calibration of the device.

In the surface code, a defective connection between a data qubit and a measurement qubit can be removed by treating the data qubit as defective \cite{Auger2017FabricationSurfaceCode}. One can apply the same technique to Floquet codes (see \cref{ssec:defective-connections}), but improvements on this strategy are possible.
With an improved strategy, one can remove defective connections from a Floquet code with no additional connections, no schedule modifications, and without removing \textit{any} data qubits. We will focus on removing a single defective connection in this example, but note that as before this approach can be used iteratively to accommodate multiple defective connections. We also note that this approach can be used in combination with the strategy in \cref{sec:strategy} to correct for a combination of defective qubits and connections. Finally, while this example will focus on the honeycomb code, we note that it can also be applied to Floquet codes defined on other tri-valent and 3-face-colourable lattices.

\begin{figure}
    \centering
    \begin{subfigure}{0.3\linewidth}
        \includegraphics[width=\linewidth]{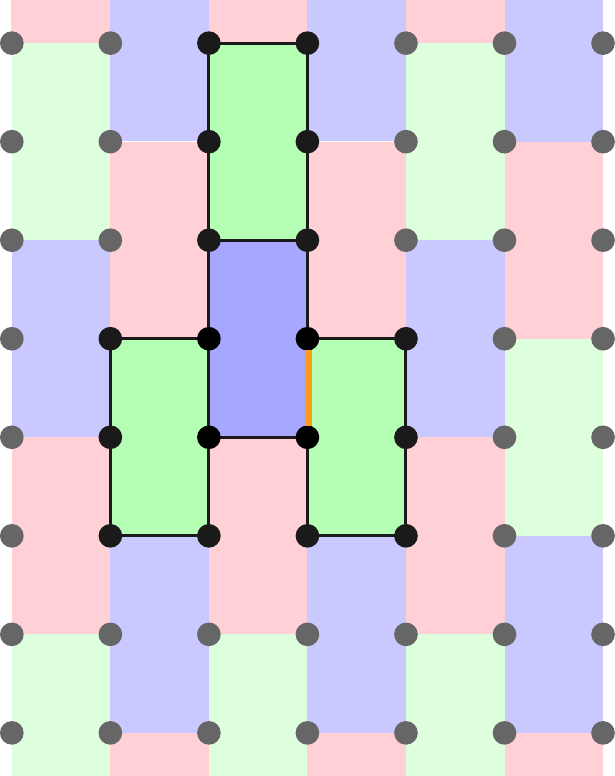}
        \caption{\label{fig:edge-defect-initial}}
    \end{subfigure}
    \begin{subfigure}{0.3\linewidth}
        \includegraphics[width=\linewidth]{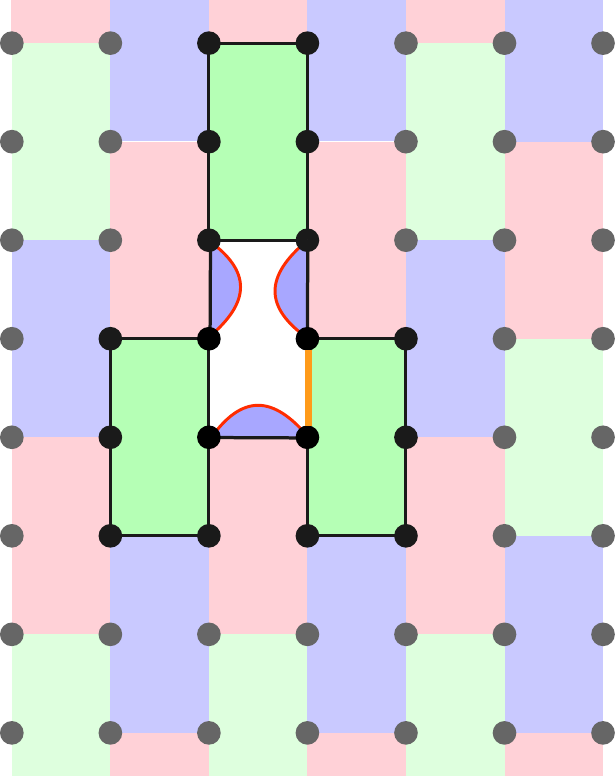}
        \caption{\label{fig:edge-defect-split}}
    \end{subfigure}
    \begin{subfigure}{0.3\linewidth}
        \includegraphics[width=\linewidth]{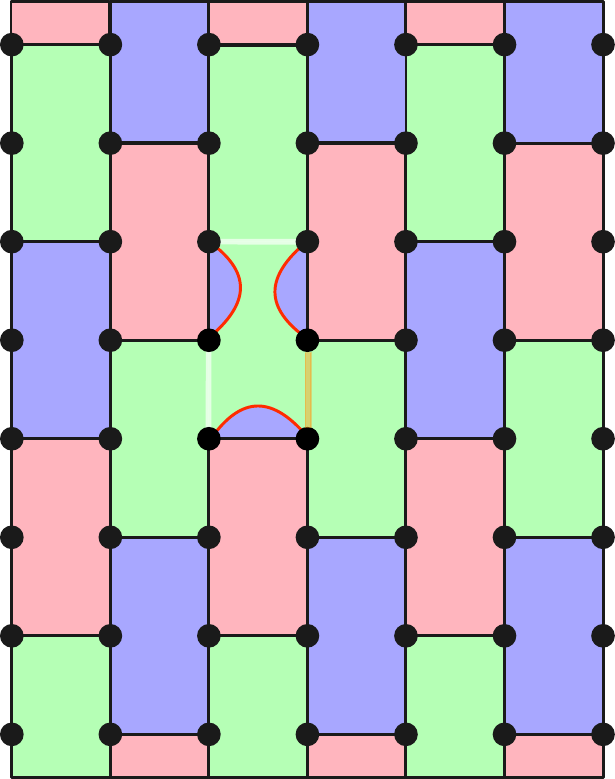}
        \caption{\label{fig:edge-defect-merge}}
    \end{subfigure}
    \caption{\label{fig:edge-defect} Strategy for removing defective connections from the bulk of the honeycomb code. (a) A defective connection between two otherwise functioning qubits is shown in orange. Merge and split plaquettes around the defective connection are highlighted. In this example there is one split plaquette in blue, and three merge plaquettes in green. (b) The shrink plaquette is replaced with three weight-2 plaquettes, formed using newly-added edges highlighted in red. (c) The merge plaquettes are combined into a single super-stabiliser. Newly added edges are highlighted in red, while functioning edges which have been removed are shown in white. The defective connection, now light orange, no longer contributes to any stabilisers and can be removed.}
\end{figure}

Our strategy for accommodating defective connections is shown in \cref{fig:edge-defect}. Initially, the defective connection is shown in orange in \cref{fig:edge-defect-initial}. This edge connects two red plaquettes, and is incident to one green plaquette and one blue plaquette.

The main difference between this strategy and the strategy in \cref{ssec:superplaquette-algorithm} is the choice of merge and split plaquettes. Using the algorithm in \cref{ssec:superplaquette-algorithm}, we would choose to split the blue and green plaquettes incident to the defect edge and merge surrounding red plaquettes. Instead, we choose one plaquette incident to the defective edge to be the split plaquette. The other plaquette incident to the defect edge is chosen to be a merge plaquette, with the other merge plaquettes being the plaquettes of the same type which are adjacent to the split plaquette. In \cref{fig:edge-defect-initial}, we have chosen to split the blue plaquette incident to the defective connection and merge the three green plaquettes adjacent to this blue plaquette.

Having chosen our merge and split plaquettes, we then proceed using Steps 4-7 of the algorithm in \cref{ssec:superplaquette-algorithm}. First, in \cref{fig:edge-defect-split}, we split the blue plaquette into weight-2 plaquettes by introducing additional red edges, akin to Steps 4 and 5. Then, in \cref{fig:edge-defect-merge} we use the newly-added red edges to merge the green plaquettes into a super-plaquette, using the same process as in Steps 6 and 7. The theorem below is equivalent to \Cref{thm:validity-defective-connections}, but looking at standard colour code lattices rather than the ``heavy'' lattices considered in \cref{ssec:defective-connections}. This difference does not affect the proof.

\setcounter{theorem}{2}
\begin{theorem}[for defective connections]
    The above process produces a tri-valent 3-face-colourable lattice.
\end{theorem}

\begin{proof}
    Our proof will follow a similar structure to the proof for \Cref{thm:validity-qubit-removal}. We will prove this for removing a single defective edge, whose basis is denoted as $D$.
    
    Note that edges removed from the lattice are edges which contribute to the shrink plaquette and whose basis is $D$. The shrink plaquette consists of a cycle of edges whose bases alternate between $D$ and one other basis $T$. This means that each vertex in the shrink plaquette is incident to one edge which is removed from the lattice. However, each vertex in the shrink plaquette is \emph{also} incident to a newly-added edge of type $D$, therefore vertices in the shrink plaquette are still of degree three. Vertices wich are not in the shrink plaquette are not affected by this, therefore tri-valence is still preserved.

    3-face-colourability follows the same proof as in \Cref{thm:validity-qubit-removal}. The newly-added weight-2 plaquettes are assigned the same colour as the shrink plaquette, and the super-plaquette is assigned the same colour as the merge plaquette. Because the merge plaquettes are adjacent to the shrink plaquette, they must be different colours. The weight-2 plaquettes are adjacent to the super-plaquette and one plaquette which was adjacent to the shrink plaquette in the original graph, both of which must be coloured differently because of the remarks above. The super-plaquette on the other hand is adjacent to weight-2 plaquettes---which are of the same colour as the shrink plaquette---and plaquettes which were incident to a merge plaquette in the original graph, and must therefore also be a different colour. This means that the newly-added plaquettes are guaranteed to not share the same colour as any adjacent plaquette, thus preserving 3-face-colourability.
\end{proof}

Note that unlike the strategy for qubit removal, the total number of vertices, edges and plaquettes remains unchanged compared to the original lattice when removing a defective edge. We can also see that our scheme is optimal in the number of edges removed from the lattice, again considering honeycomb, 4.8.8, or 4.6.12 lattices. This is a rephrasing of \Cref{thm:minimal-auxiliary-qubit-removal}, with a focus on the standard honeycomb, 4.8.8 and 4.6.12 lattices rather than their heavy equivalents. This modification does not affect the proof.

\begin{theorem}[for defective connections]
    Consider a defective edge in the bulk of a Floquet code defined on a honeycomb, 4.8.8 or 4.6.12 lattice (without boundary); i.e., a uniform tiling of the plane. Under the constraint that we cannot form an edge between any two qubits not connected by an edge in the original lattice, the algorithm described above removes the defective edge with the minimal number of removed edges from the original lattice.
    \label{thm:minimal-edge-removal-defective-edges}
\end{theorem}
\begin{proof}
Our proof will follow a similar form to that of \Cref{thm:minimal-edge-removal}. The defective edge $E$ is incident to two vertices $v_1, v_2$. Removing $E$ will result in $v_1$ and $v_2$ both being of degree-2. In order to maintain tri-valence we must introduce additional edges incident to $v_1$ and $v_2$. Again, we assume the edges in the original lattice, with the exception of the defective edge, account for all possible available connections in hardware, and any new edges must lie along the same connections.

Without loss of generality, let $E$ be coloured red, as in \cref{fig:edge-defect}. 
With the defective edge removed, $v_1$ and $v_2$ are not incident to any red edge, meaning that red edges must be added to the graph. To maintain tri-valence and three-face-colourability, we must do this by drawing a path of edges between $v_1$ and $v_2$ which (a) alternates between red and non-red edges, (b) begins and ends with non-red edges, and (c) does not traverse $E$. Using this path, a new tri-valent graph can be constructed by removing red edges along the path and introducing new red edges parallel to the non-red edges along the path. The number of edges removed from the original lattice will therefore be $(n-2)/2$, where $n$ is the number of vertices in the path. Therefore choosing the shortest path will minimise the number of edges removed.

In the approach described above, we construct this path out of the shrink plaquette. $v_1$ and $v_2$ are incident to an edge which is part of the shrink plaquette, meaning that the length of the path will be the weight of the shrink plaquette. Therefore the number of removed edges along the path will be $(K-2)/2$, where $K$ is the weight of the shrink plaquette. This leads to 2 additional edges being removed in the honeycomb lattice, as all plaquettes are of weight 6. This path has length 5. A shorter path would have to have length 3, implying the existence of a length-4 cycle in the original graph (since $v_1$ and $v_2$ are connected by an edge in the original graph). This does not exist in the honeycomb lattice. In the 4.8.8 lattice, the number of additional edges removed will either be 1 if the defective edge is part of a weight-4 plaquette, or 3 otherwise. In the latter case, there is no shorter path between $v_1$ and $v_2$ other than along an arbitrarily chosen shrink plaquette (both having weight 8). Finally, for the 4.6.12 lattice, the number of additional edges removed will either by 1 if the defective edge is part of a weight-4 plaquette, or 2 if the edge is incident to a weight-6 plaquette\footnote{Note that the structure of the 4.6.12 lattice ensures that any edge is incident to either a weight-4 or weight-6 plaquette, so we do not have to consider choosing a weight-12 plaquette as the shrink plaquette.}. In this case as well, it is obvious that a shorter path which connects $v_1$ and $v_2$ while meeting the requirements above cannot exist.
\end{proof}

It can also be easily seen that this optimality is not constrained to planar surfaces. Below, we show a simple remark found by applying this proof to the 8.8.8 lattice used in hyperbolic Floquet codes \cite{Higgott2023HyperbolicFloquetCodes, Fahimniya2023Hyperbolic}:

\begin{remark}
    Consider a defective edge in the bulk of a Floquet code defined on an 8.8.8 lattice (without boundary), which is a particular uniform tiling of a hyperbolic surface. Under the constraint that we cannot form an edge between any two qubits not connected by an edge in the original lattice, the algorithm described above removes the defective edge with the minimal number of removed edges from the original lattice.
\end{remark}
\begin{proof}
Our proof will follow from the proof of \Cref{thm:minimal-edge-removal-defective-edges}, noting that the proof does not rely on the lattice being a plane. In the 8.8.8 lattice, the shrink plaquette will always be of weight-8, meaning that there will be 8 vertices in the path connecting the vertices incident to the defective edge, and a total of 3 edges are removed from the lattice. It is impossible to construct a shorter path in the 8.8.8 lattice satisfying the requirements stated in \Cref{thm:minimal-edge-removal-defective-edges}, therefore this edge removal is optimal.
\end{proof}

Note that we do not claim optimality for the scenario for non-trivial combinations of defective connections. As with \Cref{thm:minimal-edge-removal}, an optimal edge removal in the case of multiple defective edges would depend on the choice of shrink and merge plaquettes. In particular, if the shrink plaquette is of large weight then there might be a shorter path between the two vertices incident to the defective edge. In general, for $n$ defective connections the number of possible sets of shrink plaquettes grows as $O(2^{n})$, so it is likely that one will need to rely on heuristics to make suitable choices.

If a defective connection is incident to a qubit on the boundary, we can apply the same strategy as in \cref{ssec:boundary-defects} of embedding the patch into a larger Floquet code, applying the connection removal strategy, and then cutting the boundary out again. We give examples of applying this method to defective connections on the horizontal and vertical boundaries in \cref{fig:defect-edge-boundary}. Note however that some care must be taken when choosing merge and split plaquettes. For example, in \cref{fig:defect-edge-boundary-horizontal-bad}, a plaquette on the boundary has been selected as the split plaquette. This has the unintended consequence of creating two weight-one corner plaquettes. These plaquettes are inherently non-deterministic throughout the whole computation, potentially reducing the code distance and logical fidelity, and should therefore be avoided.

\begin{figure}
    \centering
    \begin{subfigure}{0.3\linewidth}
        \includegraphics[width=\linewidth]{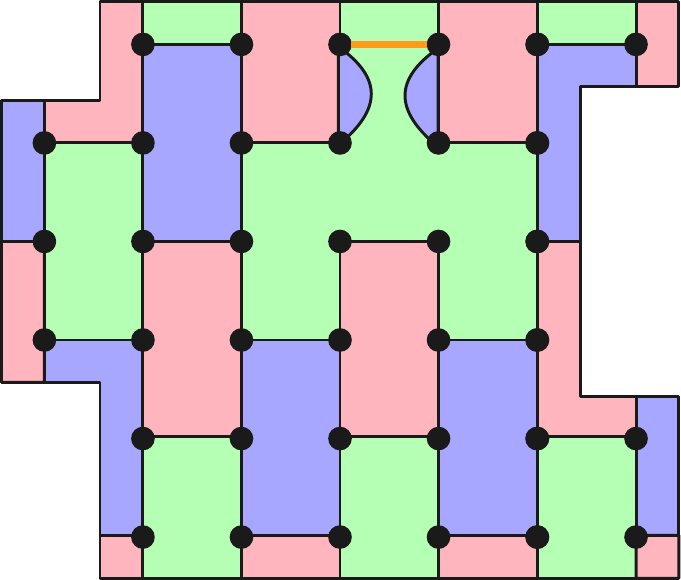}
        \caption{\label{fig:defect-edge-boundary-horizontal-good}}
    \end{subfigure}
    \begin{subfigure}{0.3\linewidth}
        \includegraphics[width=\linewidth]{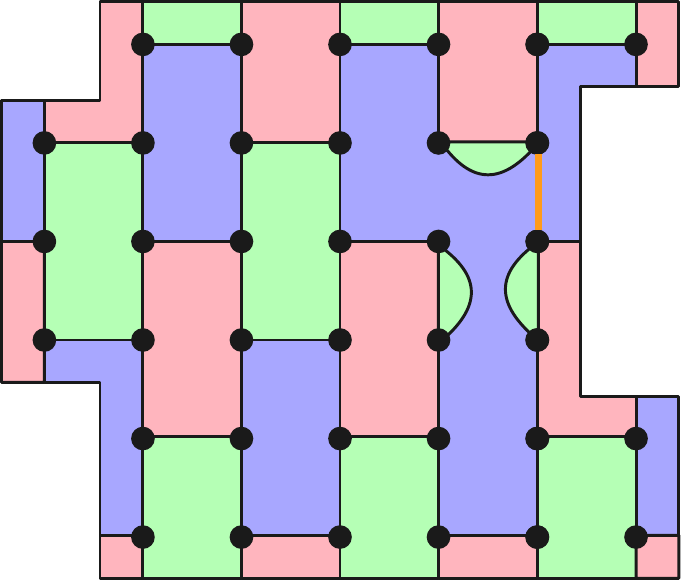}
        \caption{\label{fig:defect-edge-boundary-vertical}}
    \end{subfigure}
    \begin{subfigure}{0.3\linewidth}
        \includegraphics[width=\linewidth]{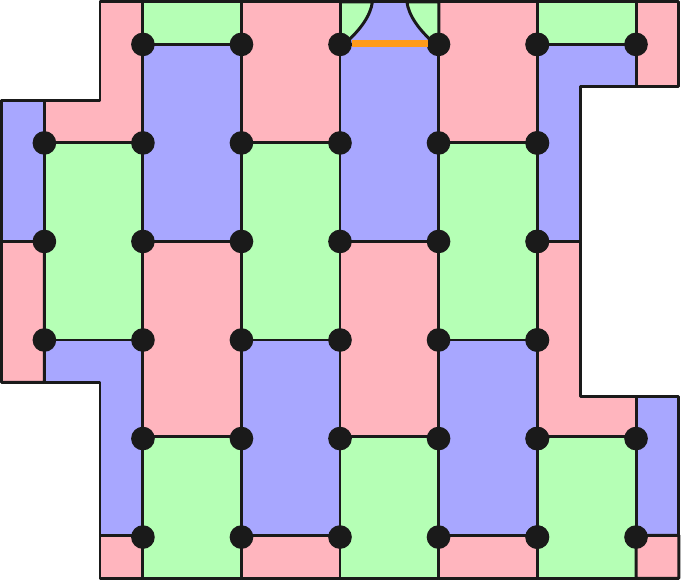}
        \caption{\label{fig:defect-edge-boundary-horizontal-bad}}
    \end{subfigure}
    \caption{\label{fig:defect-edge-boundary} Removing defective connections on the (a) horizontal and (b) vertical boundaries of the honeycomb code. (c) A plaquette on the boundary is chosen as the split plaquette, which has the unintended effect of introducing two green weight-1 plaquettes. These weight-1 plaquettes will be nondeterministic throughout the whole experiment, thus choosing to split the green boundary plaquette is not possible in this case.}
\end{figure}

\end{document}